\documentclass[review,onefignum,onetabnum]{siamonline190516}
\nolinenumbers

\usepackage[utf8]{inputenc}

\usepackage{amsmath,amsfonts, amssymb}
\usepackage{comment}
\usepackage{enumerate}
\usepackage{graphicx,placeins,multirow}
\usepackage[font=small,labelfont=small]{caption}
\usepackage[font=footnotesize,labelfont=footnotesize]{subcaption}
\usepackage{tikz}
\usepackage{xcolor,xspace,hyperref}

\DeclareMathOperator*{\argmax}{arg\,max}
\DeclareMathOperator*{\argmin}{arg\,min}
\newcommand{\D}{\nabla}
\newcommand{\dd}{{\sf d}}
\newcommand{\dmax}{d_{\rm max}}
\newcommand{\dmin}{d_{\rm min}}
\newcommand{\dbar}{d_{\rm avg}}
\newcommand{\vbar}{\bar{v}}

\DeclareMathOperator*{\Ex}{\mathbb{E}}
\newcommand{\G}{\mathcal{G}}
\newcommand{\R}{\mathbb{R}}
\newcommand{\Lap}{\mathcal{L}}

\newcommand{\I}{{\tt I}\xspace}
\newcommand{\MH}{M}
\newcommand{\one}{{\sf 1}}
\newcommand{\Proj}{\Pi}
\newcommand{\sdot}{\! \cdot \!}
\newcommand{\T}{\top}
\newcommand{\Tor}{\mathbb{T}^{m}}
\newcommand{\Torone}{\mathbb{T}^{1}}
\newcommand{\U}{{\sf u}}
\newcommand{\V}{{\sf v}}
\newcommand{\W}{{\sf w}}
\newcommand{\X}{{\sf x}}
\newcommand{\Y}{{\sf y}}
\newcommand{\zero}{{\sf 0}}

\definecolor{applegreen}{rgb}{0.55, 0.71, 0.0}
\definecolor{alizarin}{rgb}{0.82, 0.1, 0.26}
\definecolor{slategray}{rgb}{0.44, 0.5, 0.56}
\definecolor{amber}{rgb}{1.0, 0.75, 0.0}
\definecolor{mikadoyellow}{rgb}{1.0, 0.77, 0.05}
\definecolor{cadmiumgreen}{rgb}{0.0, 0.42, 0.24}
\definecolor{forestgreen}{rgb}{0.13, 0.55, 0.13}
\definecolor{lust}{rgb}{0.9, 0.13, 0.13}
\definecolor{denim}{rgb}{0.08, 0.38, 0.74}
\definecolor{purpleheart}{rgb}{0.41, 0.21, 0.61}
\definecolor{cherryblossompink}{rgb}{1.0, 0.72, 0.77}
\definecolor{darktangerine}{rgb}{1.0, 0.66, 0.07}
\definecolor{bananayellow}{rgb}{1.0, 0.88, 0.21}
\definecolor{lightblue}{rgb}{0.55,0.82,0.77}

\let\emptyset\varnothing

\allowdisplaybreaks

\begin{document}

\title{Measuring Segregation via Analysis on Graphs}
\author{Moon Duchin\thanks{Department of Mathematics, Tufts University, Medford, MA 02155, USA \email{(moon.duchin@tufts.edu,jm.murphy@tufts.edu})}
	\and James M. Murphy\footnotemark[1]
	\and Thomas Weighill\thanks{Department of Mathematics, UNC Greensboro, NC 27402, USA (\email{t\_weighill@uncg.edu})}
        }
\maketitle

\begin{abstract}In this paper, we use analysis on graphs to study quantitative measures of segregation.  We focus on a classical statistic from the geography and urban sociology literature known as  {\em Moran's $\I$}, which in our language is a score associated to a real-valued function on a graph, computed with respect to a spatial weight matrix such as the adjacency matrix associated to the geographic units that tile a city.  Our results characterizing the extremal behavior of $\I$ illustrate the important role of the  underlying graph structure, especially the degree distribution, in interpreting the score. 
In addition to the standard spatial weight matrices encoding unit adjacency, we consider  the Laplacian $L$ and a doubly-stochastic approximation $\MH$.  These alternatives allow us to connect $\I$ to ideas from Fourier analysis and random walks.
We offer illustrations of our theoretical results with a mix of stylized synthetic examples and real geographic/demographic data.
\end{abstract}

% REQUIRED

% REQUIRED

%%%
\section{Introduction}

A central question for geographers, urban sociologists, and demographers is to identify and measure levels of spatial correlation for a social statistic that is associated to geographic units.  When the topic is the human geography of a population subgroup, the presence of strong  correlation between number and place goes by the general name of {\em segregation}.

In recent decades, researchers have made increasing use of network structure to model the relationship between the geographical units that make up an area under study.  The nodes might stand for individual people or for geographical units like census blocks or counties.  Network topology can be given by simple adjacency of units, or by proximity (placing edges between units that are within a threshold  distance apart).

Figure~\ref{fig:Intro} shows the basic motivating example: a square lattice graph is first decorated with a checkerboard pattern and then with a clustered pattern.  The central question under consideration in this paper is the design of a numerical indicator that detects the intermixing of types on the checkerboard, in contrast to the separation of types on the clustered grid---and that lets us know, on the other hand, when no pattern is present at all.

\begin{figure}[ht]
\centering
\begin{tikzpicture}
\node at (-4,0) {\includegraphics[width=2.8in]{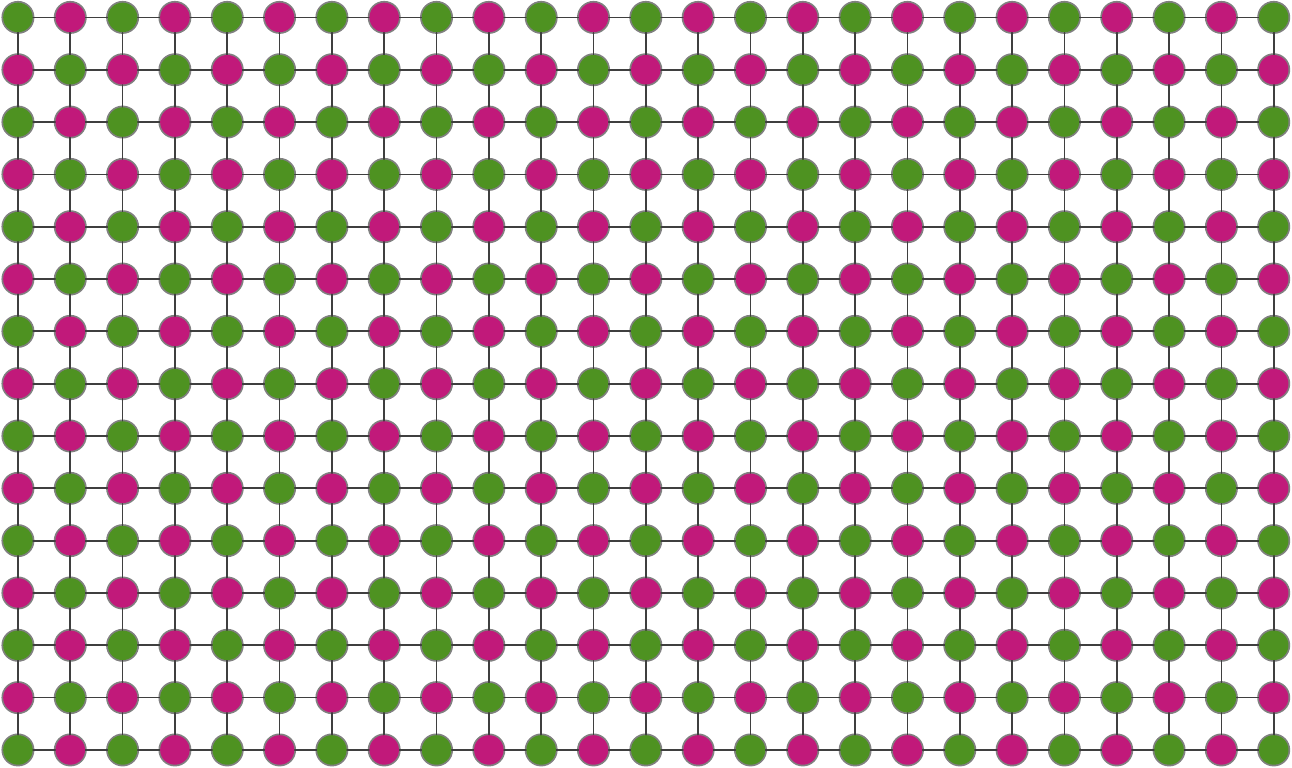}};
\node at (-4,-2.7) {Checkerboard: $\I=-1$};
\node at (4,0) {\includegraphics[width=2.8in]{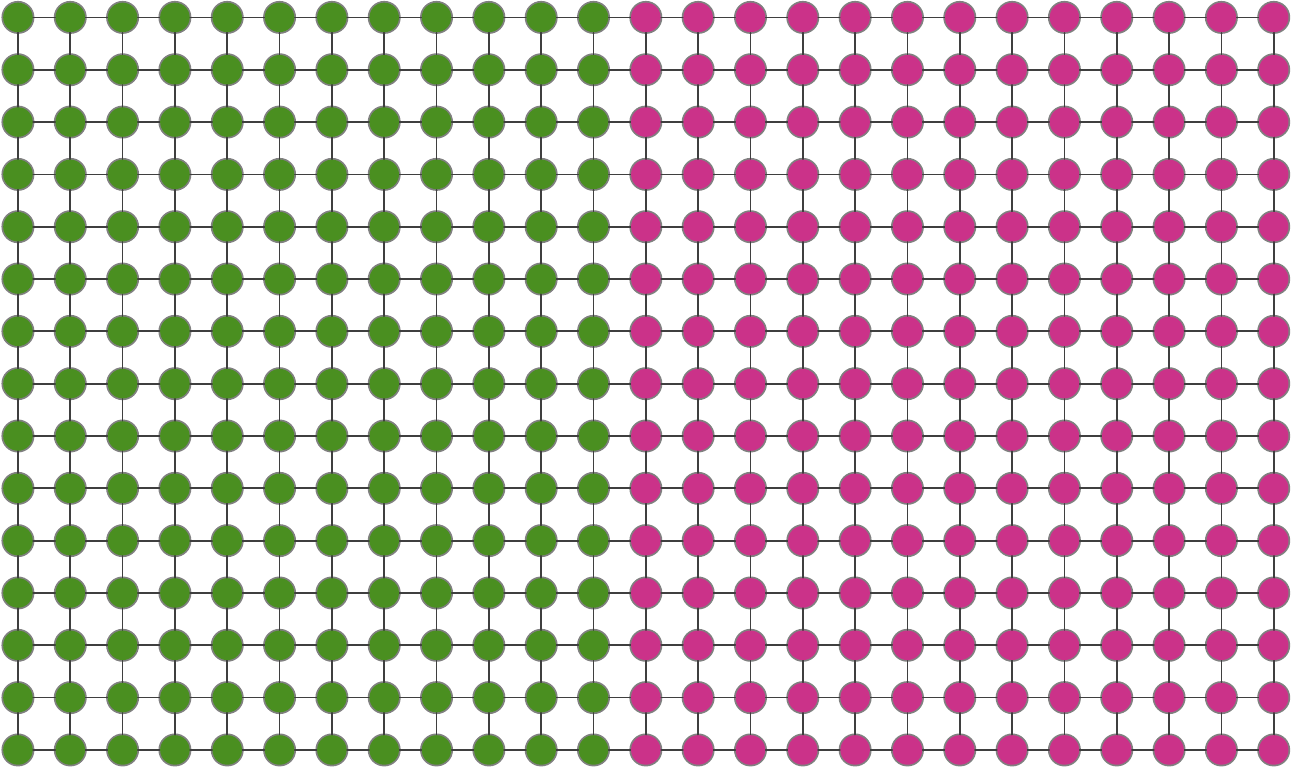}};
\node at (4,-2.7) {Clusters: $\I=0.9576$};
\end{tikzpicture}

    \caption{In these images, green and purple represent two different numerical values, say 0 and 1.  On the left, they are arranged in a checkerboard pattern; on the right, a clustered pattern. Moran's \I returns a low value of $-1$ for the checkerboard and a high value of nearly $1$ for the clusters.}\label{fig:Intro}
\end{figure}

While there is no shortage of proposed metrics to quantify segregation (see \S~\ref{subsec:LitReview}), the go-to choice in spatial statistics is Moran's \I, formally defined below in Definition~\ref{defn:I}. Introduced in the mid-20th century by P.A.P. Moran \cite{Moran1950_Notes}, this score is so prominent in the study of spatial structure in numerical data that it is almost synonymous with the concept of {\em spatial auto-correlation}. Over the years, social scientists have developed multiscale generalizations and extensive statistical frameworks  that allow for hypothesis tests in which the null hypothesis is of the form ``this population is not segregated on this network" \cite{deJong1984_Extreme}.  Despite the widespread currency of $\I$ in the field of geography, authors have articulated concerns about the feasibility of reducing a complex social phenomenon such as segregation to a simple score \cite{Legendre1993_Spatial}. Basic questions about how to make comparisons using Moran's $\I$---both to compare populations on a common network and to compare across networks---are wide open.

%%%

\subsection{Summary of Contributions}

Broadly, this paper seeks to describe features and properties of Moran's \I as it is commonly used, and to propose related alternatives that have improved properties.  

First, after introducing notation and definitions (\S\ref{sec:Defs}), we provide a spectral graph theory description of Moran's \I (with respect to a spatial weight matrix $W$) that we then use to derive basic properties and to consider the validity of standard claims pertaining to its use in spatial statistics (\S \ref{sec:SpectralGraphTheoryAndI}).  
For various choices of $W$ closely related to a graph, 
we show that the graph topology (degree distribution, cut lengths) controls the range of attainable values (\S \ref{sec:ComparingI}), which impacts our ability to interpret $\I$ within and especially across localities. 
Then, in \S \ref{sec:Empirical}, we consider alternatives to the standard choices of spatial weight matrix (classically, the adjacency matrix $A$ and its row-standardization $P$).
We particularly focus on the Laplacian $L$ and a doubly-stochastic alternative we call $\MH$, which offer connections to other rich mathematical concepts.
We derive a relationship between $\I(\V;L)$ and Dirichlet energies (\S \ref{sec:GraphFourierAnalysis}) that connects  quantitative ``smoothness" on a graph, from the point of view of harmonic analysis, to qualitative notions of segregation. 
Next, we develop a random-walk interpretation of  $\I(\V;\MH)$ in \S \ref{sec:WalksandI} that offers the version of Moran's \I that seems most promising for bounds and comparisons of any yet proposed.  

Finally, supported by a mix of theoretical and empirical work, we make concrete recommendations for the practical use of network-based segregation metrics in geography (\S \ref{sec:Conclusions}).

%%%
\section{Background}
\label{sec:Defs}

\subsection{Notation and Basic Definitions}

Suppose we have $n$ geographic units indexed by $1,2,\dots,n$ with adjacency data $A\in\R^{n\times n}$.  The matrix $A$ has entries $A_{ij}=1$ if the $i$th and $j$th units are adjacent and $A_{ij}=0$ if not, with the convention that $A_{ii}=0$ for all $i$.  For example, the units may be census tracts, with  $A_{ij}=1$ if tracts $i$ and $j$ have a shared boundary of positive length (but not if they meet at a corner).  Mathematically, $A$ is the (symmetric) adjacency matrix for a graph with nodes corresponding to the geographic units, which we will call the {\em dual graph} $\G$.  We will use $P\in\R^{n\times n}$ to denote the row-standardized adjacency matrix, i.e., the matrix with entries $P_{ij}=A_{ij}/\sum_{k=1}^n A_{ik}$.  By construction, $P$ is row-stochastic and has real eigenvalues because $P=D^{-1}A$ is conjugate to the symmetric matrix $D^{-1/2}AD^{-1/2}$ where $D$ is the diagonal matrix with $D_{ii}=\sum_{k=1}^{n}A_{ik}$. It achieves a largest eigenvalue of $1$ and has all eigenvalues greater than or equal to $-1$, which is achieved iff the graph is bipartite \cite{Chung1997_Spectral}.  Examples using real census data are shown in Figure~\ref{fig:GA-duals}; the supplemental material includes a histogram of vertex degrees.

\begin{figure}[ht]
    \centering
\begin{tikzpicture}
\node at (-5,0)    {\includegraphics[width=.32\textwidth]{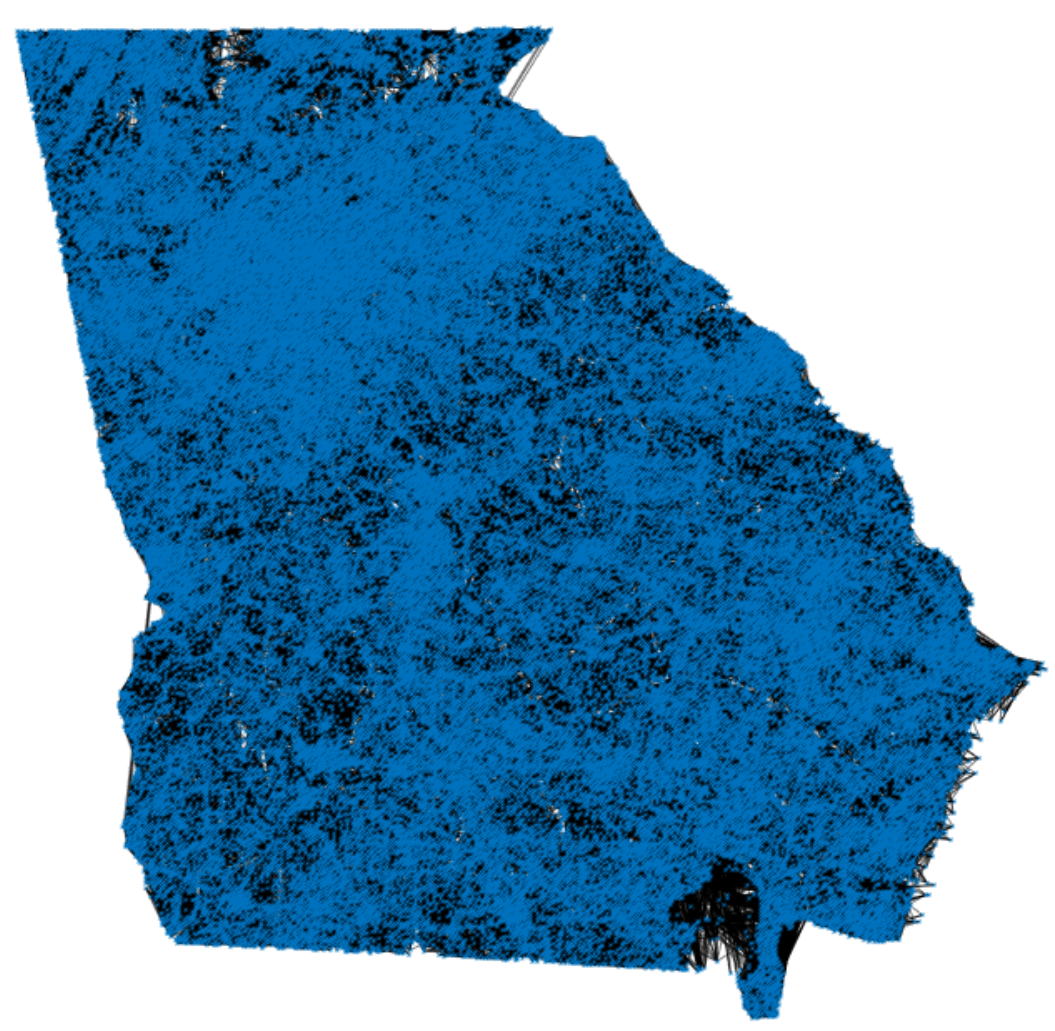}}; 
\node at (0,0)    {\includegraphics[width=.32\textwidth]{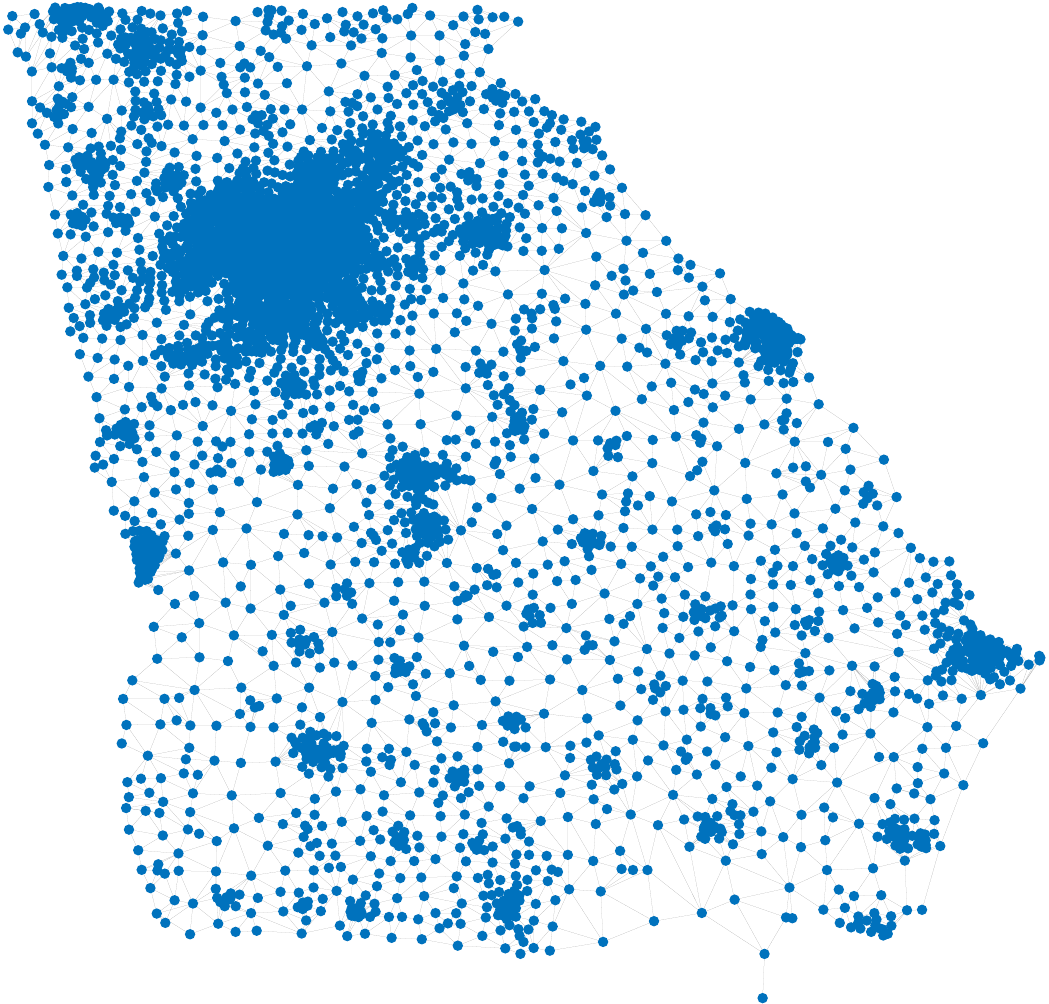}};
\node at (5,0)    {\includegraphics[width=.32\textwidth]{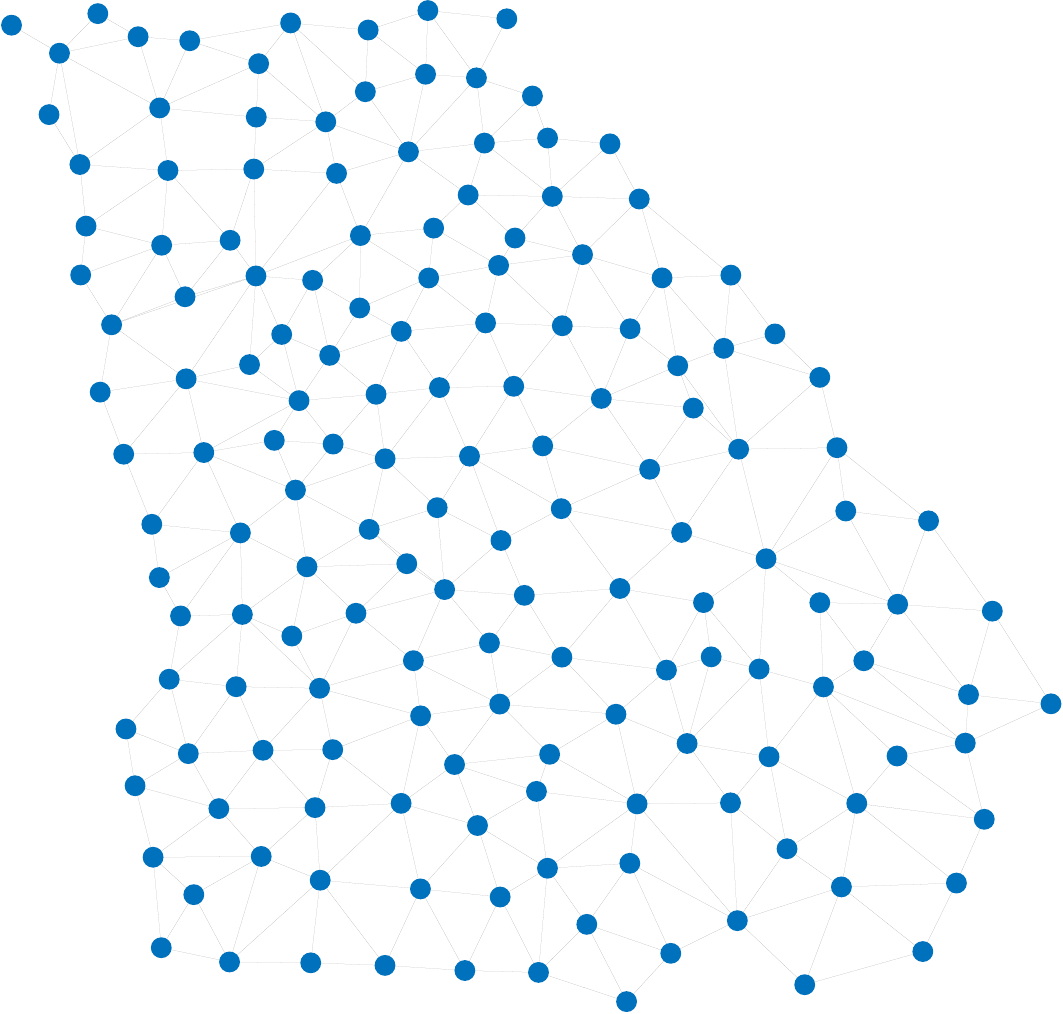}};
\end{tikzpicture}    

    \caption{The dual graphs of the partition of Georgia into 
    $291,086$ census blocks (left), 5533 block groups (middle), and 159 counties (right). 
    The Census Bureau's geographical hierarchy is nested:  each of these can be seen as a quotient of the previous one, collapsing several smaller units into each larger one.
    All are shown with centroidal embedding.  Features of the underlying graphs (like the degree distribution and spectrum) will be shown to provide bounds on the range of $\I$ values that are achievable.}
    \label{fig:GA-duals}
\end{figure}

Consider a function $\V:V(\G)\rightarrow\R$  on the graph nodes, which we will treat as a column vector $\V=(v_1,v_2,\ldots, v_{n})^{\T}\in\R^{n\times 1}$. For example, $v_i$ may be the percentage of residents in tract $i$ who are identified as belonging to a particular demographic (e.g., Hispanic) by the U.S. Census Bureau.  Figure \ref{fig:Chicago} shows real examples drawn from Chicago, Illinois. Let $\vbar=\frac 1n \sum_{i=1}^{n} v_{i}$ be the average of the entries of $\V$.  We will write $\zero,\one$ for the vectors (of length $n$) whose entries are all zero or all one, respectively.

\begin{definition}[Moran's \I]\label{defn:I}  With notation as above, let $W\in\R^{n\times n}$ be a matrix that is not the zero matrix, and let $w=\sum_{i,j=1}^{n}|W_{ij}|$.  \emph{Moran's $\I$ with respect to $W$} is a functional $\I(\ \cdot \ ;W):\R^{n}\rightarrow \R$ defined by \[\I(\V; W) := \left(n \displaystyle\sum_{i,j=1}^{n} W_{ij} (v_i - \vbar)(v_j -\vbar)\right) \bigg /\left(w \displaystyle\sum_{i=1}^{n}(v_i-\vbar)^2\right) = \frac nw\left( \frac{\X^\T W\X}{\X^\T\X}\right),\]
where $\X=\V-\vbar\one$.
\end{definition}

The most common choices of $W$ in geography are $W=A$ and $W=P$ \cite{Anselin1996_Moran}.  We shall refer to any $n\times n$ matrix that is not identically 0 as a \emph{weight matrix}; when it is associated to geographic features (as $A$ and $P$ are) we will call it a spatial weight matrix.  The usual interpretation in the geography community, dating to Moran's original work,  is that for either standard choice of adjacency weights, or for weights based on geographic distance, $\I(\V; W)$ takes higher positive values when the $v_i$ are ``spatially correlated" in the sense that neighboring/nearby units tend to have similar values \cite{Moran1950_Notes}.  Conversely, the standard understanding is that \I is negative when neighboring units tend to have very different values, and near zero when there is little or no relationship between the values of neighboring units.  The precise notion of spatial correlation is therefore graph-dependent.  In short, large values of $\I$ are taken to indicate segregation and small or negative values of $\I$ indicate a lack of segregation.  This article clarifies this intuition in a precise mathematical sense through the lens of analysis on graphs.

\begin{remark}[Zero-centering and rescaling]\label{rem:proj}It suffices to consider \I on vectors with $\ell^{2}$-norm 1 and mean 0.  
Let $X:=\{\X\in\R^{n} \ | \ \X\sdot\one=0\}=\one^\perp$ be the subspace of vectors with  mean 0.  For an arbitrary $\V\in\R^{n}$ with mean value $\vbar$, we let $\X=\V-\vbar\one$ denote its orthogonal projection onto $X$.
Then we immediately see for any $\V$ and $W$ that  $\I(\V; W)=\I(\X; W)$ and that  $\I(\alpha\V; W)=\I(\V;W)$ and $\I(\V; \alpha W)=\I(\V; W)$ for all scalars $\alpha\neq 0$.  
\end{remark}

We will refer to the $i$th rowsum $d_{i}=\displaystyle\sum_{j=1}^{n}|W_{ij}|$ as the {\em  $W$-degree} of node $i$; when $W=A$, this is the standard degree counting the number of edges ending at node $i$, which we will call either the $A$-degree or simply the vertex degree, to distinguish it from the $W$-weighted versions.  Define $\dbar:=\displaystyle\frac{1}{n}\sum_{i=1}^{n}d_i=w/n$ to be the average $W$-degree of the graph; note that if $W=P$, then $d_i$ is identically 1, so $\dbar=1$. We have 
$$\I(\V; W)= \I(\X; W)=  
\frac{\frac 1w \displaystyle\sum_{i,j=1}^{n} W_{ij} x_{i}x_{j}}{\frac 1n \displaystyle\sum_{i=1}^{n}x_{i}^2}
=\frac{\displaystyle\sum_{i,j=1}^{n} W_{ij} x_{i}x_{j}}{\dbar\displaystyle\sum_{i=1}^{n}x_{i}^2}.$$

\begin{remark}[Pairs versus singletons]
\label{rmk:pairs-v-single}
From the function $\V$, the zero-centered $\X$ records deviations above and below the mean.  The score $\I(\X;W)$ measures the patterns in these deviation values.  From the 
second-to-last expression above, we find an appealing interpretation of \I as {\em comparing the product of values at related nodes ($x_ix_j$) to the  squared values at individual nodes ($x_i^2$).}
The pair average is spatially weighted by the coefficients $W_{ij}$.  
In particular, when $W=A$, the score \I is precisely one-half the ratio of the average product across an edge to the average squared value at a vertex.  When $W=P$, it is a different ratio: the sum of the products across edges versus the sum of squared node values.  These have subtly different properties, as we will see below.
 \end{remark}

\begin{table}[htbp!]
\vspace{0.1in}
\begin{center}
\begin{small}
\begin{tabular}{c|c}
\hline
\textsc{Notation} & \textsc{Definition}  \\ \hline
\hline 
$\G$ & simple, undirected graph with $n$ nodes \\
$\V=(v_{1},v_{2},\dots,v_{n})^{\T}$ & function on graph $\G$, denoted as a column vector\\
$\vbar$ & mean of  $\V$ values\\
$\|\V\|_{2}$ & Euclidean norm of $\V$\\
$\zero, \one$ & vector of all 0s and 1s, respectively\\
$\X=(x_{1},x_{2},\dots,x_{n})^{\T}$ & arbitrary mean-0 vector \\
$X=\one^{\perp}$ & space of all mean-0 vectors\\
$\Proj$ & orthogonal projection onto  $X$\\
$W$ & arbitrary weight matrix, not identically zero\\
$Q$ & arbitrary bistochastic matrix (rows and columns sum to one) \\
$A$ & graph adjacency matrix associated to $\G$ \\
$D$ & diagonal vertex degree matrix associated to $\G$\\
$P$ & row-standardized  adjacency matrix associated to $\G$  \\
$L$ & Laplacian matrix associated to $\G$  \\
$\MH$ & bistochastic Metropolis-Hastings matrix associated to $\G$ \\
$\I(\V;W)$ & Moran's $\I$ applied to vector $\V$ with respect to matrix $W$ \\
$\I(X;W)$ & Range of all possible $\I$ values for weight matrix $W$\\
$\{(\lambda_{i},\Phi_{i})\}_{i=1}^{n}$ & eigenvalues and eigenvectors of  arbitrary  $W$\\
$\{(\mu_{i},\Psi_{i})\}_{i=1}^{n}$ & Laplacian eigenvalues and eigenvectors\\
$\dd=(d_{1},d_{2},\dots,d_{n})$ & vector of $W$-degrees $d_{i}=\sum_{j=1}^{n}|W_{ij}|$\\
$\dmin,\dmax,\dbar$ & minimum, maximum, and average $W$-degree over $i=1,\dots,n$ \\
$\{\alpha_{i}\}_{i=1}^{n}$ & coefficients in some orthonormal basis\\
$\mathcal{E}$ & Dirichlet energy functional on the graph\\
$\Tor=[0,2\pi]^m$ & $m$-dimensional torus with periodic boundaries
\end{tabular}
\end{small}
\end{center}
\vskip -0.1in
\caption{Notation used throughout the paper.}\label{tab:notation} 
\end{table}

\FloatBarrier
\subsection{Moran Scatter Plots and Linear Regression}\label{sec:scatter}
Sociologist Luc Anselin is credited with the fundamental observation that when $W$ is row-stochastic, such as for the row-standardized adjacency matrix $P$,  the definition of Moran's \I can be rearranged so that it contains a regression coefficient \cite{Anselin1996_Moran}.  Let $\U$ be defined coordinatewise as $u_i := \sum_{j=1}^n W_{ij} v_{j}$, and denote the mean of $\U$ by $\bar{u}$.
In the $W=P$ case, this is the average of the function values at the neighbors of $i$; if $W$ encodes proximity along some dimension, then this is a weighted average.  For this reason, $u$ is often called the \emph{(spatially) lagged variable}, by analogy with autocorrelation for time series. Moran's \I then reduces to the slope of a regression of the lagged variable ($u_{i}$) on the original variable ($v_{i}$), as follows:
\begin{align*}
    \I(\V;W)  =  \frac{\displaystyle\sum_{i=1}^{n} (v_i - \vbar) \displaystyle\sum_{j=1}^{n} W_{ij} (v_j - \vbar)} { \displaystyle\sum_{i=1}^{n}(v_i-\vbar)^2} 
     = \frac{\displaystyle\sum_{i=1}^n (v_{i}-\vbar) \cdot (u_i-\vbar)}{\displaystyle\sum_{i=1}^n (v_{i}-\vbar)^2} = \frac{\displaystyle\sum_{i=1}^n (v_{i}-\vbar) \cdot (u_i-\bar{u})}{\displaystyle\sum_{i=1}^n (v_{i}-\vbar)^2},
     \end{align*}
where we use the fact that $\sum_{i,j=1}^n W_{ij} = 1$ and $\sum_{i=1}^n (v_i - \vbar)\vbar = \sum_{i=1}^n (v_i - \vbar)\bar{u}   = 0$. We recognize the final expression as the slope of a regression of $\U$ on $\V$.  

\begin{figure}[ht] \centering
\begin{tikzpicture}
\node at (0,6) {\includegraphics[width=1.9in]{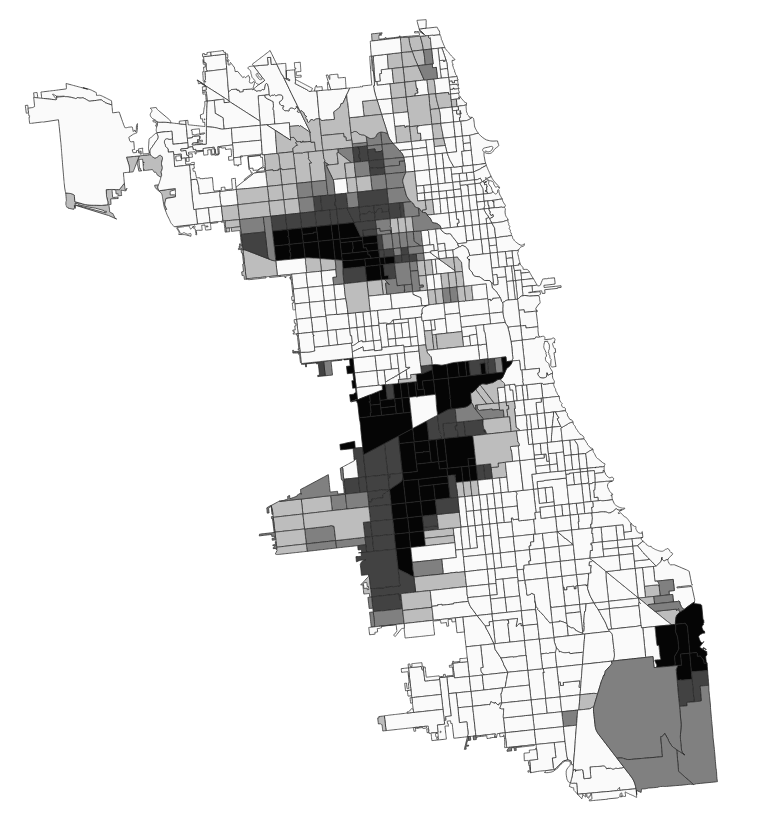}};
\node at (-2.2,6) [rotate=90] {Hispanic};
\node at (5,6) {\includegraphics[width=1.9in]{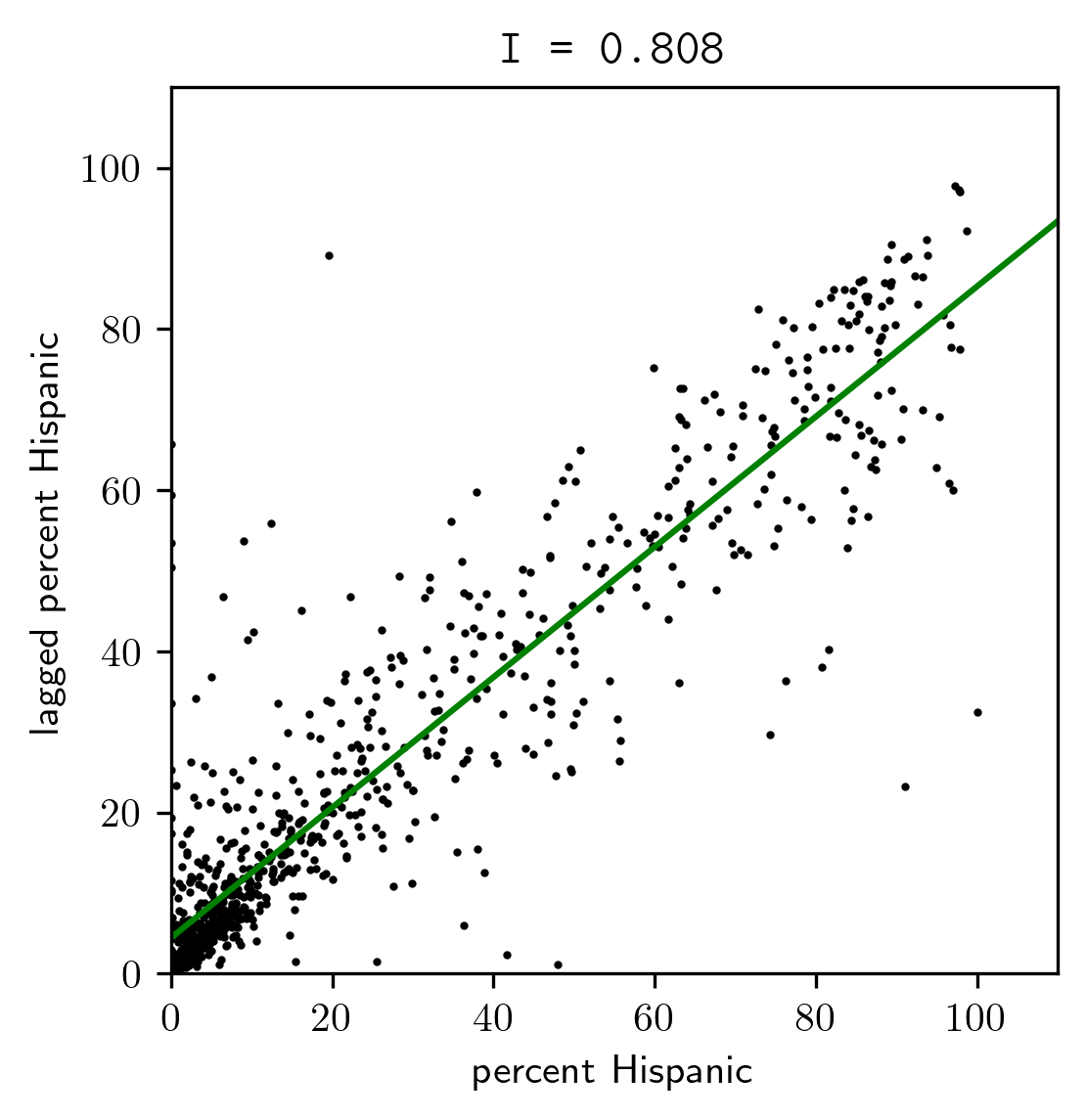}};
\node at (5.3,9.25) {By share};
\node at (10.5,6) {\includegraphics[width=1.9in]{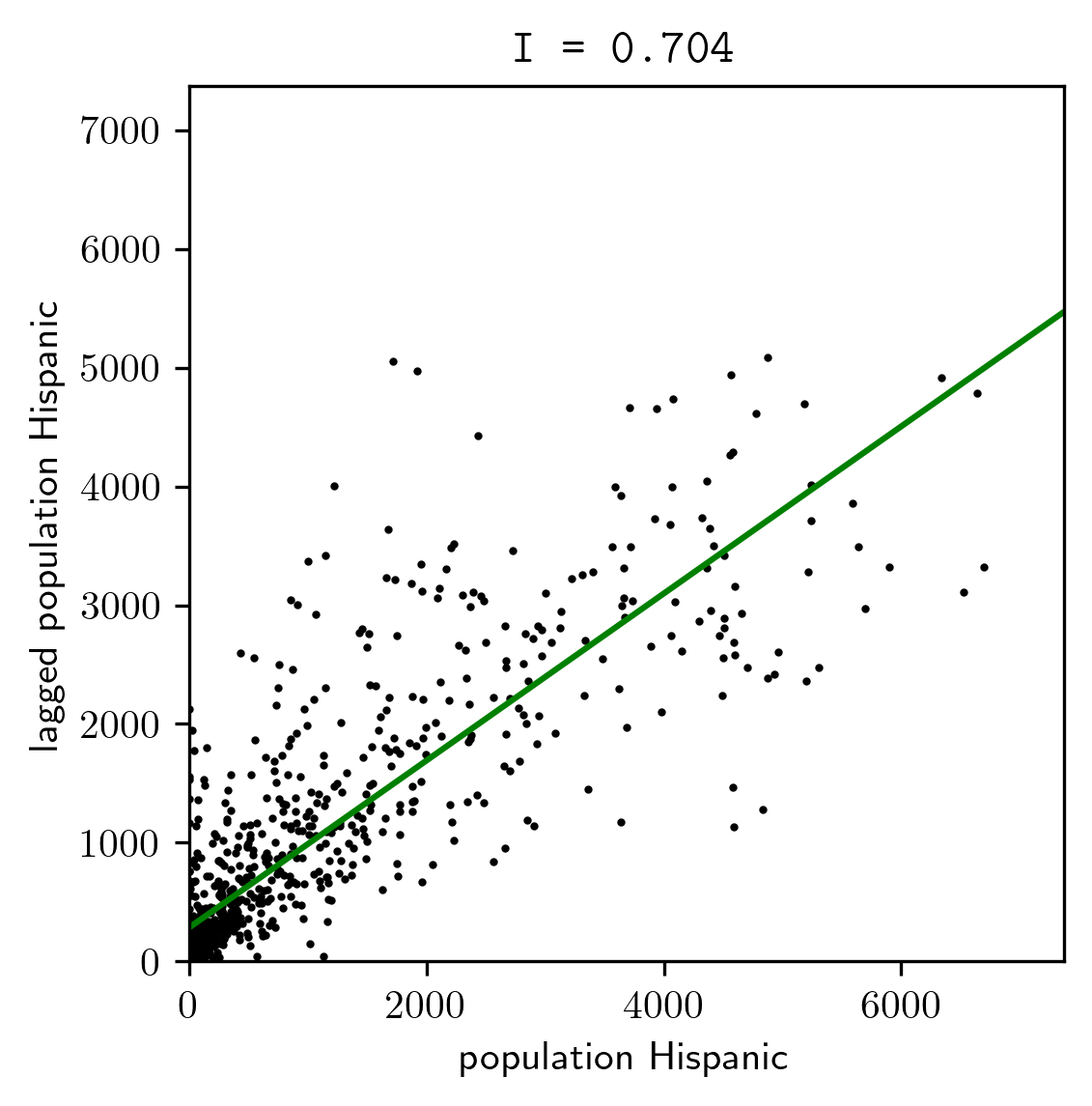}};
\node at (10.8,9.25) {By count};

\node at (0,0) {\includegraphics[width=1.9in]{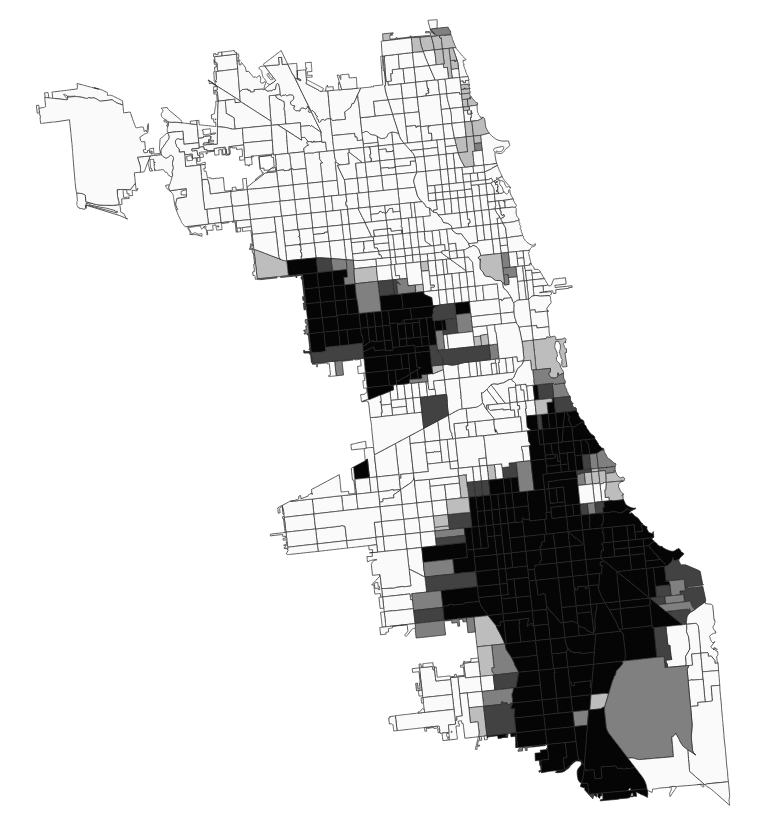}};
\node at (-2.2,0) [rotate=90] {Black};
\node at (5,0) {\includegraphics[width=1.9in]{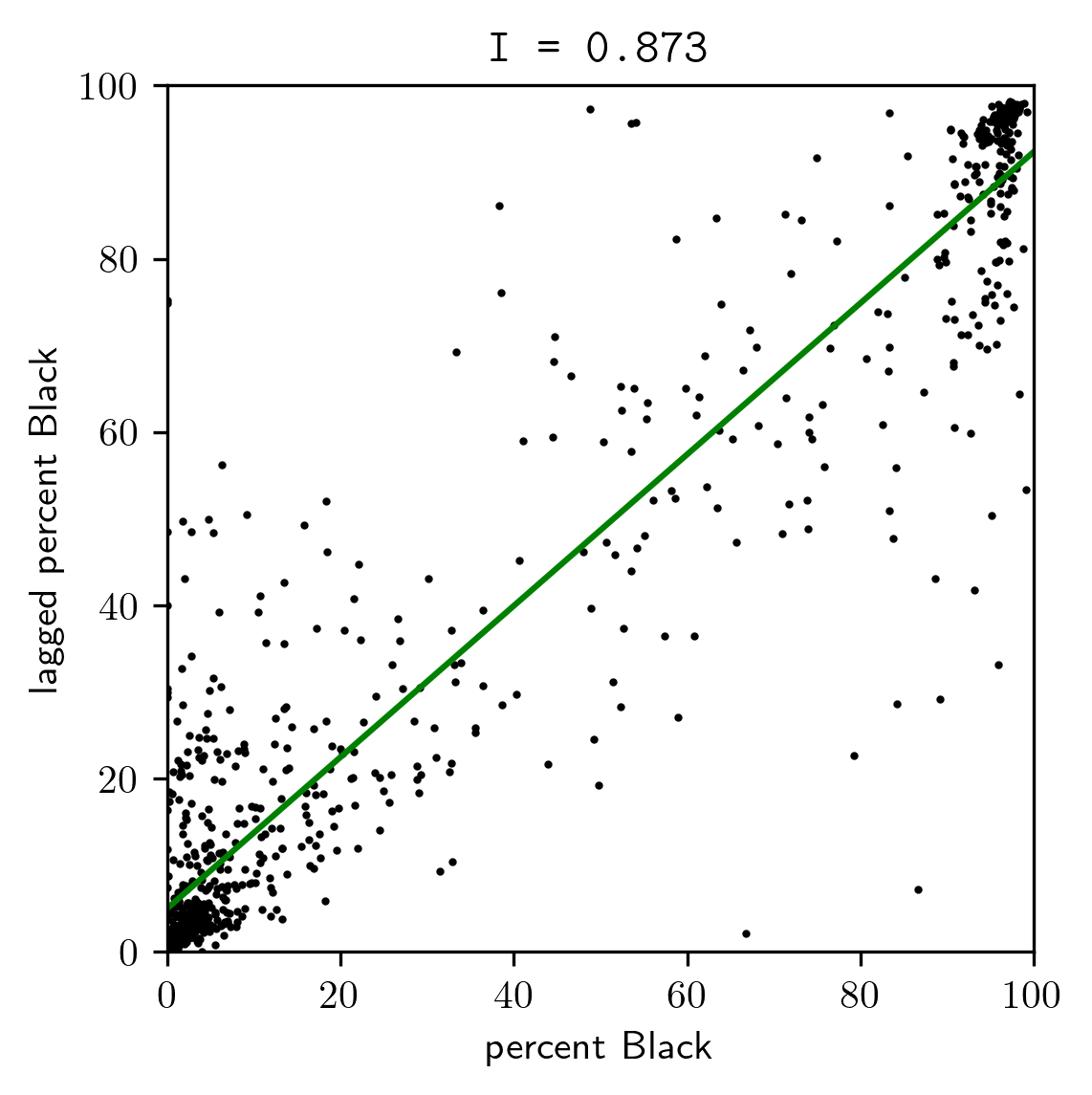}};
\node at (10.5,0) {\includegraphics[width=1.9in]{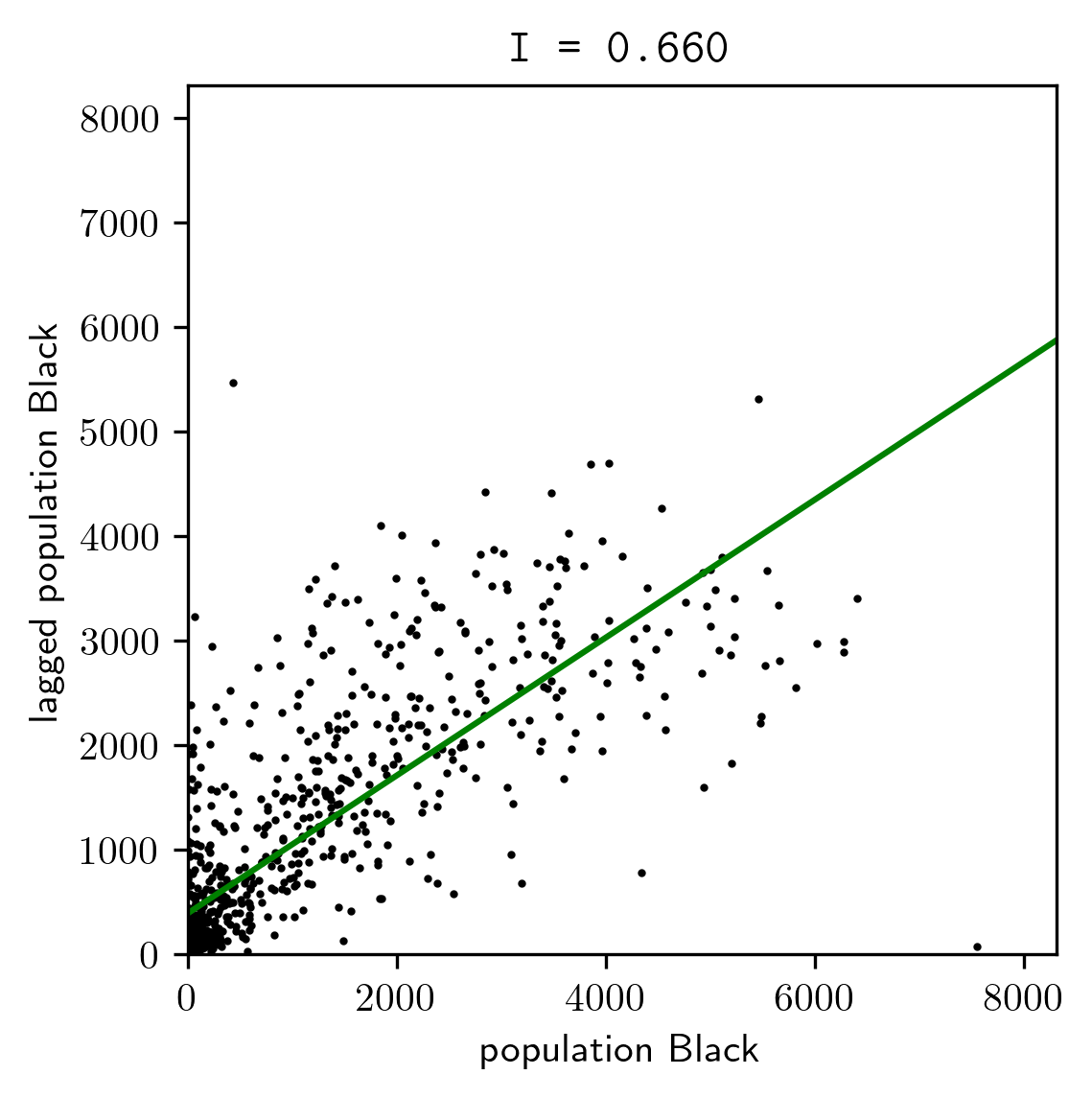}};

\node at (0,-3.5)
{\includegraphics[width=1.8in]{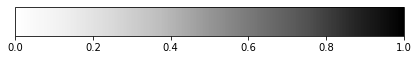}};
\end{tikzpicture}

\caption{Hispanic and Black population data in the 2010 census tracts of Chicago, colored by proportion with white corresponding to 0\% and black corresponding to 100\%. In the scatterplots, each dot is positioned according to the share (proportion) or count (total number) of the subgroup in that tract ($x$-axis) and its neighboring tracts ($y$-axis).  Here, Moran's \I is calculated with respect to the row-standardized adjacency matrix $P$.}
    \label{fig:Chicago}
\end{figure}
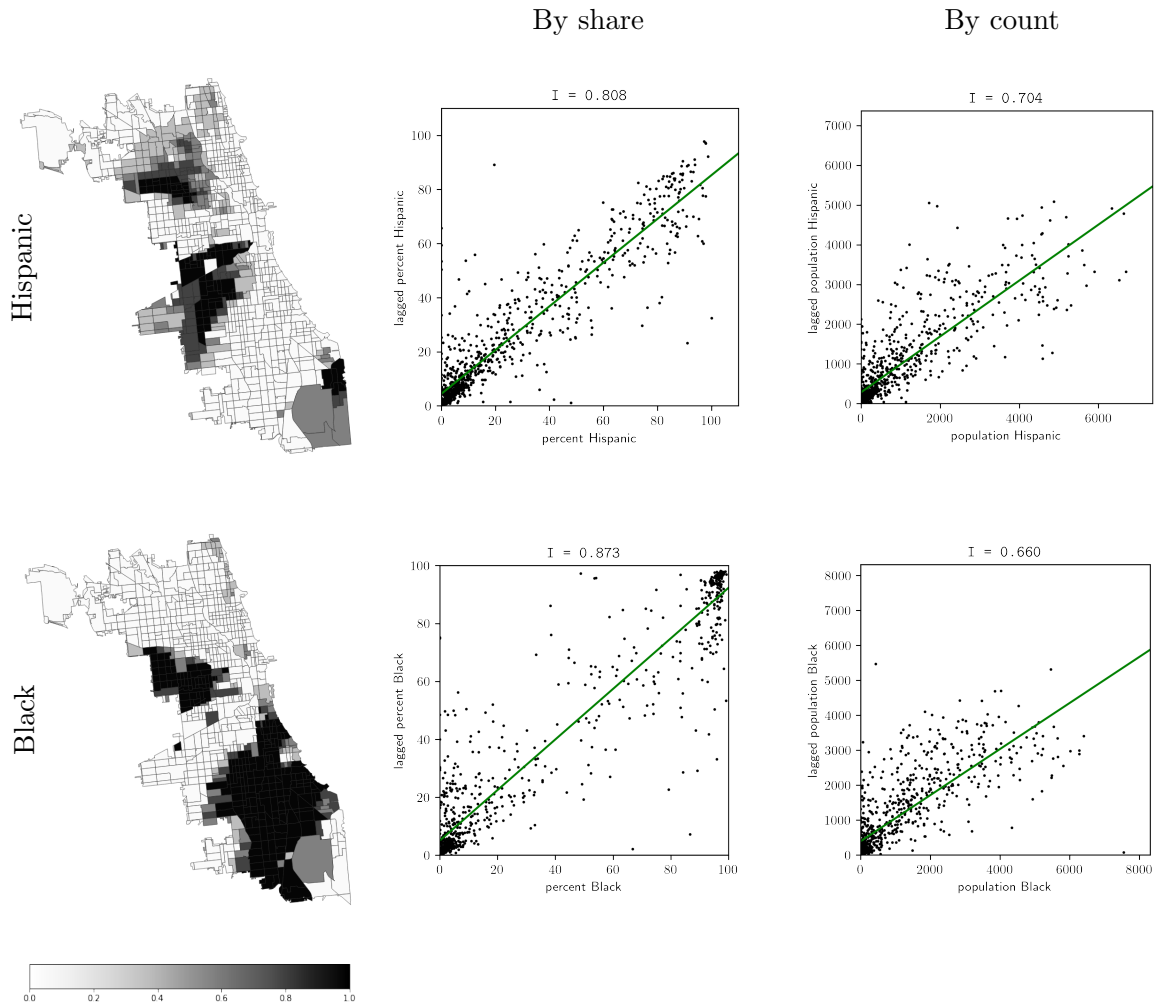

A scatterplot of $\U$ (spatially lagged variable) versus $\V$ (original variable) has been called a \emph{Moran scatterplot} \cite{Anselin1996_Moran}. Figure \ref{fig:Chicago} shows such a plot for the Hispanic and Black populations of Chicago by census tract.  Moran's \I is just the slope of the best fit line---shown in green in each plot. The positive correlation, and thus positive value of Moran's \I, is easily observable in these cases. Indeed, as the figure shows, the Hispanic and Black population in Chicago are both very clustered.  The connection between segregation, clustering, and graph geometry is developed in \S \ref{sec:GraphFourierAnalysis}.

\FloatBarrier
\subsection{Brief Summary of Prior Work}
\label{subsec:LitReview}

We consider several related methods of quantifying segregation; see the broad surveys \cite{Massey1988_Dimensions, Dawkins2004_Measuring} for further details and the book chapter \cite{Duchin2021_Explainer} for an accessible introduction to $\I$ and related measures.

Moran's $\I$, initially introduced by P.A.P. Moran, was brought into geography during the rise of spatial analysis in the late 1940s. In 1968, Cliff and Ord presented an influential conference paper \cite{Cliff1969_Problem} which introduced and argued for the use of alternative spatial weight matrices for Moran's $\I$, in particular $W$ matrices based on more than just contiguity (see \cite{getis2009spatial} for a further discussion on spatial weight matrices). A common use for Moran's $\I$ is as part of a significance test to see if data are or are not spatially autocorrelated; in such cases, one compares an observed Moran's $\I$ to its distribution under a random function $\V$ on the nodes \cite{agnew2011sage, Anselin1996_Moran, Berry1968_Spatial, deJong1984_Extreme}.  One can also test for spatial autocorrelation in the residuals of ordinary least squares models~\cite{anselin2013spatial}. Direct comparison of Moran's $\I$ values has sometimes been used in comparing segregation levels between regions \cite{roberts2009urban}; \S \ref{subsubsec:I_near0} will suggest that this is a dangerous practice. Anselin introduced the Moran scatterplot in \cite{Anselin1996_Moran} as a way to visualize Moran's $\I$ (see \S \ref{sec:scatter}). Moran scatterplots can also be useful in determining the contribution of particular sub-regions to the overall value of Moran's $\I$. This idea was developed further as part of Anselin's \emph{local indicators of spatial association (LISA)}, a class of methods which forms calculations based on the neighborhood of one node in a network, including local Moran's $\I$ and local Geary's $c$ \cite{anselin1995local}. Outside of geography Moran's $\I$ has been applied in fields as diverse as epidemiology, urban planning, and environmental studies~\cite{getis2008history}. 

Geographers have observed that \I is sensitive to the {\em modifiable areal unit problem}, or MAUP \cite{Berry1968_Spatial}.  Indeed, \I depends quite heavily on the choice of geographic units used (e.g., finer-scale census blocks or coarser-scale census tracts) to construct the associated graph. As a prototype example of this phenomenon, a checkerboard pattern has a Moran's \I of roughly $-1$ (total anticorrelation) as in Figure~\ref{fig:Intro}.  But if a simple checkerboard is aggregated so that its $2\times 2$ regions become the new units of analysis, then the distribution becomes roughly uniform and Moran's \I is roughly 0 (no correlation).   
When choosing units for analysis, the scale and placement of the units will impact all the measurable properties of the region, including the value of Moran's \I.  
On the other hand, there has been considerable development of multiscale graph signal processing tools \cite{Shuman2013_Emerging, Ortega2018_Graph} that allow for multiscale partitions of data-driven graphs using diffusion processes, wavelets, and neural networks; these are potentially interesting tools to capture notions of segregation across spatial scales.  We will address the choice of units in real examples throughout the present paper.

The spatial statistics literature contains numerous examples interpreting Moran's $\I$ in linear algebraic terms, as we do here.  This has been used to provide a framework for regression modeling \cite{griffith2006spatial, tiefelsdorf2007semiparametric} and other statistical analysis \cite{deJong1984_Extreme, Tiefelsdorf1995_Exact}.  Linear algebra is particularly relevant in the context of statistical testing \cite{Tiefelsdorf1995_Exact} where the goal is to understand, for a given weight matrix $W$ and function $\V$ on the nodes, whether $\V$ is more segregated in a statistically significant sense than would be expected under a null model. That is, does $\I(\V;W)$ deviate significantly from $\Ex(\I(\ \cdot \ ; W))$, where the expectation is taken over a suitable null model of vectors $\V$? For example, under the model of no spatial correlation in the graph, where node values are randomly sampled independently from each other, Moran himself observed that $\Ex(\I(\ \cdot \ ; W))=-1/(n-1)$ \cite{Moran1950_Notes}. Distributions around the mean are given in some cases in terms of spectral properties of the graph \cite{Tiefelsdorf1995_Exact}.  

In contrast to the static characterization of segregation captured by $\I$---namely that $\I\gg 0$ indicates segregation---the \emph{Schelling model} provides a dynamical perspective for segregation on graphs \cite{Schelling1971_Dynamic}.  In this model, every node on a network has a label (e.g., the membership in a demographic group), and the network evolves randomly in time as nodes change their label with a propensity towards being similar to their neighbors.  The degree of \emph{homophily} is a model parameter and quantifies how much nodes want to have the same labels as their neighbors.  One expects a network with high homophily to converge in the limit to a more segregated pattern than one with low homophily.  The Schelling model---which bears a family resemblance to models of ferromagnetism in statistical physics such as the Ising model \cite{Ising1925_Beitrag}---has been generalized to characterize complex segregation dynamics on extremely regular networks (e.g., hexagonal lattices) \cite{Bhakta2014_Clustering, Zhao2021_Socioeconomic}.  The perspective taken by this literature is to determine the basic properties of the steady-state distribution of the dynamics (e.g., whether large homogeneous regions emerge, depending on the homophily or related parameters) rather than how the underlying network geometry and population distribution impact the dynamics.

\emph{Network assortativity} \cite{newman2002assortative, newman2003mixing} was developed in network science to measure the propensity of like nodes to connect to one another, by counting the proportion of edges that link similarly labeled nodes and comparing that to the number expected under a null hypothesis of no special preference. 
As with $\I$, assortativity scores are influenced by the degree distribution of the network and the underlying sizes of the populations being measured \cite{newman2003structure}.
Previous work by Alvarez et al. \cite{Alvarez2018_Clustering} has generalized this to the setting that node properties can be real-valued rather than discrete, and the authors construct generalized assortativity scores called clustering propensity (or {\em capy}) scores that have linear algebra definitions similar to the ones that will be discussed in the present paper.  
Based on the observation (Remark \ref{rem:proj} below) that $\I$ is invariant under translation and rescaling, Alvarez et al. also note that interpretations of \I become risky when comparing datasets with different variances, and that the interpretation is particularly noisy when a population is near uniform.  To give a stylized example, consider a city where the east side is 100\% Hispanic and the west side is 0\% Hispanic and another city where every east side tract is 51\% Hispanic while every west side tract is 49\% Hispanic.  To an observer, it would be obvious that the first city is far more segregated than the second, but Moran's \I sees no difference at all.  
Indeed, \I can take any value at all when node values are all between $\vbar-\epsilon$ and $\vbar+\epsilon$, even for very small $\epsilon$.  
We will return to the question of scale-sensitivity in the discussion of future directions presented in the conclusion.

%%%%
\section{Spectral Graph Interpretation}
\label{sec:SpectralGraphTheoryAndI}

In interpreting the values taken by $\I$, it is essential to understand how the graph itself determines the range of achievable values.  In this section we will show that when $W$ is symmetric,  $\I(\ \cdot \ ; W)$ achieves maximum and minimum values at generalized eigenvectors for the pair $(\Proj W \Proj,\Proj)$, i.e., solutions to the equation
$\Proj W \Proj\V=\lambda \Proj\V$, where $\Proj$ is orthogonal projection onto $X$.  When the weight matrix is symmetric and has constant rowsum, this reduces to a standard eigenvalue problem.

\begin{definition}[Rayleigh quotient]
For $W\in\R^{n\times n}$ and $\V\in\R^{n\times 1}$, the \emph{Rayleigh quotient} is  $R(\V; W):=\displaystyle\frac{\V^\T W \V}{\V^\T \V}$.
\end{definition}

It is a standard linear algebra fact that when  $W$ is symmetric, the functions $\V$ that realize  extreme values of $R(\V; W)$ are the eigenvectors corresponding to the smallest and largest eigenvalues of $W$ \cite{Horn2012_Matrix}.  That is, if $W$ has  eigenvalues 
$\lambda_{1}\ge \dots \ge \lambda_{n}$ and corresponding eigenvectors $\Phi_1,\dots, \Phi_n$, then 
 $$\min_{\V\neq 0}\frac{\V^{\T} W\V}{\V^{\T}\V}=\lambda_{n}, \quad \max_{\V\neq 0}\frac{\V^{\T} W\V }{\V^{\T}\V}=\lambda_{1},$$
and those extreme values are realized in the eigenspaces corresponding to $\lambda_n$ and $\lambda_1$, respectively.  This means that the $\V$ realizing extreme values are determined up to scaling when the extreme eigenvalues are simple (multiplicity one).  As we will see below, $\I$ is in general not quite a Rayleigh quotient, but close enough to allow us to analyze it in terms of generalized eigenvalues.

\subsection{Bounds with Adjacency Weights}

When $W=A$, the standard adjacency matrix of an undirected, simple (no self-loops) graph, a range of properties of the graph can be inferred from the the spectrum of $A$.  Recall that a graph is \emph{$d$-regular} if $d_i=d$ for all $i$.  Equivalently, a graph $\G$ is $d$-regular iff $\one$ is an eigenvector of $A$ with eigenvalue $d$.  To emphasize that eigenvalues of $A$ (not necessarily regular) encode connectivity facts about the graph or network $\G$, we record some standard facts from \cite{Chung1997_Spectral}.  In these statements we refer to the $A$-degrees of vertices, which are the standard vertex degrees of the graph.
Recall that a graph is called {\em bipartite} if there are two disjoint sets $A,B\subset V(\G)$
such that all edges of $\G$ have one endpoint in each set.
\begin{itemize}
    \item $\displaystyle\sum_{i=1}^{n} \lambda_{i}=0$, $\displaystyle\sum_{i=1}^{n} {\lambda_i}^2=n\dbar$, and $\displaystyle\sum_{i=1}^{n}{\lambda_i}^3=6t$, for $t$ the number of triangles in $\G$;
    \item $\lambda_1\le \dmax:=\max_i d_i$, the largest degree of any vertex, with equality iff $\G$ is regular.  %In this case, $\lambda_{n}<0$ necessarily.
    \item $|\lambda_i|\le \lambda_1$ for all $i$, and  $\G$ connected iff $\lambda_2<\lambda_1$.
    \item if $\G$ is bipartite, then the eigenvalues are symmetric:  $\lambda_i=-\lambda_{n+1-i}$ for all $i$.
    \item if $\G$ is not bipartite, then $|\lambda_n|<\lambda_1$.
\end{itemize} 

These fundamental facts from spectral graph theory immediately yield bounds on $\I(\cdot \ ; A)$ and $\I(\cdot \ ; P)$.

\begin{theorem}[$\I$ bounds for general graphs with adjacency weights]\label{lem:BoundingI}  Let $A$ be the adjacency matrix of an undirected graph $\G$, with eigenvalues 
$\lambda_1\ge \dots \ge \lambda_n$.
Let $\dmin,\dbar, \dmax$ be the minimum, average, and maximum vertex degree (i.e., $A$-degree) of $\G$, and let $P$ be the row-standardized adjacency matrix described above.

\begin{enumerate}[(a)]
    \item  The range of possible $\I$ values satisfies $\I(X;A)\subseteq \left[\frac{\lambda_{n}}{\dbar},\frac{\lambda_{1}}{\dbar}\right]
    \subseteq \left[\frac{-\dmax}{\dbar},\frac{\dmax}{\dbar}\right]
    $
    and $\I(X;P)\subseteq
    \left[\frac{\lambda_{n}}{\dmin},\frac{\lambda_{1}}{\dmin}\right] \subseteq \left[\frac{-\dmax}{\dmin},\frac{\dmax}{\dmin}\right]$.  Note that the bounds for $P$ are still in terms of the eigenvalues and degrees of $A$.
    
    \item If $\G$ is irregular and not bipartite, $\I(X;A)\subsetneq \left( -\frac{\dmax}{\dbar},
\frac{\dmax}{\dbar}\right).$
\end{enumerate}
\end{theorem}

\begin{proof}To see (a), note that \[\frac{\lambda_{n}}{\dbar}=\min_{\V\neq \zero} \frac{1}{\dbar}\frac{\V^\T A\V}{\V^\T \V} \le 
\min_{\X\neq \zero,\X\perp\one}\frac{1}{\dbar}\frac{\X^{\T} A\X}{\X^\T \X}
=\min \I(X;A).\]
The upper bound is similar.  To see the inclusion for $\I(X;P)$, note that

\begin{align*}
    \min_{\V\neq \zero}\frac{\V^{\T}P\V}{\V^{\T}\V}=&\min_{\V\neq \zero}\frac{\V^{\T}D^{-1}A\V}{\V^{\T}\V}
    =\min_{\V\neq \zero}\frac{\V^{\T}D^{-1/2}(D^{-1/2}AD^{1/2})D^{-1/2}\V}{\V^{\T}\V}\\
    =&\min_{\V\neq \zero}\frac{(D^{-1/2}\V)^{\T}(D^{-1/2}AD^{1/2})(D^{-1/2}\V)}{\V^{\T}\V}=\min_{\U\neq \zero}\frac{\U^{\T}(D^{-1/2}AD^{1/2})\U}{\U^{\T}D\U},
\end{align*}
where in the last equality we simply make the change of variables $\U=D^{-1/2}\V$.  Now, because $D$ is diagonal with smallest entry $\dmin$, we have $0\le\dmin\U^{\T}\U\le\U^{\T}D\U$.  Moreover, $D^{-1/2}AD^{1/2}$ is symmetric and similar to $A$, hence has the same eigenvalues as $A$ which lie in the range $[-\dmax,\dmax]$.  We conclude that
\[\frac{-\dmax}{\dmin}\le \frac{\lambda_{n}}{\dmin}\le\min_{\V\neq \zero}\frac{\V^{\T}P\V}{\V^{\T}\V}\le \min_{\X\neq \zero, \X\perp\one}\frac{\X^{\T}P\X}{\X^{\T}\X}=\min \I(X; P).\]  The upper bound is similar.

To see (b), note that the upper bound follows from (a) and observing that $\lambda_1\le \dmax$, with equality iff $\G$ is regular (in which case $\dmax=\dbar$).  The lower bound follows similarly, since non-bipartiteness gives $\lambda_{n}>-\lambda_{1}>-\dmax$.
\end{proof}

We note that non-bipartiteness is easily observed visually on the graphs of interest in geographic applications by the presence of at least one triangle, and that the real-world graphs are never exactly regular.

While the proof of Theorem \ref{lem:BoundingI} suggests the eigenvectors $\Phi_{1}$ and $\Phi_{n}$ as candidate maximizers and minimizers of $\I$, respectively, we shall see that the projection onto $X$ complicates matters when the graph is not regular, and requires us to pass to generalized eigenvectors (Theorem~\ref{thm:GEV}).  

\begin{table}[ht] \centering 
\resizebox{\textwidth}{!}{
\begin{tabular}{|c|cc|cc|cc|}
\hline 
& \multicolumn{2}{c|}{blocks} & \multicolumn{2}{c|}{block groups} & \multicolumn{2}{c|}{counties} \\
\hline 
\small{\# nodes} & \multicolumn{2}{c|}{291,086}  & \multicolumn{2}{c|}{5,533} &\multicolumn{2}{c|}{159}  \\
\small{\# edges} & \multicolumn{2}{c|}{1,393,216}  & \multicolumn{2}{c|}{15,344} &  \multicolumn{2}{c|}{418} \\
\scriptsize{$\dmin<\dbar<\dmax$} & \multicolumn{2}{c|}{$1< 9.5725<92$}  & \multicolumn{2}{c|}{$1<5.5464<16$} & 
\multicolumn{2}{c|}{$1<5.2579<10$} \\   
\hline
&  $A$ & $P$ &$A$ & $P$ & $A$ & $P$\\
\hline 
\small{deg.~bounds} & 
\scriptsize{($-9.6108, 9.6108)$} & 
\scriptsize{$(-92,92)$} & \scriptsize{$ (-2.8848, 2.8848)$} & \scriptsize{$(-16,16)$}
& \scriptsize{$(-1.9019, 1.9019)$} & \scriptsize{$(-10,10)$}
  \\
\small{eig.~bounds} 
&\scriptsize{$(-1.0957,1.4486)$} & 
\scriptsize{$(-10.4886,13.8669)$}&
\scriptsize{$(-0.6978,1.1554)$} &
\scriptsize{$(-3.8702,6.4079)$} 
& \scriptsize{$( -0.5849,1.1151)$} & 
\scriptsize{$(-3.0751,5.8629)$}
\\
\small{true \I range}  &
 \scriptsize{$(-1.0957,1.4475)$} & 
 \scriptsize{$(-1.1007,1.2249)$}& 
 \scriptsize{$(-0.6978,1.1526)$} & 
 \scriptsize{$(-.7510,1.0326)$} &
\scriptsize{$(-0.5845,1.0763)$}&
 \scriptsize{$(-.7673,1.0260)$}\\
\hline 
\end{tabular}}

\caption{Values of $\I(\cdot \ ;A)$ and $\I(\cdot \ ;P)$ achievable on Georgia dual graphs from Figure~\ref{fig:GA-duals}, to four decimal places.  We see that the degree bounds from Theorem~\ref{lem:BoundingI} can be far from the eigenvalue bounds when $\dmin\ll \dbar\ll \dmax$.  Note that there are very high-degree nodes  the block graph; these typically occur when large blocks within bodies of water are adjacent to high numbers of coastal blocks.
The true range of achievable $\I$ values is obtained by numerically solving the generalized eigenvalue problem for $A$, as in Theorem~\ref{thm:GEV}, or by the Lagrange multiplier method for $P$ described in Remark~\ref{rmk:lagrange}.
The eigenvalue bounds for $P$ are not tight because the estimate that divides by $\dmin$ is far from sharp, and these graphs all have leaves ($\dmin=1$).}\label{tab:GA}
\end{table}

To test the sharpness of the bounds, we can look at three graphs based on real data from Georgia: the graphs dual to blocks, block groups, and counties that are depicted in Figure \ref{fig:GA-duals}.  Results are summarized in Table \ref{tab:GA}.

\FloatBarrier

This shows that it is possible for \I to fall outside of the $[-1,1]$ range on a realistic graph; later, we will see that in fact the values of $\I$ with respect to $A$ and $P$ can get arbitrarily large (positive or negative) when the degree disparity is exaggerated.
Note that these examples contradict what seems to be the prevailing understanding of the behavior of $\I( \cdot \ ; P)$ in the geography community, where it is said that $|\I(\V ; P)|\le 1$ for all vectors $\V$.\footnote{For instance, this is a stated reason to use $W=P$ in the user guide material for ArcGIS  \cite{ArcGIS2}, the dominant spatial statistics software package, which states that ``In general, the Global Moran's Index is bounded by $-1.0$ and $1.0$. This is always the case when your weights are row standardized."  And later: ``Row standardized weighting is often used with fixed distance neighborhoods and almost always used for neighborhoods based on polygon contiguity. This is to mitigate bias due to features having different numbers of neighbors."}  Presumably this belief can be traced to the fact that, for arbitrary graphs $\G$, the vertices have $P$-degree one and the eigenvalues of $P$ fall in $[-1,1]$, ``fixing" the fact that the $A$-degrees may be non-constant and the eigenvalues of $A$ may be large.  However, we know of no way to bound $\I(\cdot \ ;P)$ in terms of the eigenvalues of $P$; our spectral analysis for $P$ depends on the eigenvalues of $A$.

In contrast, the extremal behavior of $\I$ over $X$ is straightforward when the underlying graph is regular, and a statement similar to the previous theorem becomes sharp.

\begin{theorem}[$\I$ bounds for regular graphs with adjacency weights]\label{thm:ConnectedBipartite} Let $A$ be the adjacency matrix of an undirected, $d$-regular graph $\G$, with row-standardization $P=\frac 1d A$.  Let $d=\lambda_{1}\ge\lambda_{2}\ge\dots\ge\lambda_{n}$ be the eigenvalues of $A$ with associated eigenvectors $\Phi_{1},\Phi_{2},\dots,\Phi_{n}$.  Then
\begin{enumerate}[(a)]
\item $\I(X;A)=\left[\I(\Phi_n; A),\I(\Phi_2; A) \right]\subseteq [-1,1]$.
\item $\I(\Phi_2 ; A)=1$ iff $\G$ is disconnected.
\item  $\I(\Phi_n; A)=-1$ iff $\G$ is bipartite.
\end{enumerate}
Since $A$ is just a scalar multiple of $P$ in the regular case, the same bounds and equalities hold for $\I(\cdot \ ; P)$.
\end{theorem}

\begin{proof} Theorem \ref{lem:BoundingI} gives 
$\I(X;A)\subseteq[-1,1]$
after noting that $\lambda_{1}\le \dmax$ and $\lambda_{n}\ge -\lambda_{1}$.  
To finish (a), first note that if $A$ is the adjacency matrix of a $d$-regular graph, then $\Phi_{1}=\one$.  By the Courant–Fischer–Weyl min-max principle \cite{Horn2012_Matrix}, \[\Phi_{2}\in\argmax_{\V\neq\zero, \V\perp\Phi_{1}}\frac{\V^{\T} A\V}{{\|\V\|_{2}}^{2}}=\argmax_{\V\neq\zero, \V\perp\one}\frac{\V^{\T} A\V}{{\|\V\|_{2}}^{2}}=\argmax_{\V\neq 0}\I(\V; A).\] The other bound follows from noting that  $\Phi_{n}\perp \Phi_{1}$, giving \[\argmin_{\V\neq\zero, \V\perp\one}\frac{\V^{\T} A\V}{{\|\V\|_{2}}^{2}}=\argmin_{\V\neq\zero}\frac{\V^{\T} A\V}{{\|\V\|_{2}}^{2}}\ni\Phi_{n}.\]

For (b), note that $\displaystyle\max_{\V\neq 0}\I(\V; A)=\frac{\lambda_{2}}{d}\le\frac{\lambda_{1}}{d}$.  The Perron-Frobenius Theorem gives $\lambda_{2}=\lambda_{1}=d$ iff $\G$ is disconnected \cite{Horn2012_Matrix}.

To see (c), we use  $\displaystyle\min_{\V\neq 0}\I(\V; A)=\frac{\lambda_{n}}{d}$.   Perron-Frobenius  gives $-\lambda_{n}\le\lambda_{1}=d$, with equality iff $A$ is bipartite.
\end{proof}

\subsection{Analysis for Symmetric Weight Matrices} \label{subsec:IrregularGraphs}
In preparation for proposing alternatives for the spatial weight matrix $W$, we now provide an analysis for arbitrary symmetric matrices.
When the graph is regular, its adjacency matrix $A$ has $\one$ as the eigenvector with largest eigenvalue.  
This makes projection onto $X=\one^\perp$ interact nicely with spectral analysis.   For general graphs, we will handle the projection more carefully.

Let $\Proj=I-\frac{1}{n}\one\one^{\T}$ denote the orthogonal projection onto $X$, i.e., $\Proj\V=\V-\vbar\one=\X$.  Noting that $\Proj \Proj ^{\T}=\Proj ^{2}=\Proj $, we have 
\[\I(\V;W)=\frac{\V^{\T} \Proj W\Proj \V}{\V^{\T}\Proj \V}.\]  
This is no longer a Rayleigh quotient, but rather a \emph{generalized Rayleigh quotient}.  Since $\Proj $ is singular, having $\one$ in its kernel, $\I(\V;W)$ cannot be reduced to a standard Rayleigh quotient via $\Proj ^{-\frac{1}{2}}$.  Instead, one can use the theory of generalized eigenvalues for the matrix pair $(A,B)$, i.e., solutions to $A\V=\lambda B \V$ \cite{van1996matrix, parlett1998symmetric}.
In our case we will consider the generalized spectrum $\{(\lambda_{i},\Phi_{i})\}_{i=1}^{n-1}$ of non-constant $\Phi_{i}$ satisfying $\Proj W\Proj \Phi_{i}=\lambda_{i}\Proj \Phi_{i}$.  
Since $W$ and $\Pi$ and thus $A=\Pi W \Pi$ are symmetric and $B=\Pi$ is positive semi-definite, the generalized eigenvalues are real, and the eigenvectors can be chosen to satisfy $(\Proj \Phi_{i})^{\T}\Proj \Phi_{j}=\Phi_{i}^{\T}\Proj \Phi_{j}=0, \ i\neq j$ and $(\Proj \Phi_{i})^{\T}\Proj \Phi_{i}=\Phi_{i}^{\T}\Proj \Phi_{i}=1$. This means that the vectors are orthonormal after projection, so we will say such generalized eigenvectors are {\em $\Proj $-orthonormal}.  Using the fact that $\one$ is orthogonal to each $\Proj \Phi_{i}$, we get the following diagonalization-style statement
 for \I.  

\begin{theorem}[Spectral interpretation of $\I$ for symmetric weight matrices]\label{thm:GEV}  Let $W$ be a symmetric $n\times n$ weight matrix, and let $\{(\lambda_{i},\Phi_{i})\}_{i=1}^{n-1}$
be   $\Proj $-orthonormal generalized eigenvectors for
the pair $(\Proj W\Proj ,\Proj )$.
Then for all non-zero $\V\in\R^{n\times 1}$,

\begin{enumerate}[(a)]
    \item  $\V=\left(\displaystyle\sum_{i=1}^{n-1}\alpha_{i}\Proj \Phi_{i}\right)+\vbar\one$, for some coefficients $\{\alpha_{i}\}_{i=1}^{n-1}$.
    
    \item  $\I(\V;W)=\displaystyle\sum_{i=1}^{n-1}\alpha_{i}^{2}\lambda_{i}\bigg/\displaystyle\sum_{i=1}^{n-1}\alpha_{i}^{2}$.
\end{enumerate}

\end{theorem}

\begin{proof}The result in (a) follows immediately from the $\Proj $-orthogonality of $\{\Phi_{i}\}_{i=1}^{n-1}$ and the fact that $\one$ generates the kernel of $\Proj $.   Note that this could be done either on the left or on the right in the non-symmetric case, but is unambiguous since $W$ symmetric.

To see (b), we compute
\begin{align*}\I(\V;W)=&\frac{\V^{\T} \Proj W\Proj \V}{\V^{\T} \Proj \V}=\frac{\left(\displaystyle\sum_{i=1}^{n-1}\alpha_{i}\Proj \Phi_{i}+\vbar\one\right)^{\T}\Proj W\Proj \left(\displaystyle\sum_{i=1}^{n-1}\alpha_{i}\Proj \Phi_{i}+\vbar\one\right)}{\left(\displaystyle\sum_{i=1}^{n-1}\alpha_{i}\Proj \Phi_{i}+\vbar\one\right)^{\T}\Proj \left(\displaystyle\sum_{i=1}^{n-1}\alpha_{i}\Proj \Phi_{i}+\vbar\one\right)}\\
=&\frac{\left(\displaystyle\sum_{i=1}^{n-1}\alpha_{i}\Proj ^{2}\Phi_{i}+\vbar\Proj \one\right)^{\T}W\left(\displaystyle\sum_{i=1}^{n-1}\alpha_{i}\Proj ^{2}\Phi_{i}+\vbar\Proj \one\right)}{\left(\displaystyle\sum_{i=1}^{n-1}\alpha_{i}\Proj \Phi_{i}+\vbar\one\right)^{\T}\left(\displaystyle\sum_{i=1}^{n-1}\alpha_{i}\Proj ^{2}\Phi_{i}+\vbar\Proj \one\right)}\\
=&\frac{\left(\displaystyle\sum_{i=1}^{n-1}\alpha_{i}\Proj \Phi_{i}\right)^{\T}W\left(\displaystyle\sum_{i=1}^{n-1}\alpha_{i}\Proj \Phi_{i}\right)}{\left(\displaystyle\sum_{i=1}^{n-1}\alpha_{i}\Proj \Phi_{i}+\vbar\one\right)^{\T}\left(\displaystyle\sum_{i=1}^{n-1}\alpha_{i}\Proj \Phi_{i}\right)}
=\frac{\left(\displaystyle\sum_{i=1}^{n-1}\alpha_{i}\Phi_{i}\right)^{\T}\left(\displaystyle\sum_{i=1}^{n-1}\alpha_{i}\Proj W\Proj \Phi_{i}\right)}{\left(\displaystyle\sum_{i=1}^{n-1}\alpha_{i}\Phi_{i}\right)^{\T}\left(\displaystyle\sum_{i=1}^{n-1}\alpha_{i}\Proj \Phi_{i}\right)}\\
=&\frac{\left(\displaystyle\sum_{i=1}^{n-1}\alpha_{i}\Phi_{i}\right)^{\T}\left(\displaystyle\sum_{i=1}^{n-1}\alpha_{i}\lambda_{i}\Proj \Phi_{i}\right)}{\displaystyle\sum_{i,j=1}^{n-1}\alpha_{i}\alpha_{j}\Phi_{i}^{\T}\Proj \Phi_{j}}
=\frac{\displaystyle\sum_{i,j=1}^{n-1}\alpha_{i}\alpha_{j}\lambda_{i}\Phi_{i}^{\T}\Proj \Phi_{j}}{\displaystyle\sum_{i,j=1}^{n-1}\alpha_{i}\alpha_{j}\Phi_{i}^{\T}\Proj \Phi_{j}}
=\frac{\displaystyle\sum_{i=1}^{n-1}\alpha_{i}^{2}\lambda_{i}}{\displaystyle\sum_{i=1}^{n-1}\alpha_{i}^{2}}.
\end{align*}

\end{proof}

By the usual scale-invariance, 
Theorem \ref{thm:GEV} allows us to understand the behavior of $\I$ just considering $\displaystyle\sum_{i=1}^{n-1}\alpha_{i}^{2}\lambda_{i}$ when $\displaystyle\sum_{i=1}^{n-1}\alpha_{i}^{2}=1$.  
Theorem \ref{thm:GEV} says, in other words, that $\I$ values are precisely the convex combinations of the generalized eigenvalues.  Since $\{\Proj \Phi_{i}\}_{i=1}^{n-1}$ is an orthonormal basis for $X$, this is analogous to the classical spectral analysis of the Rayleigh quotient. 

We note that the generalized eigenpairs of $(\Proj W\Proj,\Proj)$ can be put in correspondence with those of $\Proj W\Proj$ as follows:

\begin{lemma}Let $W$ be a symmetric matrix.

\begin{enumerate}[(a)]
\item  If $(\lambda, \Phi)$ is an eigenpair of $(\Proj W \Proj, \Proj)$, then $(\lambda, \Proj \Phi)$ is an eigenpair of $\Proj W \Proj$.  

\item If $(\lambda, \Phi)$ is an eigenpair of $\Proj W \Proj$, then $(\lambda, \Phi)$ is an eigenpair of $(\Proj W \Proj, \Proj)$.

\end{enumerate}
\begin{proof}To see (a), note that $\Proj^{2}=\Proj$ and thus $\Proj W \Proj (\Proj \Phi) = \Proj W \Proj \Phi = \lambda \Proj \Phi$.  To see (b), notice that $\Proj W \Proj \Phi = \lambda\Phi$ implies $\Proj^{2} W \Proj \Phi = \lambda\Proj\Phi$.  Again, $\Proj^{2}=\Proj$ and the result follows.
\end{proof}
\end{lemma}

\begin{corollary}[Extreme values of $\I$]
\label{cor:extremeI}
When $W$ is symmetric, the minimum and maximum values of $\I(\cdot \ ; W)$ are the smallest and largest generalized eigenvalues of $(\Proj W\Proj \, , \, \Proj )$, respectively, and are achieved at the corresponding generalized eigenvectors.

Suppose additionally that $W$ has 
 constant rowsum $k$, i.e., the graph is $k$-regular with respect to $W$-degree.
Then the eigenvectors of $W$ are equal to the generalized eigenvectors of $(\Pi W \Pi,\Pi)$. The eigenvalues agree except possibly for the eigenvalue associated to $\one$, which is $k$ for $W$ and zero for the generalized problem.
\end{corollary}

\begin{proof}
The observation that $\I$ values are convex combinations  establishes that the extremes are realized at the largest and smallest generalized eigenvalues.
Next, we note that $W\one=k\one$ while 
$\Pi W \Pi \one = \zero$, which establishes the last statement.

Now consider $\X\in X$.  From 
Theorem~\ref{thm:GEV}, we can express $\X$ in the eigenbasis $\{(\lambda_i,\Phi_i)\}_{i=1}^{n-1}$ that spans $X$, and we note that $\Pi$ is the identity on $X$, so it preserves $\X$ and $W\X$.
This gives us
$$W\X=\lambda\X \iff W \Pi \X=\lambda\Pi\X
\iff \Pi W \Pi \X = \lambda \Pi \X,
$$
identifying the eigenvectors with the generalized eigenvectors, as needed.
\end{proof}

\begin{remark}[Extension to $\I(\cdot\ ;P)$]
\label{rmk:lagrange}
The matrix $P$ is not symmetric, so the above orthogonal decomposition does not apply (in particular, the left and right eigenvectors of $P$ are different).  However, the variational problem of minimizing or maximizing $\frac{\V^{\T}\Proj P\Proj\V}{\V^{\T}\Proj\V}$ over the space of non-zero vectors reduces, by scale-invariance, to the constrained optimization of  $\V^{\T}\Proj P\Proj\V$ subject to the constraint that $\V^{\T}\Proj\V=1$. This has associated Lagrangian function $\V^{\T}\Proj P\Proj\V+\zeta(\V^{\T}\Proj\V-1)$ for some scalar $\zeta$.  The $\V$-derivative of the Lagrangian is $\V^{\T}\Proj(P+P^{\T})\Proj+\zeta(2\V^{\T}\Proj)$.  Setting equal to 0, we see $\frac{\V^{\T}\Proj P\Proj\V}{\V^{\T}\Proj\V}$ is maximized and minimized at the generalized eigenvectors of $(\frac{1}{2}\Proj(P+P^{\T})\Proj,\Proj)$ corresponding to the extreme eigenvalues.
\end{remark}

\section{Comparing \I Within and Across Graphs}
\label{sec:ComparingI}
%%%%%

Theorem \ref{thm:GEV}  gives us tools to study the question of what kinds of vectors $\V_1,\V_2$ have $\I(\V_{1};W)\approx \I(\V_{2};W)$.  We consider two cases: $|\I|\gg 0$ and $|\I|\approx 0$.  We will continue to suppose that vectors are scaled so that $\X=\V-\vbar\one$ has $\ell^{2}$ norm 1, i.e., $\displaystyle\sum_{i=1}^{n-1}\alpha_{i}^{2}=1$, and we let the generalized eigenvalues of $(\Proj W\Proj ,\Proj )$ be $\lambda_{1}\ge \lambda_{2}\ge\dots\ge\lambda_{n-1}$.

\begin{figure}[ht]
    \centering

\begin{tikzpicture}[yscale=.9]
\node at (0,7) {\includegraphics[width=2.9in]{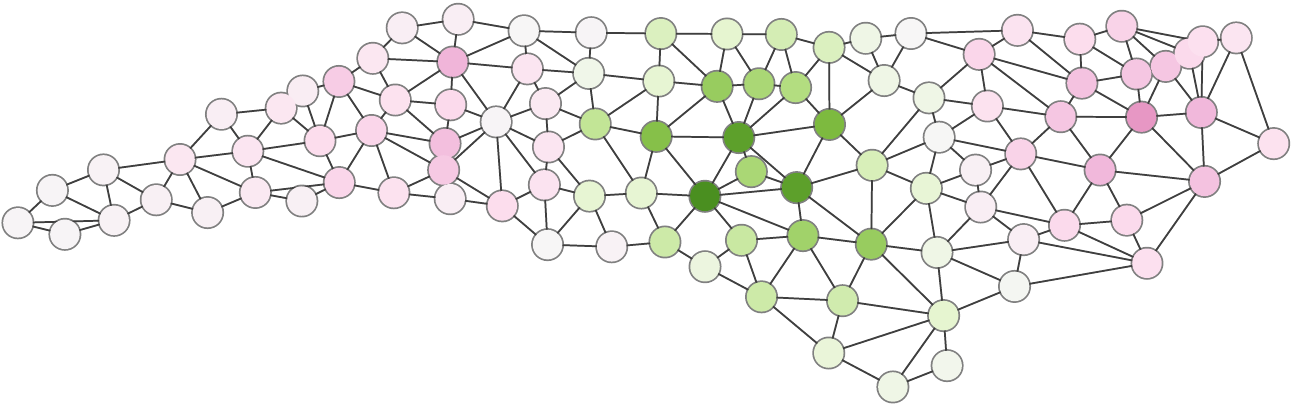}};
\node at (0,8.5) {$\I=1$};
\node at (8,7) {\includegraphics[width=.45\textwidth]{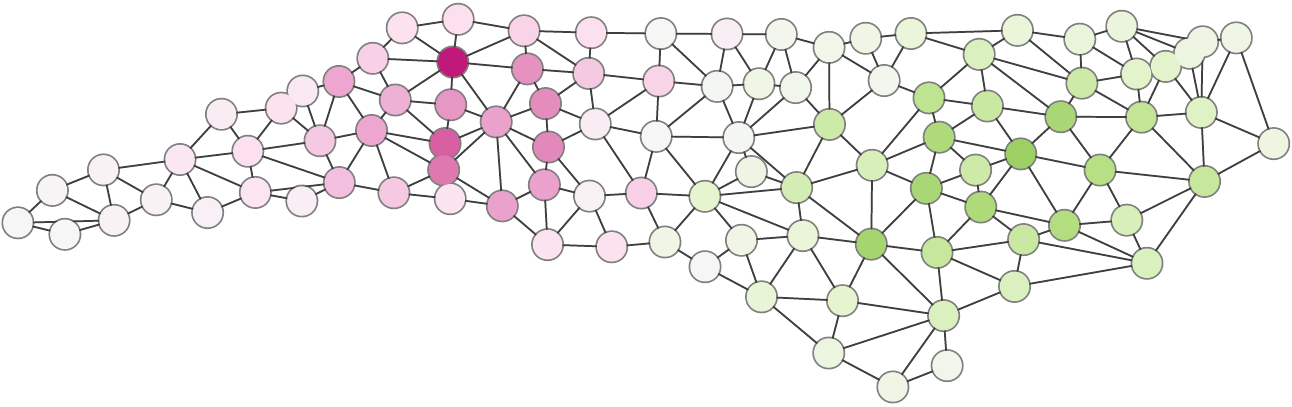}};
\node at (8,8.5) {$\I=1$};

\node at (0,3.5) {\includegraphics[width=2.9in]{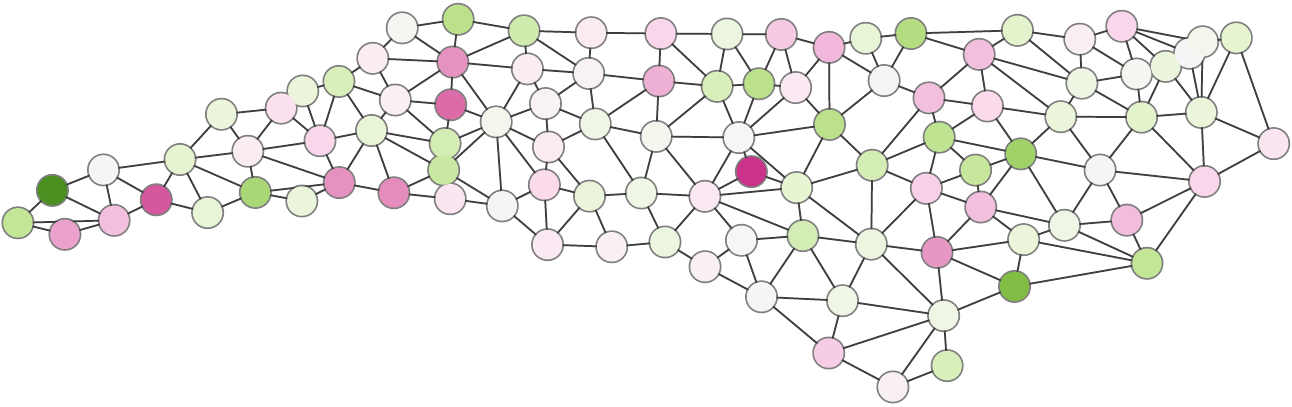}};
\node at (8,3.5) {\includegraphics[width=2.9in]{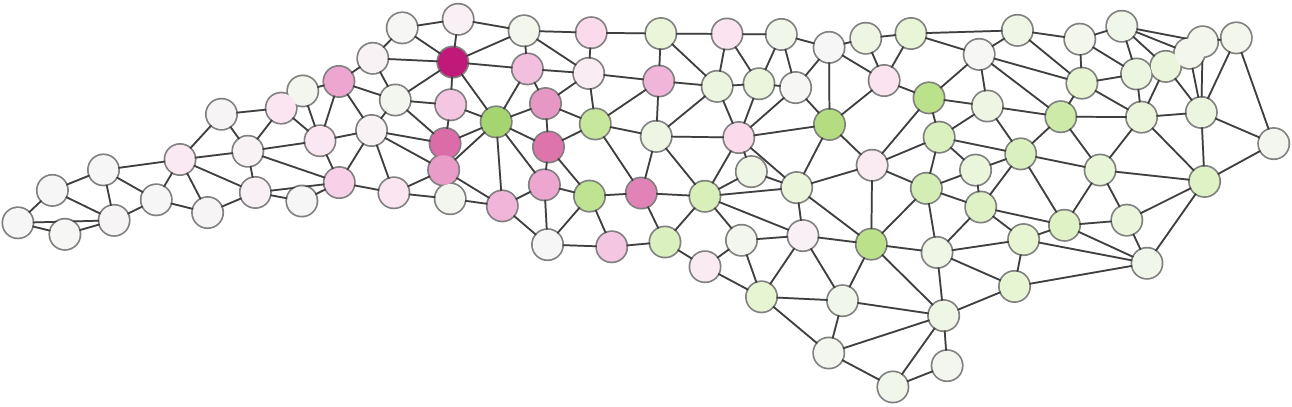}};
\node at (0,5) {$\I=0$};
\node at (8,5) {$\I=0$};

\node at (0,0) {\includegraphics[width=2.9in]{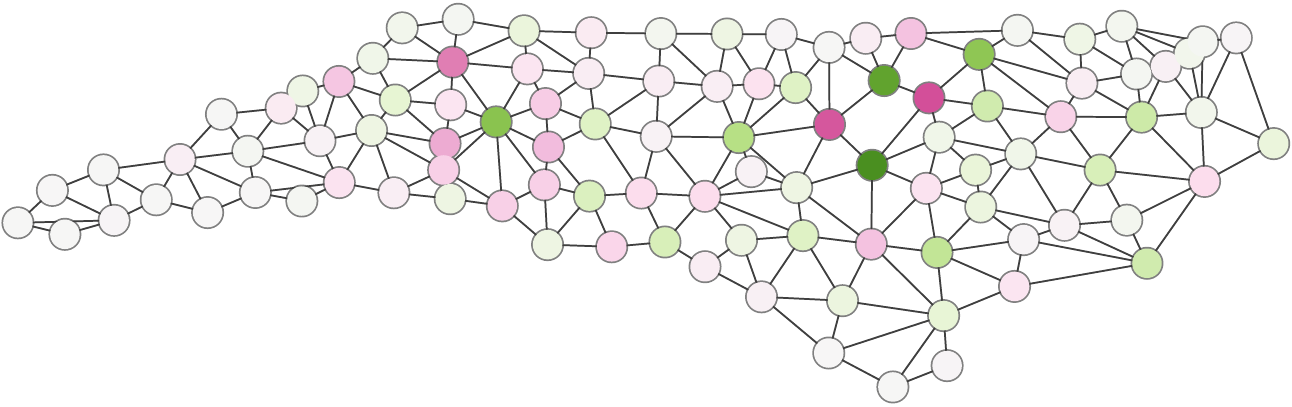}};
\node at (8,0) {\includegraphics[width=2.9in]{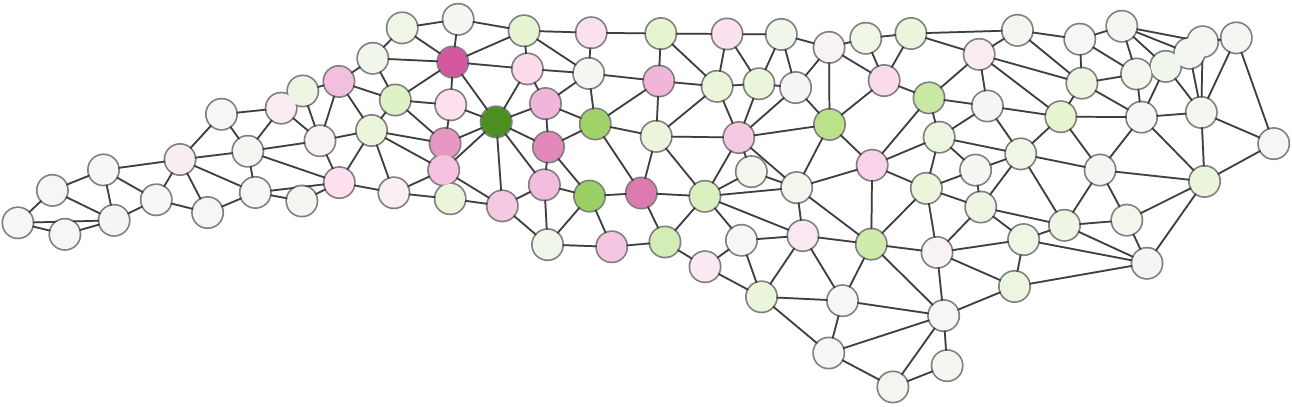}};
\node at (0,1.5) {$\I=-1/2$};
\node at (8,1.5) {$\I=-1/2$};
\end{tikzpicture}    
    
    \caption{The maximum $\I(\V;A)$ value on this graph is $\approx 1.1034$ and there are several eigenvectors yielding $\I>1$ with respect to $A$; in particular there is not a large spectral gap between $\lambda_{1}$ and $\lambda_{2}$. This means that visibly different cluster patterns can realize the same \I value, even when it is quite close to the maximum. 
    The two functions in the middle row both realize $\I(\V;A)=0$, but they are qualitatively rather different: the one on the left is close to spatially uncorrelated, while the one on the right is a linear combination of a cluster pattern with a localized checkerboard.
    The bottom row shows two functions with $\I(\V;A)=-.5$, reasonably close to the minimum of $-.5983$,  both exhibiting something like a ``localized checkerboard" pattern. }
    \label{fig:InterpretationOfI_NC}
\end{figure}
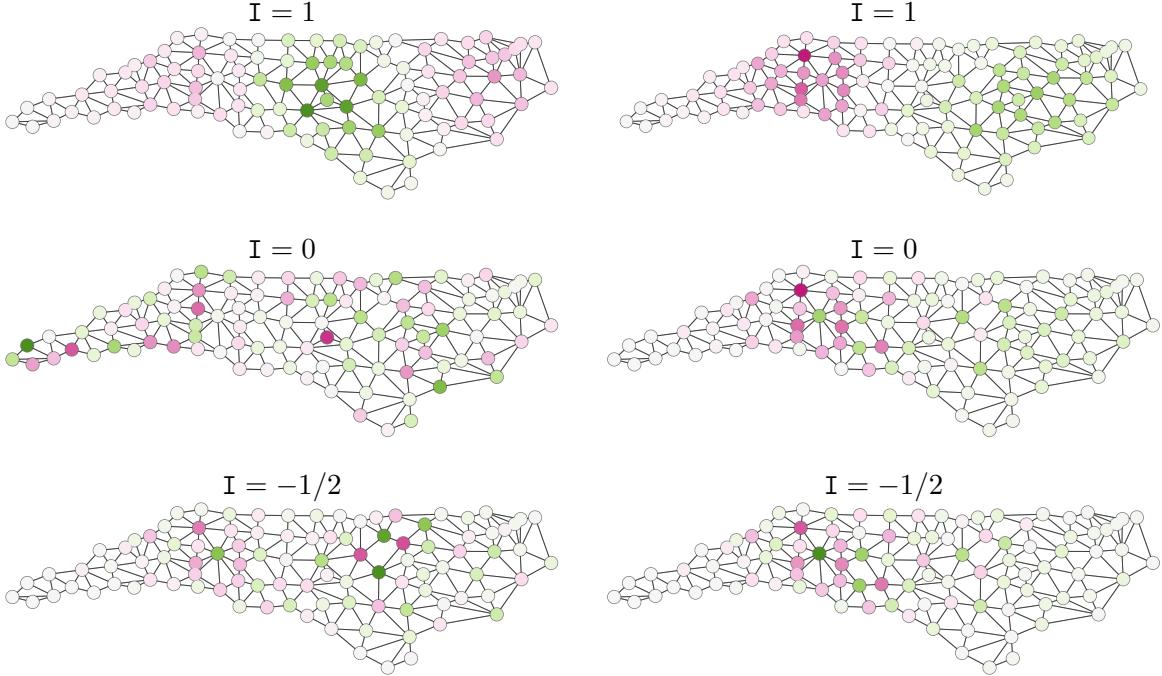

\subsection{Analysis When \texorpdfstring{$|\I|\gg 0$}{|I| Much Greater Than 0}}  We first consider $\I(\V;W)\gg 0$.  Since we can assume $\sum_{i=1}^{n-1}\alpha_{i}^{2}=1$, we have $\I(\V;W)=\sum_{i=1}^{n-1}\alpha_{i}^{2}\lambda_{i}.$  So, $\I(\V;W)$ is large iff most of the coefficient energy in $(\alpha_{1},\dots,\alpha_{n-1})$ is localized towards the lowest-indexed values (corresponding to largest eigenvalues). 

To study large $\I$, let us write $\beta_i=\alpha_i^2$; then we can express the zero-centered vectors with $\I\ge \lambda$ as 
$$X^{+}(\lambda)=\left\{\X=\sum_{i=1}^{n-1}\sqrt{\beta}_{i}\Proj \Phi_{i} \ \bigg| \ \sum_{i=1}^{n-1}\beta_{i}=1,\quad \beta_{i}\ge 0, \quad \sum_{i=1}^{n-1}\beta_{i}\lambda_{i}\ge\lambda\right\}.$$
In the generic case that $\lambda_1$ is simple, we have $X^{+}(\lambda_{1})=\{\Proj \Phi_{1}\}$.  
Clearly $X^{+}(\lambda)\subseteq X^{+}(\lambda')$ when $\lambda\ge \lambda'$ and $X^{+}(\lambda)=\emptyset$ for $\lambda>\lambda_{1}$.  

The expression for $X^{+}(\lambda)$ can be understood geometrically as a portion of the standard simplex $\sum\beta_i=1$, $\beta_i>0$ to one side of the hyperplane $\sum\beta_i\lambda_i=\lambda$.  In high dimensions, most of the mass of the simplex concentrates away from the vertices \cite{Paouris2006_Concentration}, which implies that the volume of $X^{+}(\lambda)$ is small for $\lambda\approx\lambda_{1}$.  In particular, if $\lambda$ is large, any $\X\in X^{+}(\lambda)$ will need to have a large portion of its coefficient energy coming from the $\Proj \Phi_{i}$ with largest eigenvalues.  In the case that there is a significant spectral gap ($\lambda_{1}\gg \lambda_{2})$, then the coefficient energy needs to localize on $\Proj \Phi_{1}$.

A consequence of the concentration of $X^{+}(\lambda)$ for $\lambda$ close to $\lambda_{1}$ is that distinct $\X,\X'\in X^{+}(\lambda)$ share certain \emph{qualitative} properties when $\lambda$ is close to $\lambda_{1}$.  Indeed, under mild assumptions (see \S \ref{sec:GraphFourierAnalysis}), the largest eigenvectors correspond to \emph{clustered} patterns in the data.  If there is a spectral gap, then a single clustered pattern must dominate for $\lambda$ large enough.  See Figure \ref{fig:InterpretationOfI_NC} for a visualization, showing conversely that with a small spectral gap, many different clustered patterns can achieve the same high $\I$ scores.

The same arguments apply to sublevel sets $X^{-}(\lambda)$ when $\lambda\approx\lambda_{n-1}$.  In the case of irregular graphs, the interpretation of the eigenvectors of $(\Proj W\Proj ,\Proj )$ with smallest eigenvalues is more difficult \cite{Desai1994_Characterization}, but one could qualitatively comment that they capture (localized) checkerboard patterns, as in Figure \ref{fig:InterpretationOfI_NC}.

\subsection{Analysis When \texorpdfstring{$|\I|\approx 0$}{I Near 0}}
\label{subsubsec:I_near0}

On the other hand, the $\X$ with $\I(W;\X)\in (-\epsilon,\epsilon)$  need not have any meaningful qualitative properties in common, even as $\epsilon\to 0^+$.  The constraint that $\I$ lie in $(-\epsilon,\epsilon)$ does not imply that the energy of the coefficients $(\alpha_{1},\alpha_{2},\dots,\alpha_{n-1})$ must localize on specific indices.  If $\I(\X; W)\approx 0$ then there must be coefficients that place energy on both positive and negative eigenvalues, so that $\X$ could show some cluster structure, some localized checkerboarding, or appear spatially uncorrelated (recall that the expected value of $\I(\cdot \ ; A)$ is $-\frac{1}{n-1}$ under a spatially uncorrelated random model, giving values near zero for large graphs).  See Figure \ref{fig:InterpretationOfI_NC}  for two qualitatively very different functions, both with $\I=0$.

\FloatBarrier
\section{Overview of Choices of \texorpdfstring{$W$}{W}}\label{sec:Empirical}

The choice of spatial weight matrix $W$ can have a major impact on the interpretability of $\I$.  We now overview and compare four choices of $W$ for use in Moran's $\I$.

\newpage

\begin{definition}[Alternative spatial weight matrices]
Given a graph $\G$, let $D$ be the diagonal matrix given by $D_{ii}$ as the $i$th vertex degree.
Then we will consider the following matrices.
\begin{itemize}
    \item $W=A$, the standard adjacency matrix;
    \item $W=P$, the row-standardization $P=D^{-1}A$;
    \item $W=L$, the unnormalized graph Laplacian $L=D-A$; \quad and
    \item $W=\MH$, the doubly-stochastic matrix defined by 
    $$\MH_{ij}=\begin{cases} A_{ij}/\max(D_{ii},D_{jj}) & i\neq j \\ 1-\sum_{k\neq i} M_{ik} & i=j \end{cases}
    $$
\end{itemize}
\end{definition}
The Laplacian is a ubiquitous choice of matrix associated to a graph that encodes its geometry and topology, so it is a natural choice here.  
In \S\ref{sec:WalksandI} below, we will motivate $\MH$ as a doubly-stochastic approximation to $A$, which will provide both nice numerical properties and an appealing random-walk interpretation.
To see that $\MH$ is indeed doubly stochastic, i.e., has rows and columns summing to one, note that the off-diagonals are defined symmetrically in $i$ and $j$. Since the diagonal entries are defined to make the rows sum to one, the columns must sum to one as well.

While $A$, $L$, and $\MH$ are symmetric, $P$ is not, but it can often be handled by similar techniques, as we saw in Remark~\ref{rmk:lagrange}.
Extremizing $\I$ amounts to solving a standard eigenvalue problem on $L$ and $\MH$, and a generalized eigenvalue problem on $A$ and $P$.  (See Corollary~\ref{cor:extremeI}.)

\begin{remark}[Vertex-regular case]\label{rem:AvL}
These spectra, and the \I scores, will differ in general across the choice of weight matrix.  However, in the case of $d$-regular graphs, $P$ is symmetric, and $P=\MH=\frac 1d A$.
The vertex degree matrix is $D=d\sdot I$, so that $L=d\sdot I-A$ and the spectra are related by $\mu_{i}=d-\lambda_{i}$ for each $i$.  In particular,  \[\I(\X;A)=\frac{1}{d}\cdot \frac{\X^{\T}A\X}{\X^{\T}\X}=\frac{1}{d}\cdot\frac{\X^{\T}(dI+A-dI)\X}{\X^{\T}\X}=1-\frac{1}{d}\cdot\frac{\X^{\T}L\X}{\X^{\T}\X}=1-2\cdot\I(\X;L),\] 
so the \I scores for $A,P,\MH$ are equal, and compare to $L$ by a precise affine relationship.
\end{remark} 

Though this relationship will not be exact for general graphs, it helps translate the conventional wisdom of anti-correlation, non-correlation, and clustering to $\I(\cdot \ ;L)$ values of roughly $1$, $1/2$, and $0$, respectively---in particular, lower values of $\I$ are more segregated when the spatial weight matrix is the Laplacian.

\subsection{Comparison on Families of Graphs}

In Sections \ref{sec:GraphFourierAnalysis} and \ref{sec:WalksandI} below, we explore the different mathematical connections and interpretations made possible by using $L$ or $\MH$ instead of the more traditional $A$ or $P$.  First, we empirically compare $\I(\cdot \ ; W)$ for the different spatial weight matrices.
We compare our four spatial weight matrices using a common nearly-regular graph: a hexagon with 16 vertices on each side, drawn in the hex lattice.   The empirical maximizers and minimizers are shown in Figure~\ref{fig:HexagonalLattice}.

\begin{figure}[ht]
    \centering
\begin{tikzpicture}[yscale=.9]

\node at (0,6) {\footnotesize $\I(\Phi_1^A; A)=.9396$}; 
\node at (3,6) {\footnotesize $\I(\V_{\text{max}}^A; A)=1.0265$}; 
\node at (6,6) {\footnotesize $\I(\V_{\text{max}}^P; P)=.9986$}; 
\node at (9,6) {\footnotesize $\I(\V_{\text{min}}^L; L)=.0022$}; 
\node at (12,6) {\footnotesize $\I(\V_{\text{max}}^M; M)=.9958$}; 

\node at (0,2) {\footnotesize $\I(\Phi_n^A; A)=-.5188$}; 
\node at (3,2) {\footnotesize $\I(\V_{\text{min}}^A; A)=-.5188$}; 
\node at (6,2) {\footnotesize $\I(\V_{\text{min}}^P; P)=-.5414$}; 
\node at (9,2) {\footnotesize $\I(\V_{\text{max}}^L; L)=.7817$}; 
\node at (12,2) {\footnotesize $\I(\V_{\text{min}}^M; M)=-.4964$}; 

\node at (0,4) {\includegraphics[width=1.1in]{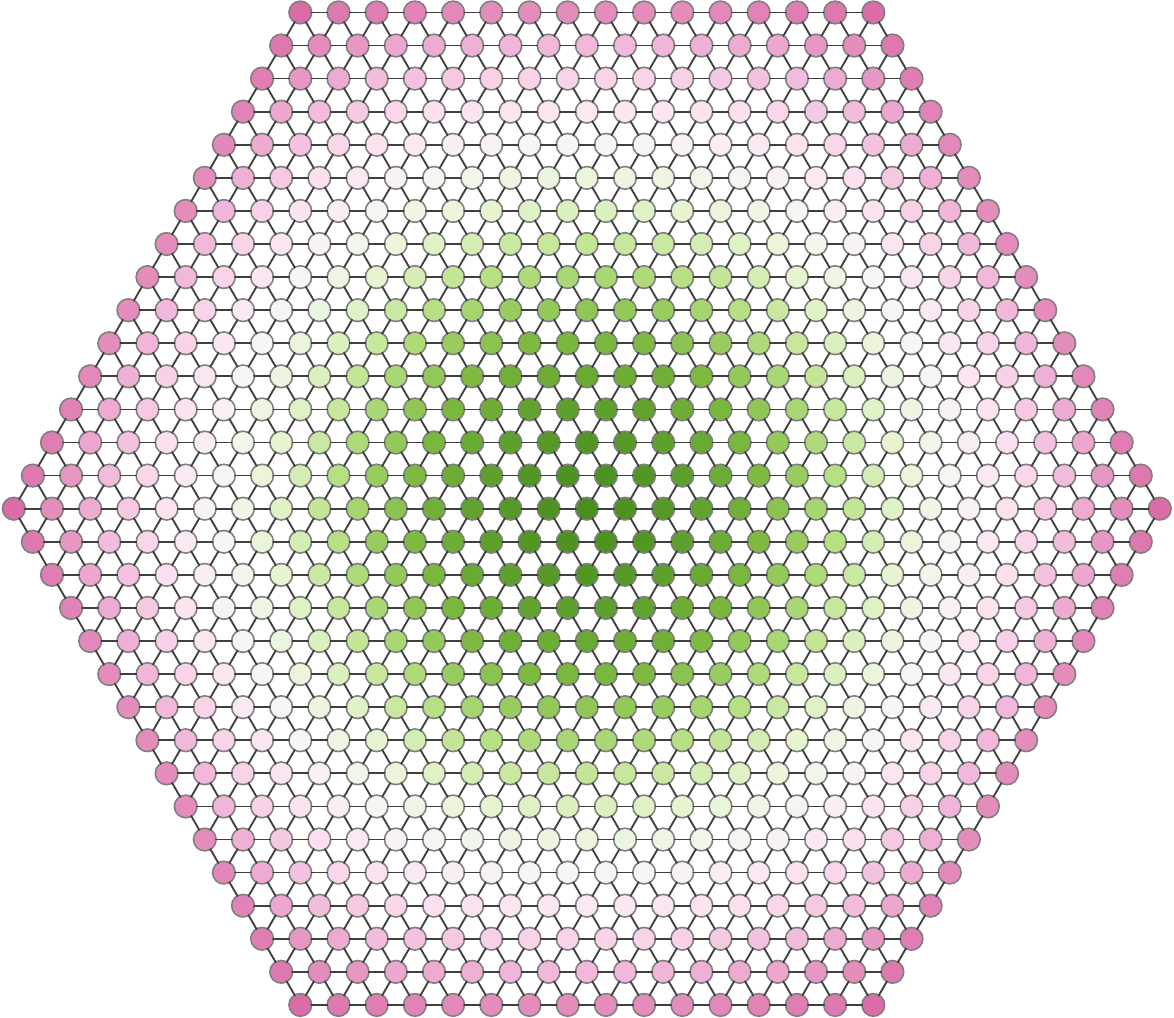}};
\node at (3,4) {\includegraphics[width=1.1in]{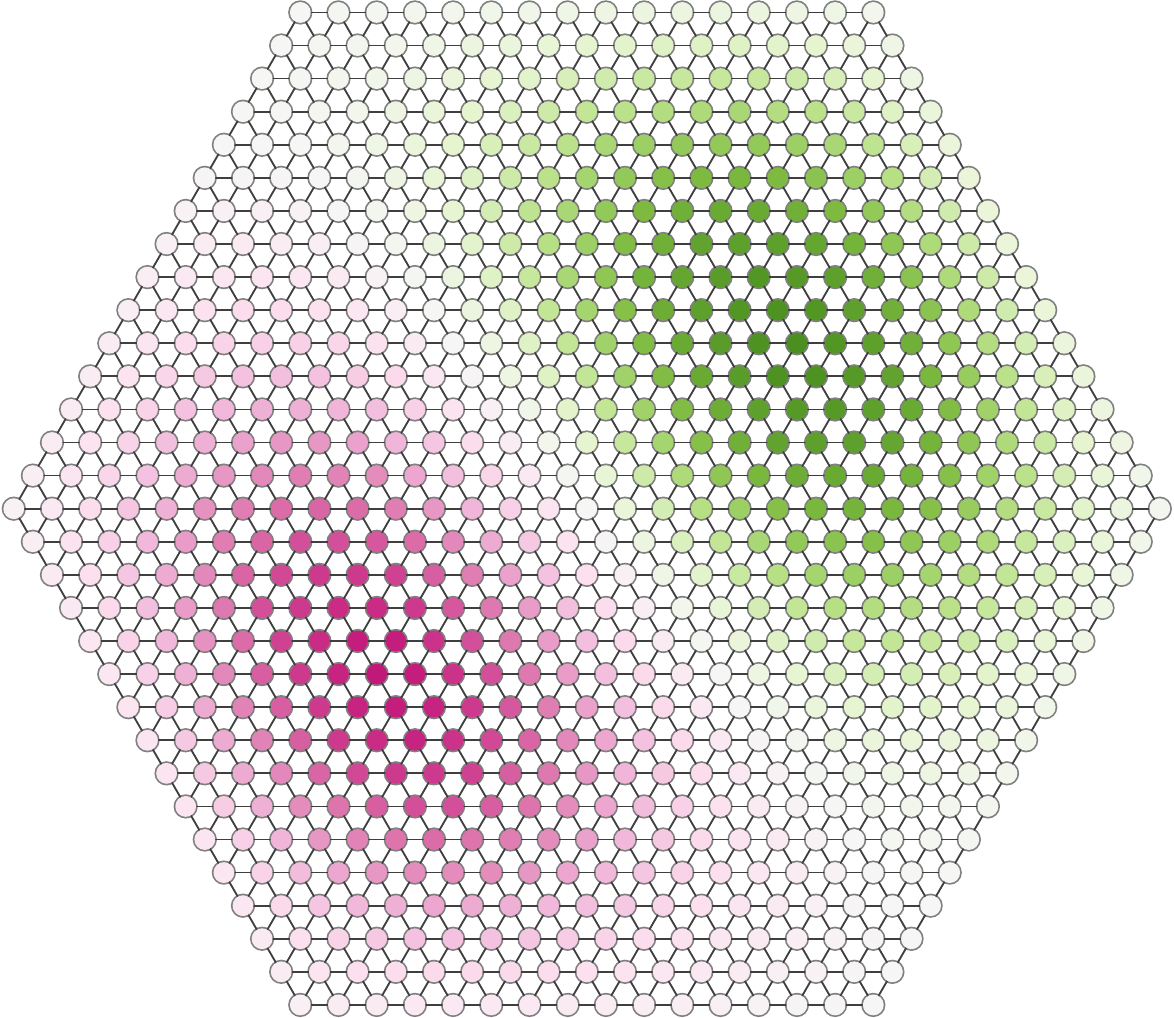}};
\node at (6,4) {\includegraphics[width=1.1in]{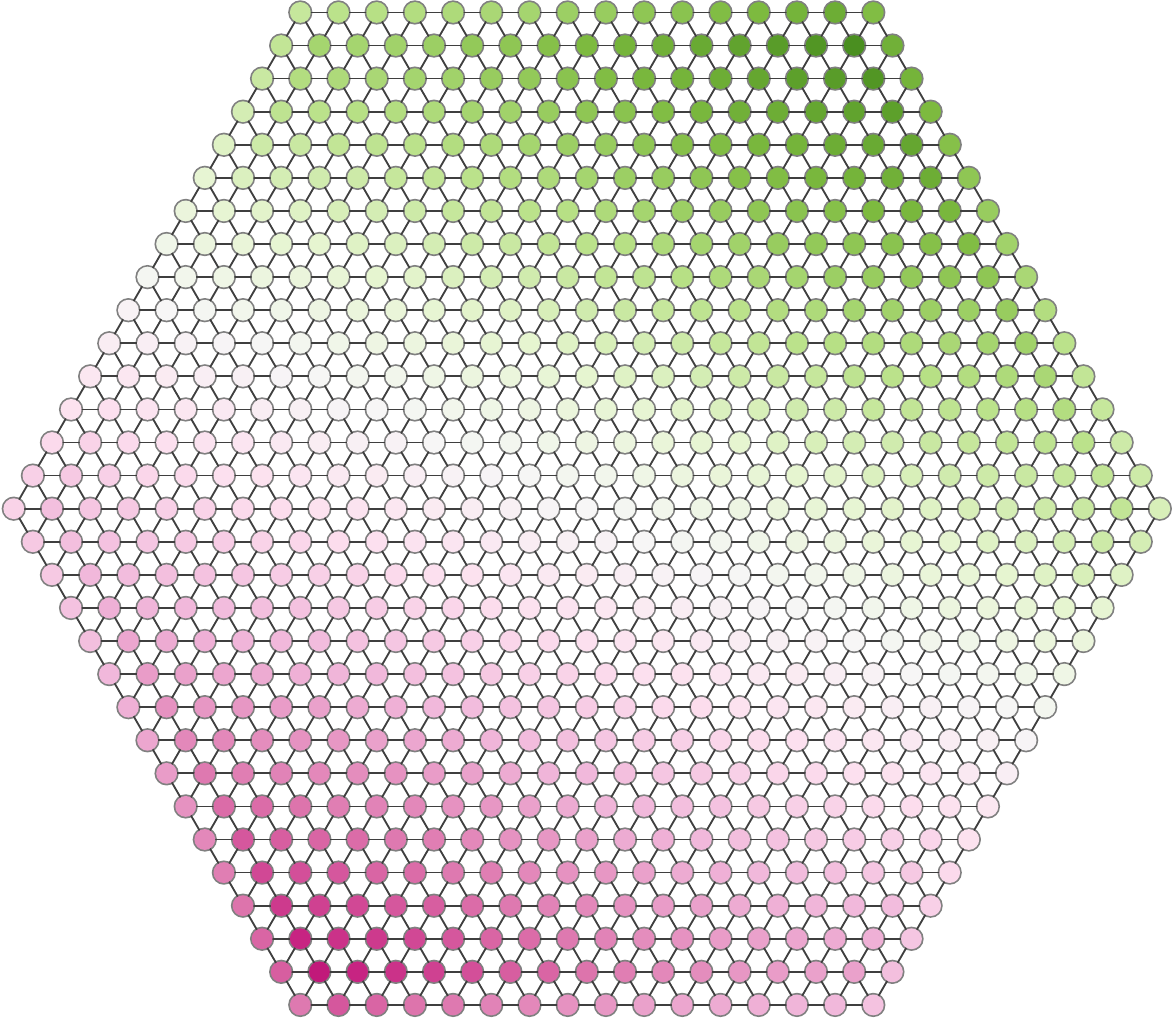}};
\node at (9,4) {\includegraphics[width=1.1in]{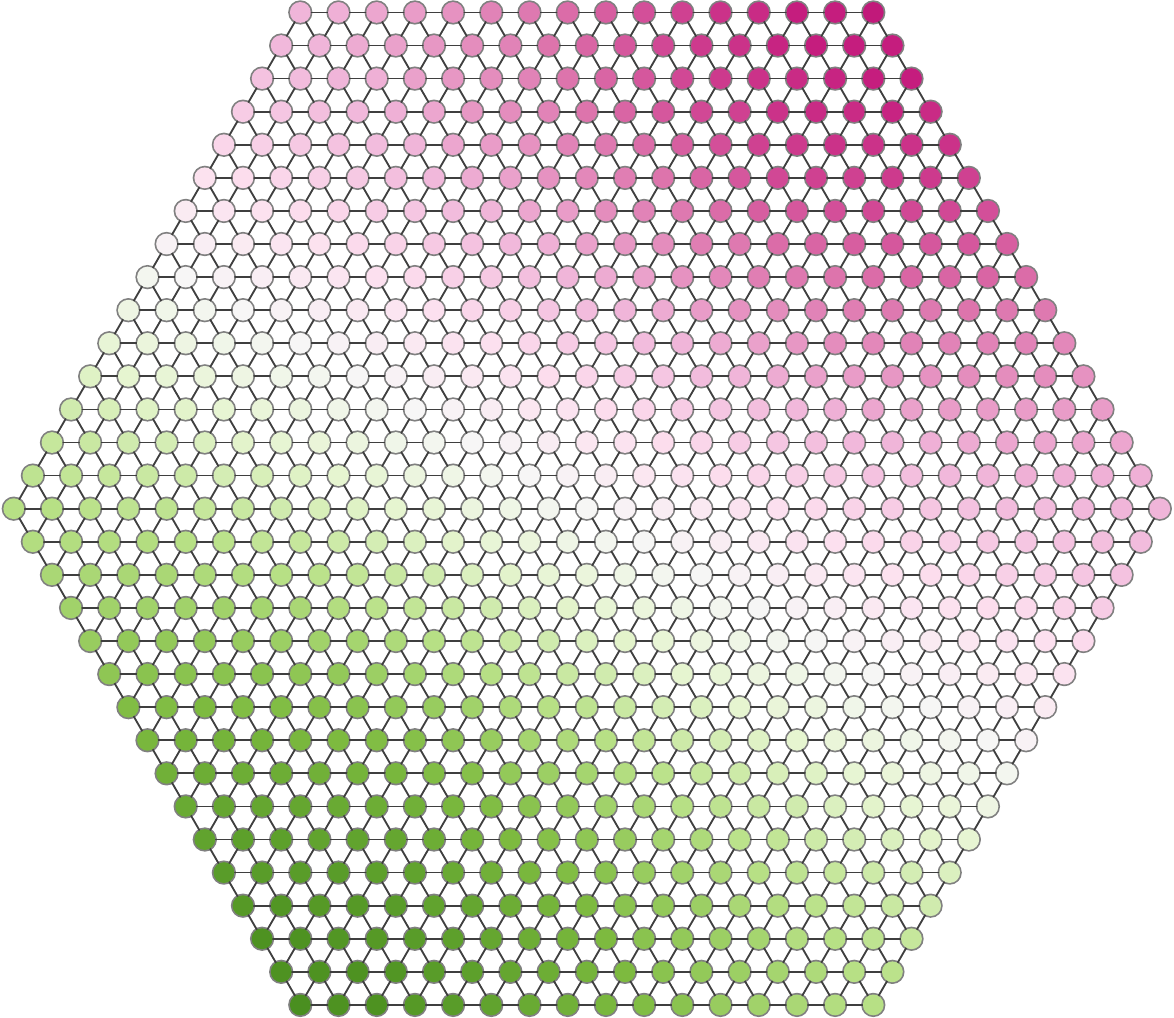}};
\node at (12,4) {\includegraphics[width=1.1in]{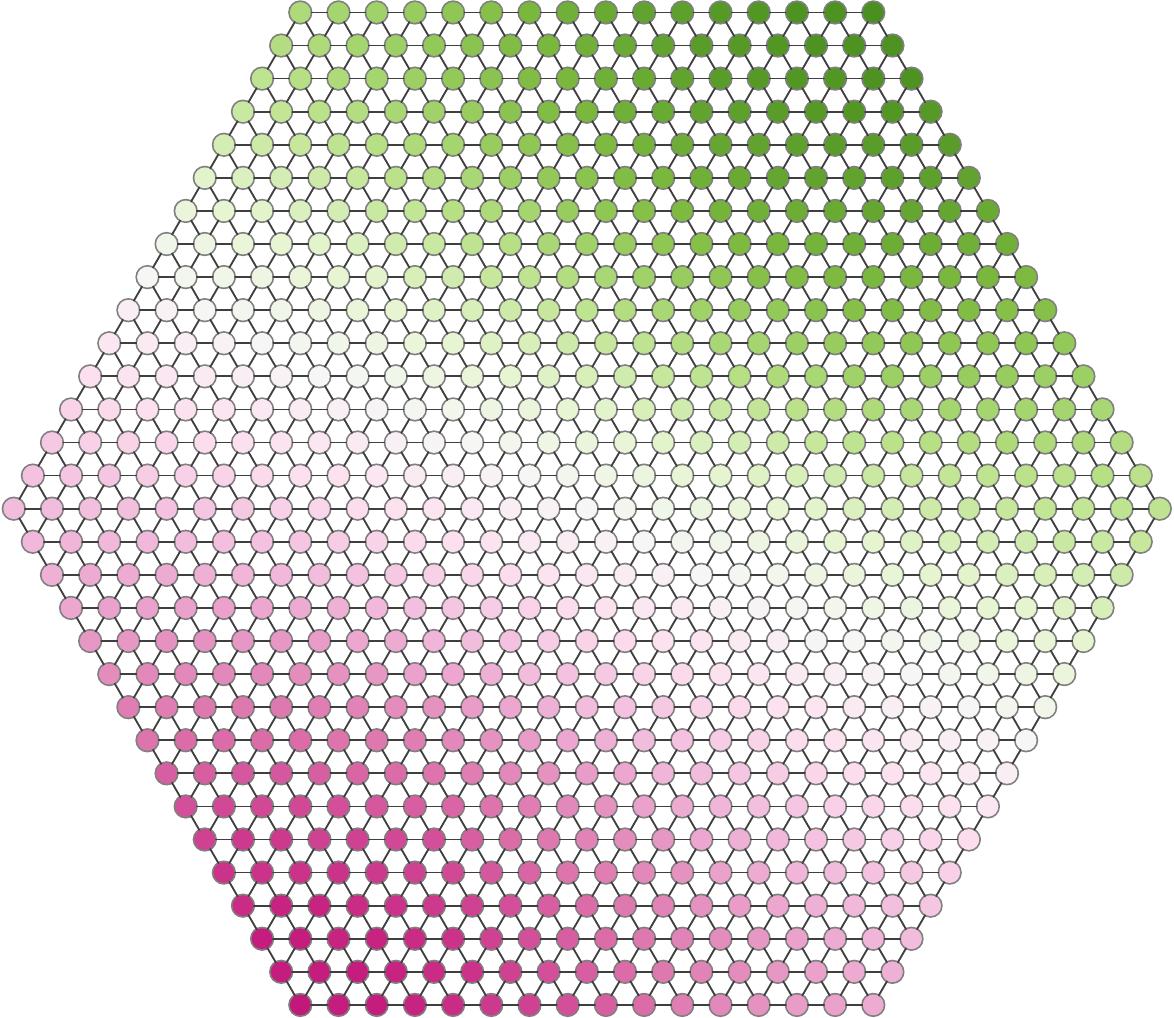}};

\node at (0,0) {\includegraphics[width=1.1in]{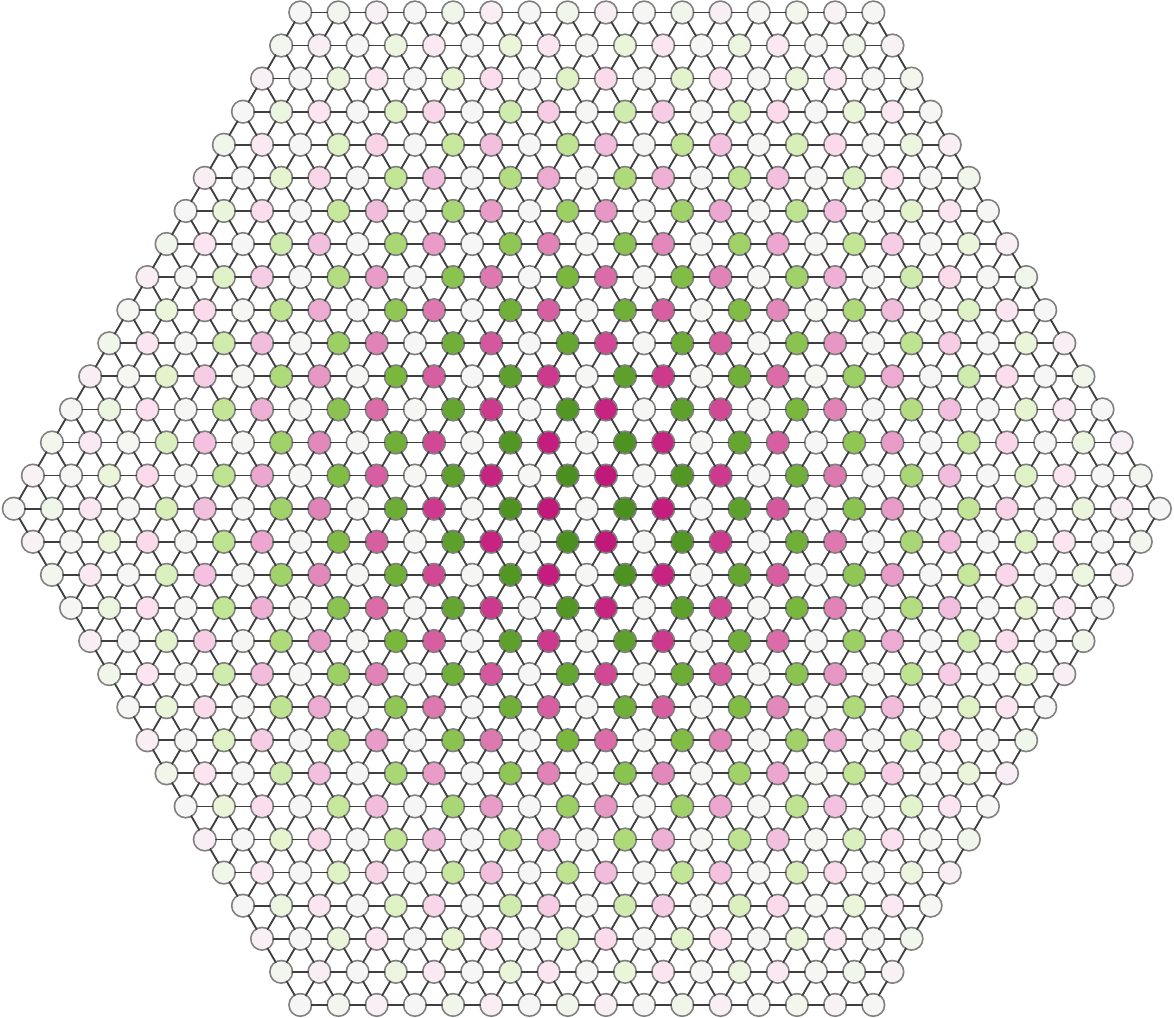}};
\node at (3,0) {\includegraphics[width=1.1in]{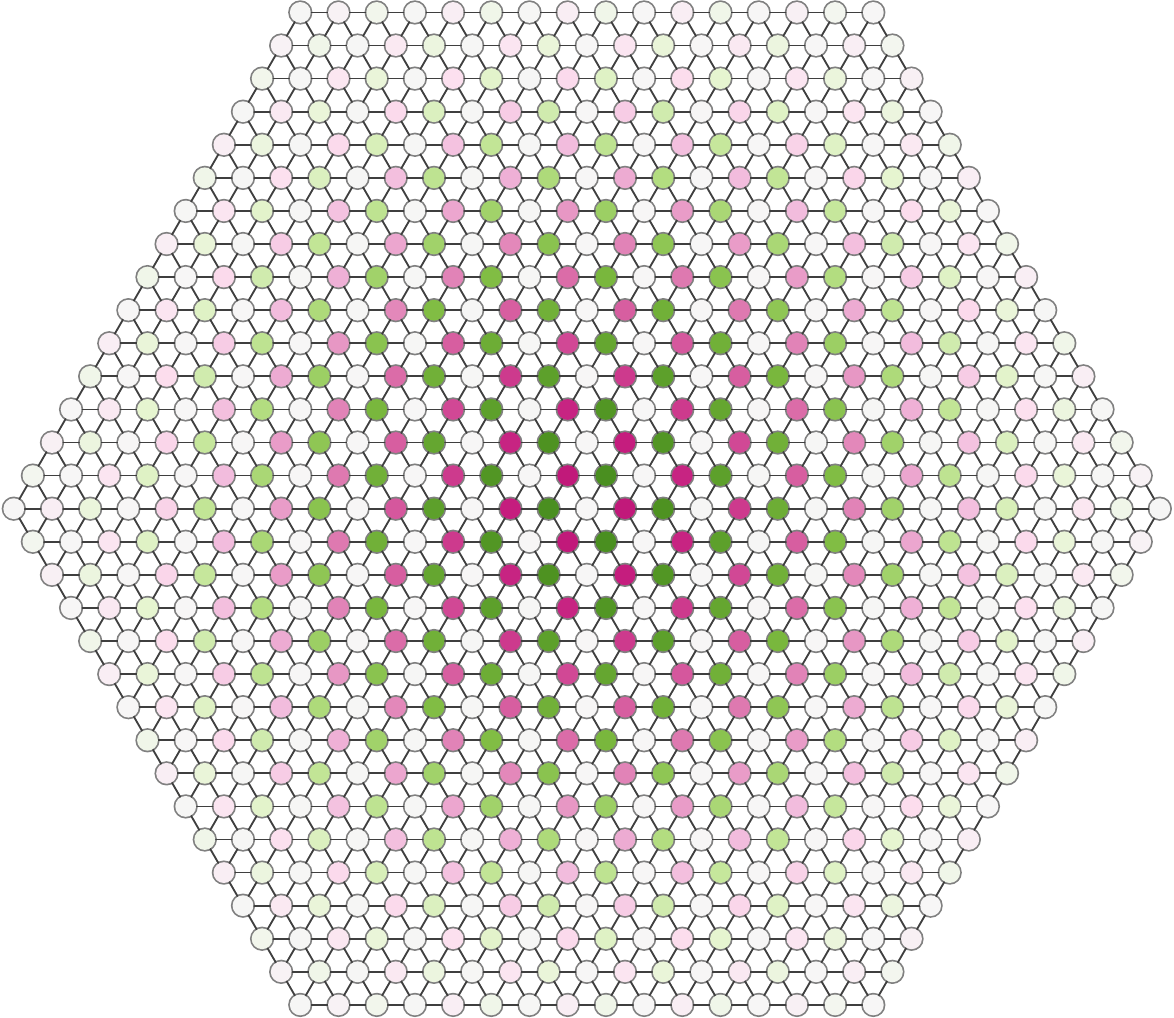}};
\node at (6,0) {\includegraphics[width=1.1in]{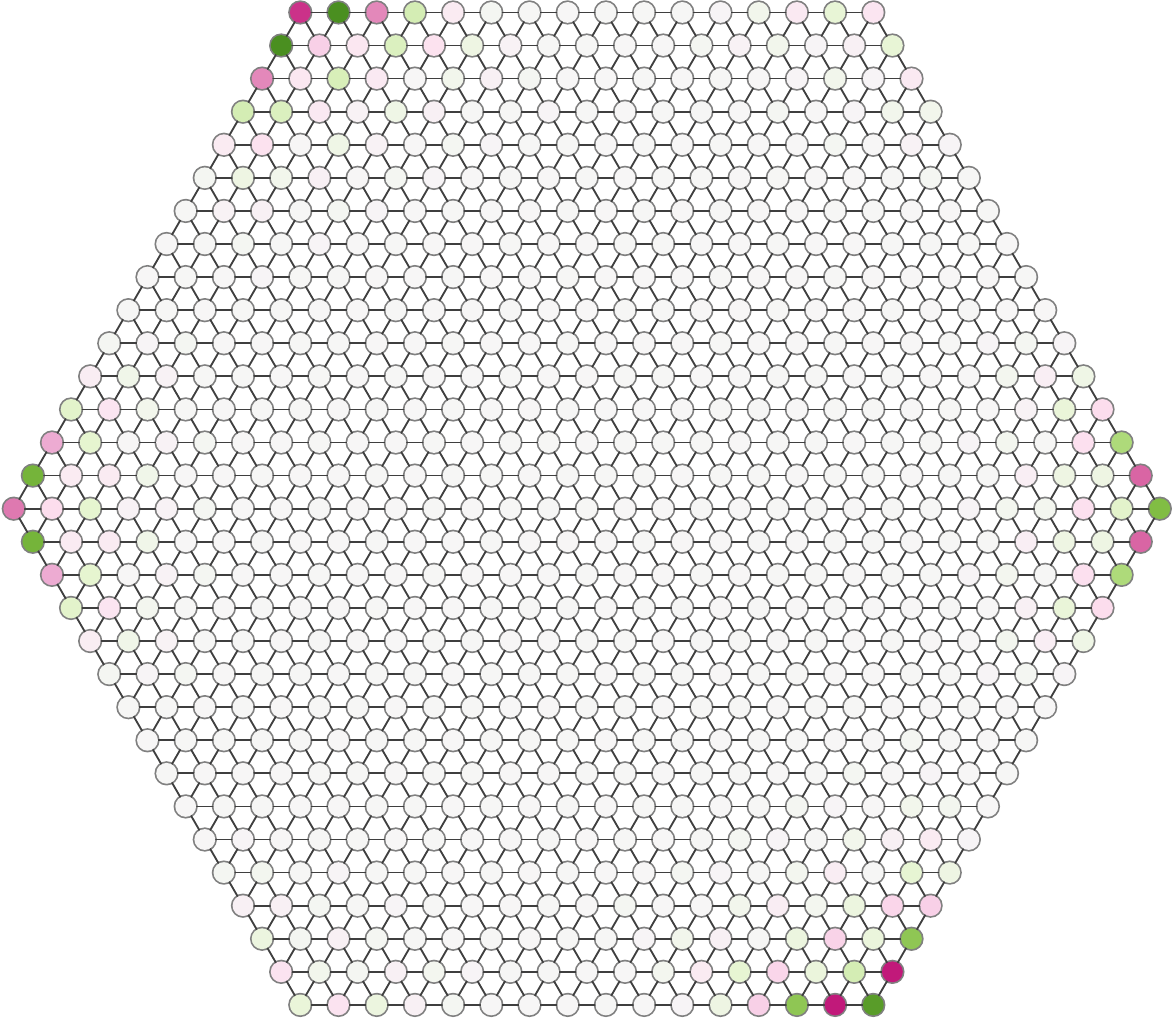}};
\node at (9,0) {\includegraphics[width=1.1in]{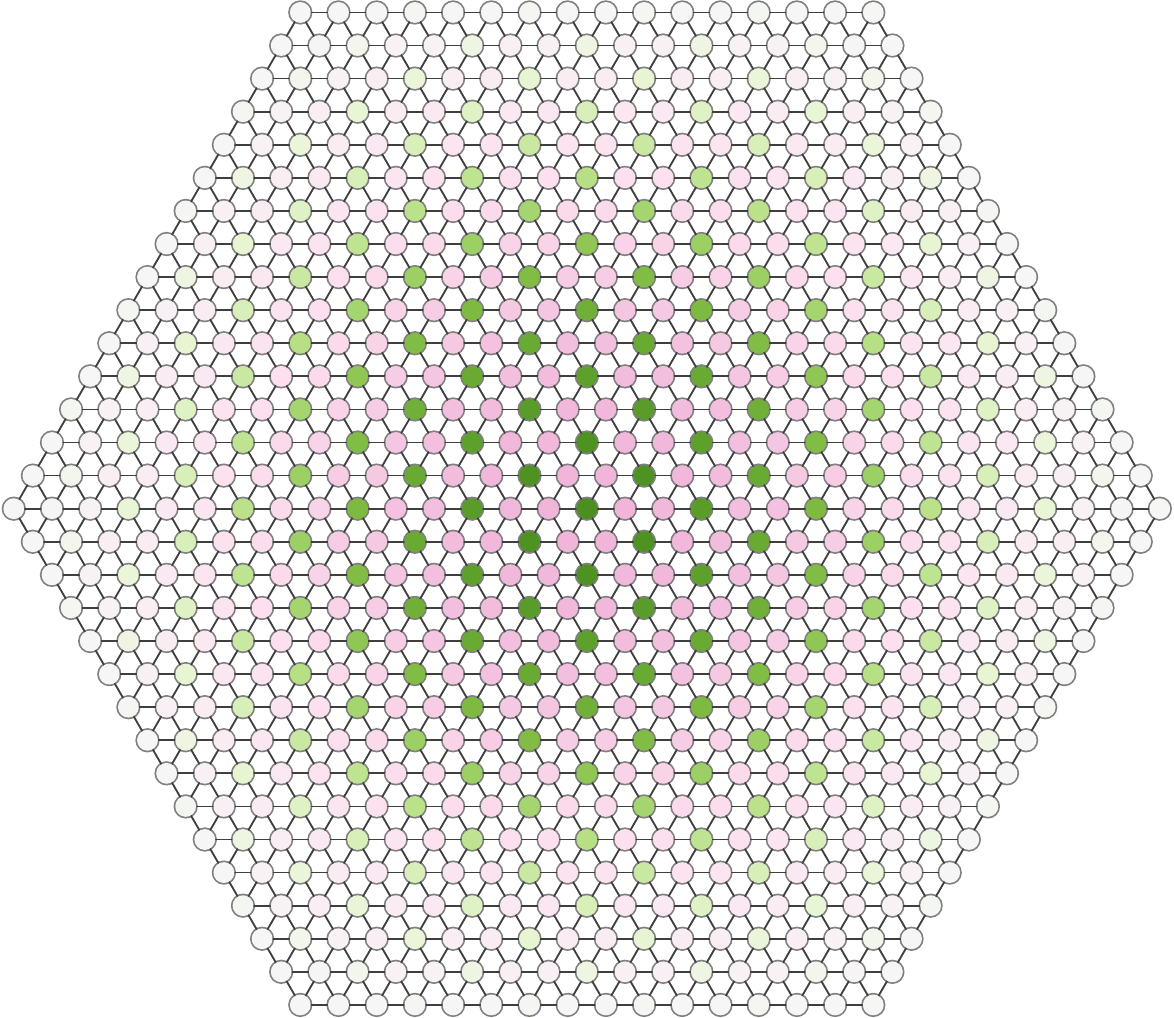}};
\node at (12,0) {\includegraphics[width=1.1in]{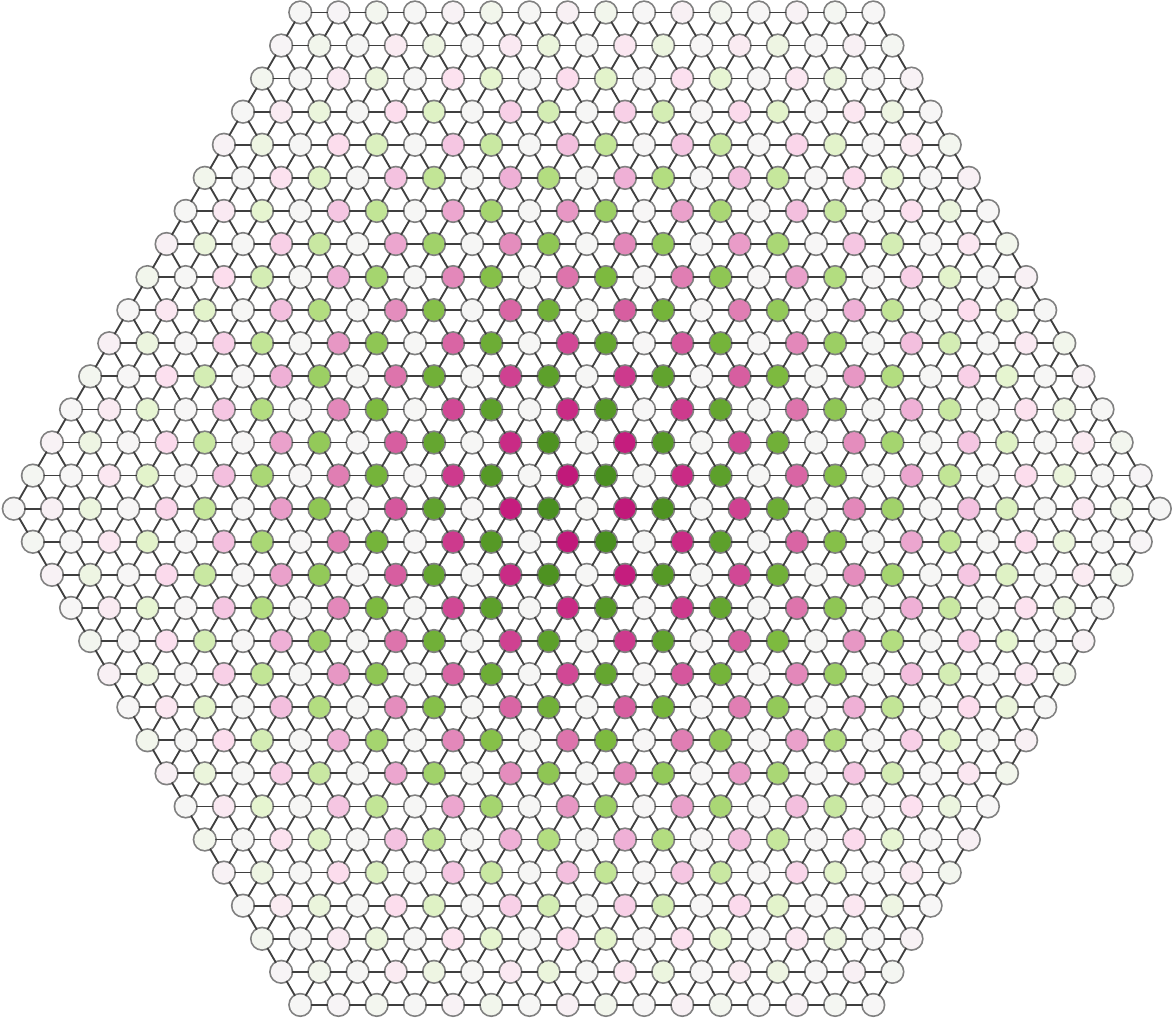}};

\end{tikzpicture}

    \caption{These hexagonal graphs are 6-regular except on the boundary.  Depending on the spatial weight matrix, the extremizers can differ, particularly in the extent to which the lower-degree vertices along the boundary are reflected in the pattern. (The classical eigenvectors of $A$ are shown to the left for comparison.)  In particular, the minimizer of $\I(\cdot \ ; L)$ and maximizer of $\I(\cdot \ ; \MH)$ are less impacted by low degree nodes than the maximizer of $\I(\cdot \ ; A)$.}
    \label{fig:HexagonalLattice}
\end{figure}
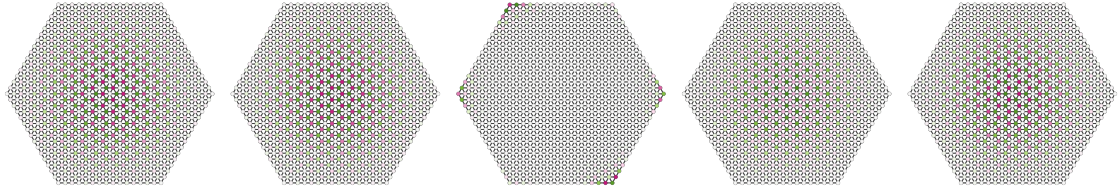

The leftmost pair of plots in Figure~\ref{fig:HexagonalLattice} shows the extreme eigenvectors of $A$, and we see that the lower degrees on the boundary have a visible effect on the pattern.  Passing to the true $A$ extremizer (via the generalized eigenvector) or row-normalizing to obtain $P$ might be thought to fix the degree effects, but this particularly fails with the minimizer of $P$.  It is also interesting to note that the maximizer for $P$ takes its strongest values on the second rung---adjacent to the vertices of lowest degree. 
Note that $M$ and $L$ both give the expected clustered configuration on one end of the \I range, but give two interestingly different approximations to checkerboards on the other.

\begin{figure}[ht] \centering
\begin{tikzpicture}[yscale=.9]
\begin{scope}[xscale=.5,yscale=2.5]
\draw (-.1,-1.2) rectangle (4.1,1.2);
\node at (2,1.4) {Hexagonal lattice};
\draw [rounded corners,denim,fill=denim,opacity=.7] (0.2, -0.53353) rectangle (.8,1.0211);
\draw [rounded corners,denim,fill=amber,opacity=.9] (1.2,-0.5423) rectangle (1.8,0.98617);
\draw [rounded corners,denim,fill=cherryblossompink,opacity=.9] (2.2,0.00096) rectangle (2.8,0.81403);
\draw [rounded corners,denim,fill=lust,opacity=.9] (3.2,-0.48496) rectangle (3.8,0.98181);
\foreach \x/\y in {-1/-1,-.5/,0/0,.5/,1/1}
{\draw [gray!50,thin] (-1,\x)--(23,\x);
\node at (-1,\x) [left] {$\y$};
}
\end{scope}

\begin{scope}[xscale=.5,yscale=2.5,xshift=6cm]
\draw (-.1,-1.2) rectangle (4.1,1.2);
\node at (2,1.4) {10\% deletions};
\draw [rounded corners,denim,fill=denim,opacity=.7] (0.2,-0.62855) rectangle (.8,1.0689);
\draw [rounded corners,denim,fill=amber,opacity=.9] (1.2,-0.73753) rectangle (1.8,0.9952);
\draw [rounded corners,denim,fill=cherryblossompink,opacity=.9] (2.2,0.0087107) rectangle (2.8,0.88445);
\draw [rounded corners,denim,fill=lust,opacity=.9] (3.2,-0.51705) rectangle (3.8,0.98469);
\end{scope}

\begin{scope}[xscale=.5,yscale=2.5,xshift=12cm]
\draw (-.1,-1.2) rectangle (4.1,1.2);
\node at (2,1.4) {20\% deletions};
\draw [rounded corners,denim,fill=denim,opacity=.7] (0.2,-0.71499) rectangle (.8,1.1186);
\draw [rounded corners,denim,fill=amber,opacity=.9] (1.2,-0.83636) rectangle (1.8,1.0068);
\draw [rounded corners,denim,fill=cherryblossompink,opacity=.9] (2.2,0.0076962) rectangle (2.8,0.96597);
\draw [rounded corners,denim,fill=lust,opacity=.9] (3.2,-0.58247) rectangle (3.8,0.98715);
\end{scope}

\begin{scope}[xscale=.5,yscale=2.5,xshift=18cm]
\draw (-.1,-1.2) rectangle (4.1,1.2);
\node at (2,1.4) {Square lattice};
\draw [rounded corners,denim,fill=denim,opacity=.7] (0.2,-1.0562) rectangle (.8,1.0161);
\draw [rounded corners,denim,fill=amber,opacity=.9] (1.2,-1.0021) rectangle (1.8,0.98672);
\draw [rounded corners,denim,fill=cherryblossompink,opacity=.9] (2.2,0.0078699) rectangle (2.8,1.0676);
\draw [rounded corners,denim,fill=lust,opacity=.9] (3.2,-0.9778) rectangle (3.8,0.98475);
\end{scope}

\begin{scope}[xshift=2cm,yshift=-5cm]
\draw [denim,fill=denim, opacity=.7] 
(0,0) rectangle  node [black,opacity=1,above=.6cm] {$A$} (1,1);

\draw [denim,fill=amber, opacity=.9] 
(2,0) rectangle node [black,opacity=1,above=.6cm] {$P$}(3,1);

\draw [denim,fill=cherryblossompink, opacity=.9] 
(4,0) rectangle node [black,opacity=1,above=.6cm] {$L$} (5,1);

\draw [denim,fill=lust, opacity=.9] 
(6,0) rectangle  node [black,opacity=1,above=.6cm] {$\MH$} (7,1);
\end{scope}
\end{tikzpicture}
\caption{Ranges of possible $\I(\cdot \ ; W)$ values for $W=A,P,L,\MH$.  We compare graphs on 169 nodes formed within the hexagonal lattice and the square lattice, and we use random edge deletions to interpolate between them and introduce degree variation, as explained in the text.  We can observe that $L$ gives non-negative $\I$ values, and that only $\MH$ gives values always between $-1$ and $1$.}\label{fig:ranges}
\end{figure}
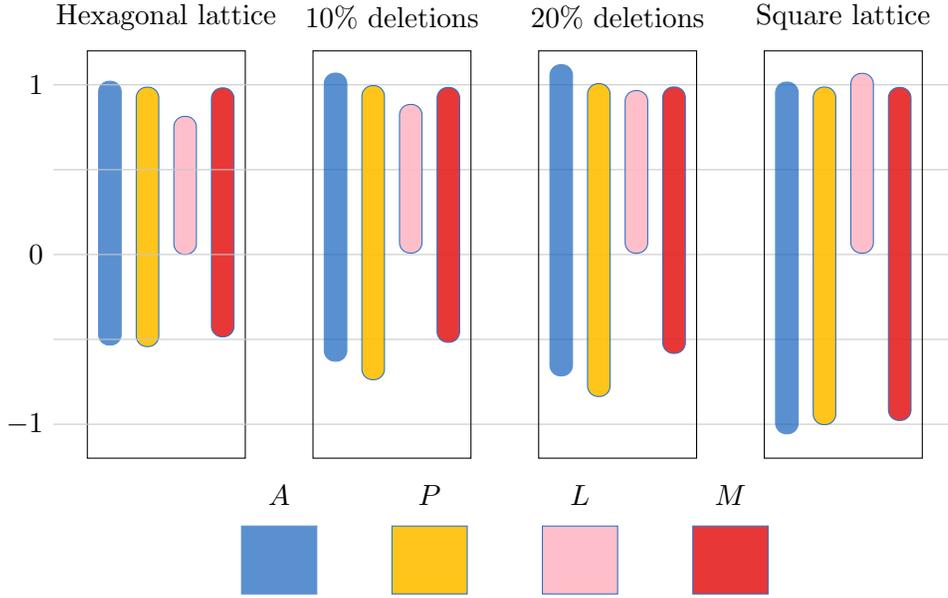

Next, Figure~\ref{fig:ranges} shows the numerically computed ranges of achievable $\I$ for $W=A,P,L,\MH$ on various planar graphs.  First, we take a subset of the hexagonal (triangular) lattice, formed as a hexagon with 8 vertices along each side (a smaller version of the graph in Figure~\ref{fig:HexagonalLattice}). From the square lattice, we use a square with 13 vertices along a side, so that it agrees with the hexagon in having 169 vertices.  Note that the hexagonal lattice can also be viewed, combinatorially, as the square lattice with diagonals added.  Therefore we can produce graphs that in a sense interpolate between the two lattices by deleting edges from the hexagonal lattice.  We choose to do so at random in order to also introduce more variation in vertex degree.\footnote{To be precise, we order the edges, then randomly select 10\% (or 20\%, respectively) for deletion, rejecting the final product if it is disconnected.
Over 1000 successful trials, we then report the average of the numerical minimizer and of the numerical maximizer of \I.  See supplementary materials for more information.}

This example makes it clear that even quite ``reasonable" planar graphs, when they are irregular, can realize $\I$ values outside of $[-1,1]$.  In the case when the graph is nearly regular bipartite (such as a large square grid graph), maximal and minimal values of Moran's \I can be seen to converge to $\pm 1$, respectively, as the number of nodes in the graph increases and the variance in degree converges to 0. On the other hand, one can construct graphs that realize arbitrarily large and small  \I as the degree disparity gets large.

To see this, consider a double-star graph (Figure~\ref{fig:TwinStarGraph}) where each hub is connected to $n$ leaves.  Consider a function $\V_{a}$ that takes the values $\pm 1$ on the hubs and $\pm 1/a$ on the leaves; note $\V_{a}$ has average value 0. For arbitrary fixed $a\neq 0$, as $n$ gets large, the average product across an edge is nearly $1/a$ while the average squared value at a vertex is nearly $1/a^2$.  This means that $\I(\V_{a};A)\to a$ as $n\to \infty$ (by Remark~\ref{rmk:pairs-v-single}).  This construction works for both positive and negative values of $a$.  Interestingly, however, putting $0$ on the leaves (denoted by $\V_{\infty}$) gives different limiting behavior, with $\I(\V_{\infty};A)\to -1/4$ as $n\rightarrow\infty$.  Normalizing $A$ to be row-stochastic does not solve this problem; if we use $P$, we get $\I(\V_{a};P)\to a$ as $n\rightarrow\infty$, while putting 0 on the leaves gives $\I(\V_{\infty};P)\to 0$.

\begin{figure}[ht]    \centering
\begin{tikzpicture}[scale=2]
\draw (0,0)--(1,0);
\draw [fill=black] (0,0) circle (0.03) node [above right] {$-1$};
\draw [fill=black] (1,0) circle (0.03) node [above left] {$+1$};
\node at (1.8,1) {$1/a$};
\node at (-.8,1) {$-1/a$};

 \foreach \x in {90,100,...,270}
{\draw [fill=black] (\x:1) circle (0.03);
\draw (\x:1)--(0,0);
}
\begin{scope}[xshift=1cm]
 \foreach \x in {90,80,...,-90}
{\draw [fill=black] (\x:1) circle (0.03);
\draw (\x:1)--(0,0);
}
\end{scope}

\node at (4.3,.9) {limit as $n\to\infty$};
\node at (4.3,0) 
{
$\begin{array}{cccc} \I(\V_a; A)& \I(\V_a;P)&\I(\V_a; L)&\I(\V_a; \MH)\\
a&a&\frac 14 (a-1)^2&1\\
&&&\\
 \I(\V_\infty; A)& \I(\V_\infty; P)&\I(\V_\infty; L)&\I(\V_\infty;\MH)\\
-\frac{1}{2} &0&\infty &0\\
\end{array}$
};
\end{tikzpicture}
\caption{The family of double-star graphs is depicted here to emphasize the major impact of having very uneven vertex degrees in a graph.  
Letting $\V_a$ take the value $\pm 1/a$ on leaves and $\pm 1$ on hubs, as indicated on the graph above, gives $\I(\V_a; A)$ values approaching any arbitrary $a$ as the number of leaves $n\to\infty$. 
In particular, this illustrates that all real values of $\I(\V;A)$ are achievable for $W=A,P$ and all non-negative values are achievable for $W=L$.
That is, passing to row-normalized $P$ does not mitigate the degree effect, and the problem is even more pronounced for the Laplacian (see \S\ref{sec:GraphFourierAnalysis}). Only the use of the doubly-stochastic approximation $\MH$ keeps $\I$ bounded (see \S\ref{sec:WalksandI}).}
    \label{fig:TwinStarGraph}
\end{figure}
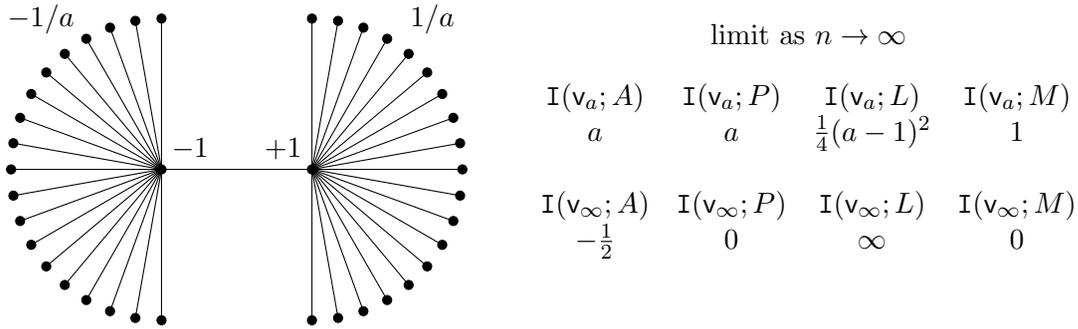

This double-star example is designed to exaggerate the phenomenon that causes \I to explode, but degree effects of this kind are reflected in the real-world examples below: when there are adjacent nodes of relatively high  vertex degree, extreme values of $\I$ can be obtained by placing positive and negative values on those nodes, and near-zero values everywhere else (see Figure~\ref{fig:LowAndHighFrequency}, lower right).

\FloatBarrier

\subsection{Correlation on Realistic Examples}

Next, we confirm that, despite significant differences observable in theory, the choices of $W$ give  outputs that are fairly tightly correlated on real-world 
dual graphs $\G$ and population functions $\V$.
In particular, we consider the spatial weight matrices $A$, $L$ and $\MH$ applied to Black and Hispanic population in the census tracts of all fifty states. Since the underlying graphs are not regular, we know of no theoretical relationship between the \I values when we change the matrix $W$. Despite this, Figure~\ref{fig:IvsI} shows strong correlations. The supplemental material contains a more extensive pairwise comparison among all four choices of weight matrix.\footnote{Census tract graphs for all states (based on 2010 census geography) were obtained from \cite{daryldatabase}. For tracts with zero population, we define $\V_i$ using the average population values (Hispanic, Black and total) of the neighboring tracts.}

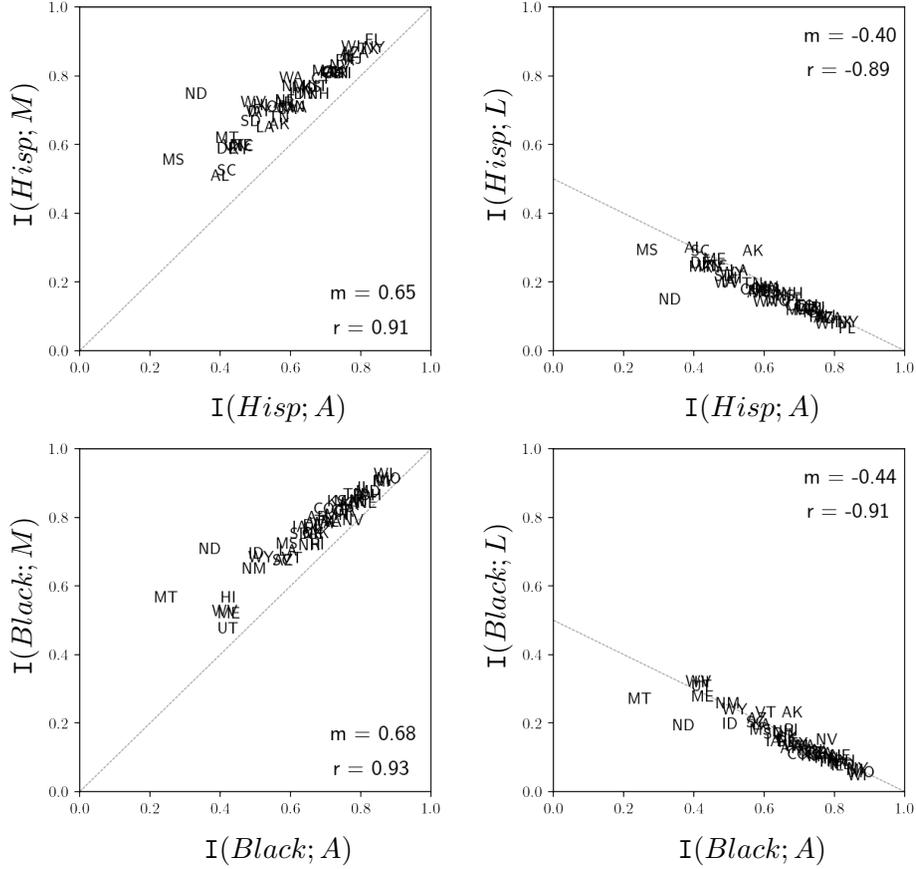
\begin{figure}[ht]
\centering
\begin{tikzpicture}[scale=.9]
\begin{scope}[xshift=0cm,yshift=6.5cm]
\node at (.5,-3.25) {$\I(Hisp ; A)$};
\node at (-3.25,.5) [rotate=90] {$\I(Hisp ; \MH)$};
\node at (0,0) {\includegraphics[width=5.5cm]{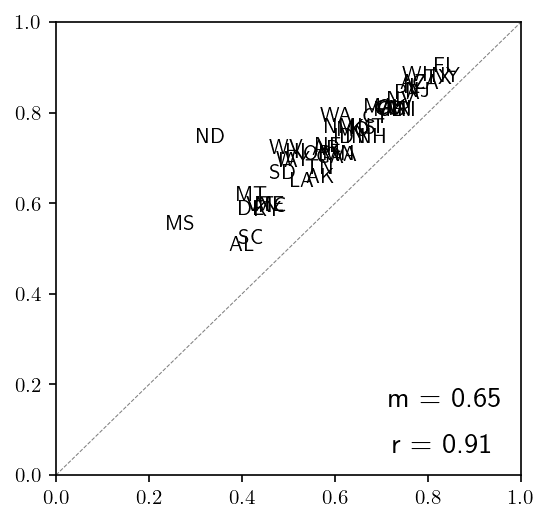}};
\end{scope}
\begin{scope}[xshift=7cm,yshift=6.5cm]
\node at (.5,-3.25) {$\I(Hisp ; A)$};
\node at (-3.25,.5) [rotate=90] {$\I(Hisp ; L)$};
\node at (0,0) {\includegraphics[width=5.5cm]{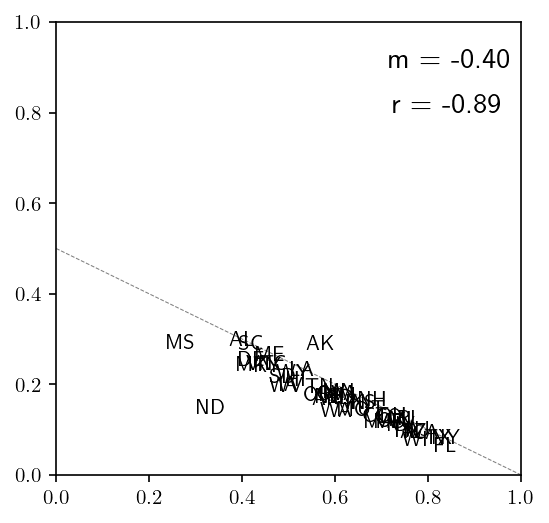}};
\end{scope}

\begin{scope}[xshift=0cm,yshift=0cm]
\node at (.5,-3.25) {$\I(Black ; A)$};
\node at (-3.25,.5) [rotate=90] {$\I(Black ; \MH)$};
\node at (0,0) {\includegraphics[width=5.5cm]{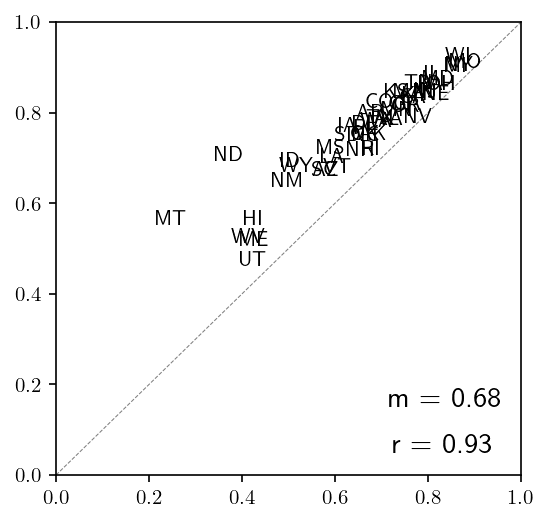}};
\end{scope}
\begin{scope}[xshift=7cm,yshift=0cm]
\node at (.5,-3.25) {$\I(Black ; A)$};
\node at (-3.25,.5) [rotate=90] {$\I(Black ; L)$};
\node at (0,0) {\includegraphics[width=5.5cm]{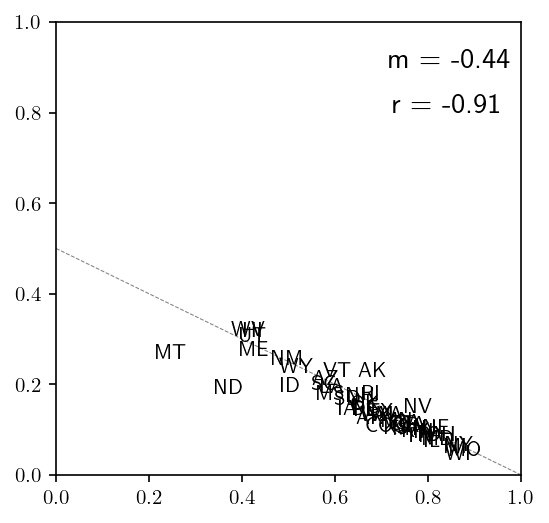}};
\end{scope}

\end{tikzpicture}
    \caption{Comparing $\I$ for the matrices $A,L,\MH$ on 50-state census tract data for Hispanic and Black population shares. The line $y=x$ is shown for the $A$ versus $\MH$ comparison, and the line $y=(1-x)/2$ is shown for comparing $A$ to $L$, because the data points would have to fall on these lines if the graphs were regular. The correlation $r$ and slope of the best fit line $m$ are reported for each plot, but those fit lines are not plotted; their slopes are not equal to $m=1$ and $m=-1/2$, indicating different ways of handling degree disparity.}
    \label{fig:IvsI}
\end{figure}

It is interesting to look at states for which $\I(\V;A)$ and $\I(\V;\MH)$ differ substantially. In order to localize the source of the disparity, we employ Anselin's definition of local Moran's \I at the $i$th vertex \cite{anselin1995local}, which looks just like the standard definition except the numerator only considers neighbors of $i$.
\[
\I_i(\V; W) := \left(n \displaystyle\sum_{j=1}^{n} W_{ij} (v_i - \vbar)(v_j -\vbar)\right) \bigg /\left(w \displaystyle\sum_{j=1}^{n}(v_j-\vbar)^2\right).
\]
We can then study the contribution of tract $i$ towards the difference by defining
$D_i(\V; A, \MH) := | \I_i(\V; A)  - \I_i(\V; \MH)|$,
motivated by the fact that $\I(\V; A)$ and $\I(\V; \MH)$ agree for regular graphs.
Nodes with much higher $\V$ values than their neighbors will tend to have high $D_i(\V; A, \MH)$ values since $A$ does not have diagonal entries, while $\MH$ does. Also, nodes with low $A$-degrees tend to have higher diagonal entries in $\MH$, thus higher $D_i$.  In Figure~\ref{fig:IvsI} a handful of states---North Dakota, Montana and Mississippi, especially---stand out as having a large discrepancy between $\I(\V; A)$ and $\I(\V; \MH)$ for one or both $\V$. When we localize to the tracts that contribute most to the disparity (Table~\ref{tab:military}), we find, as expected, nodes with low $A$-degree (typically 1 or 2), with concentrations of the minority group that are typically 5-10 times greater than the share in the state overall, and than the share in the neighboring nodes.  

\begin{table}[ht!]
    \centering
%    \begin{adjustbox}{angle=270}
\resizebox{\textwidth}{!}{
    \begin{tabular}{|c|c|c|c|l|c|c|c|c|}
    \hline
        State & $\V$ & $\vbar$ & avg $D_i$ & Tracts with highest $D_i$ & $\V_i$ & $(P\V)_i$ & $d_i$ & $D_i$  \\
        \hline
        ND & Hisp & 0.020 & 0.0004 & Grafton & 0.214 & 0.100 & 2 & 0.024 \\
        & & & & Minot Air Force Base & 0.099 & 0.012 & 1 & 0.009 \\
        & & & & Grand Forks Air Force Base & 0.100 & 0.024 & 2 & 0.006 \\
        \hline
        MS & Hisp & 0.028 & 0.0006 & Morton & 0.269 & 0.049 & 1 & 0.061 \\
        & & & & Forest & 0.268 & 0.044 & 2 & 0.049 \\
        & & & & Key Field Air National Guard Base & 0.238 & 0.025 & 2 & 0.038 \\
         \hline
        ND & Black & 0.009 & 0.0001 & Minot Air Force Base & 0.096 & 0.004 & 1 & 0.008 \\
        & & & & Grand Forks Air Force Base & 0.095 & 0.006 & 2 & 0.006 \\
        & & & & Fargo & 0.053 & 0.0013 & 9 & 0.001 \\
                \hline
        MT & Black & 0.004 & 1.85e-5 & Malmstrom Air Force Base & 0.086 & 0.015 & 4 & 0.0025 \\
        & & & & Crossroads Correctional Center & 0.026 & 0.003 & 1 & 0.0006 \\
        & & & & Yellowstone National Park & 0.030 & 0.002 & 2 & 0.0005 \\
        \hline
    \end{tabular}
    }
%    \end{adjustbox}
    \caption{Examining the tracts that contribute most to differences between $\I(\V; A)$ and $\I(\V; M)$. We identify the top three $D_i(\V;A,M)$ values and for those tracts we report the  share of Hispanic or Black population ($\V_i$), the average share in neighboring tracts ($(P\V)_i$) and the vertex degree $d_i$. For each state, we also report the average $\V$ value and the average value of $D_i$.}\label{tab:military}
    \label{tab:outliers}
\end{table}

The definition of $D_i(\V; A, \MH)$ was motivated by the fact that $\I(\V; A)$ and $\I(\V; \MH)$ agree for regular graphs. We could similarly compare $\I(\V; A)$ and $\I(\V; L)$ using the theoretical relationship for regular graphs given in Remark \ref{rem:AvL}. We note that ND, MS and MT remain noticeable outliers in the $A$ vs. $L$ comparisons in Figure~\ref{fig:IvsI}.

\FloatBarrier

\section{Laplacian Weights}
\label{sec:GraphFourierAnalysis}
As noted above, the Laplacian $L$ is a very natural choice for matrix-based analysis of a network. The Laplacian has been closely connected with the topic of {\em community detection} in networks, especially when potential communities are of different sizes \cite{newman2018networks,Von2007_Tutorial}.  The Laplacian also has a rich theoretical interpretation in terms of relaxed graph cuts---which lends itself well to measurements of clustering---as well as to notions of smoothness on graphs via Dirichlet energy functionals. 
Let $\Tor$ be the $m$-dimensional torus ($[0,2\pi]^{m}$ with opposite boundary faces identified) and recall the $L^2$ inner product of $f,g\in L^{2}(\Tor)$, defined as $\langle f,g\rangle:=\displaystyle\frac{1}{(2\pi)^m}\int_{\Tor}f(x)\overline{g(x)}dx$. 

\subsection{High- and Low-Frequency Eigenvectors}

By construction, $L$ is symmetric and positive semidefinite, and therefore has a basis of orthonormal eigenvectors $\{\Psi_{i}\}_{i=1}^{n}$ with associated real eigenvalues $0=\mu_1\le \dots \le \mu_n$.  Since $L\one=\zero$, we have $\Psi_{1}=\frac{1}{\sqrt{n}}\one, \ \mu_{1}=0$.  These may be interpreted in the framework of Fourier analysis, in which the eigenvectors $\{\Psi_{i}\}_{i=1}^{n}$ of $L$ are the Fourier modes on the graph  $\G$ with frequencies $\{\mu_{i}\}_{i=1}^{n}\subsetneq [0,\infty)$.  In classical Fourier analysis on $\Torone$, the Laplacian operator $\Lap: f\mapsto-\Delta f= -\nabla \cdot \nabla f = -\frac{d^{2}}{d x^{2}}f$ has eigenfunctions $\{\exp(-ik x)\}_{k=-\infty}^{\infty}$ with corresponding eigenvalues $\{k^{2}\}_{k=-\infty}^{\infty}$, which can be organized from \emph{low frequency} ($|k|$ small) to \emph{high frequency} ($|k|$ large).  There is a well-developed literature interpreting the graph Laplacian $L$ as a discretization of the continuum differential operator $\Lap$ \cite{Belkin2006_Convergence, Trillos2018_Variational, Trillos2020_Error}.  Using analogous language, we can say that $\Psi_{2}$ is the lowest-frequency non-constant eigenvector, while $\Psi_{n}$  is the highest-frequency eigenvector.  (There is such a large literature on $\Psi_2$ that it has its own name: the {\em Fiedler vector}.)
As noted in Corollary~\ref{cor:extremeI}, these are the extremizers of $\I(\cdot \ ;L)$.

\begin{figure}[ht]
    \centering
\begin{tikzpicture}
\node at (0,4) {\includegraphics[width=1.3in]{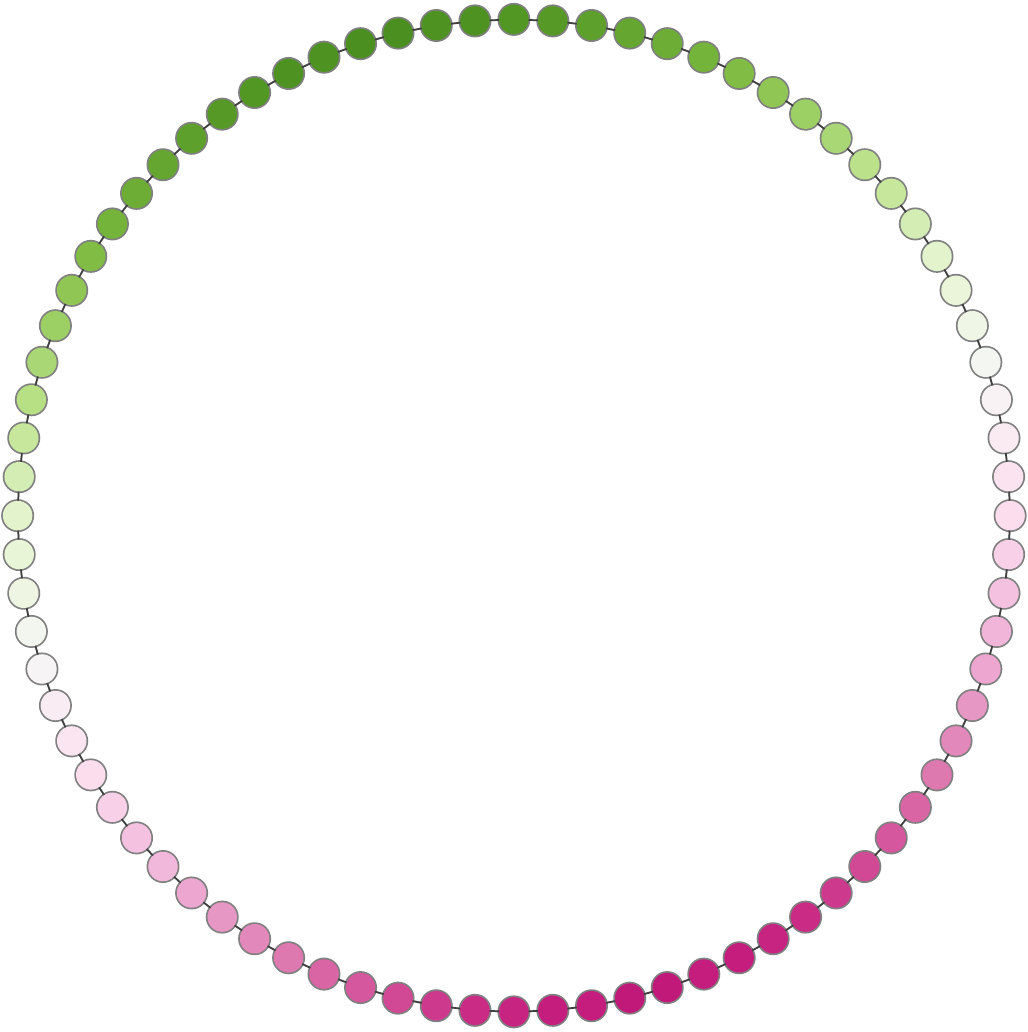}};

\node at (4.4,4) {\includegraphics[width=1.9in]{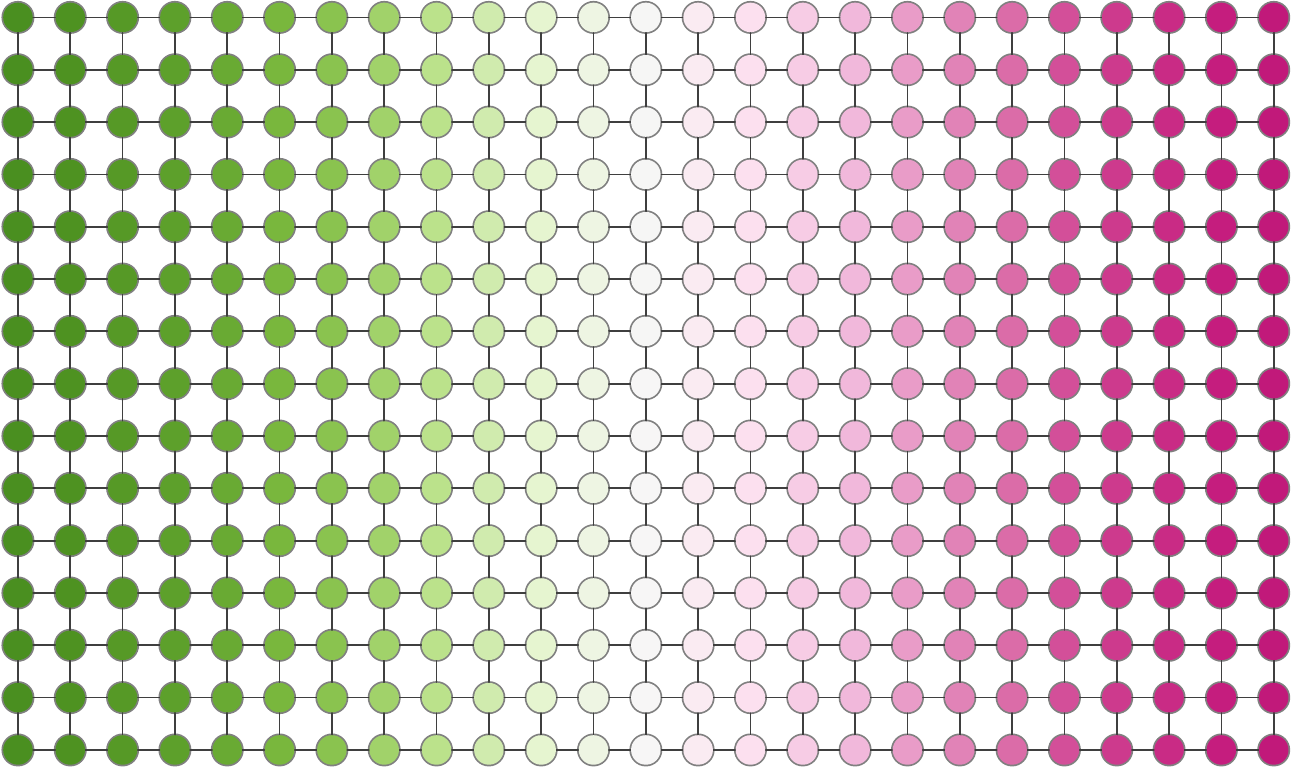}};

\node at (10.2,4) {\includegraphics[width=2.4in]{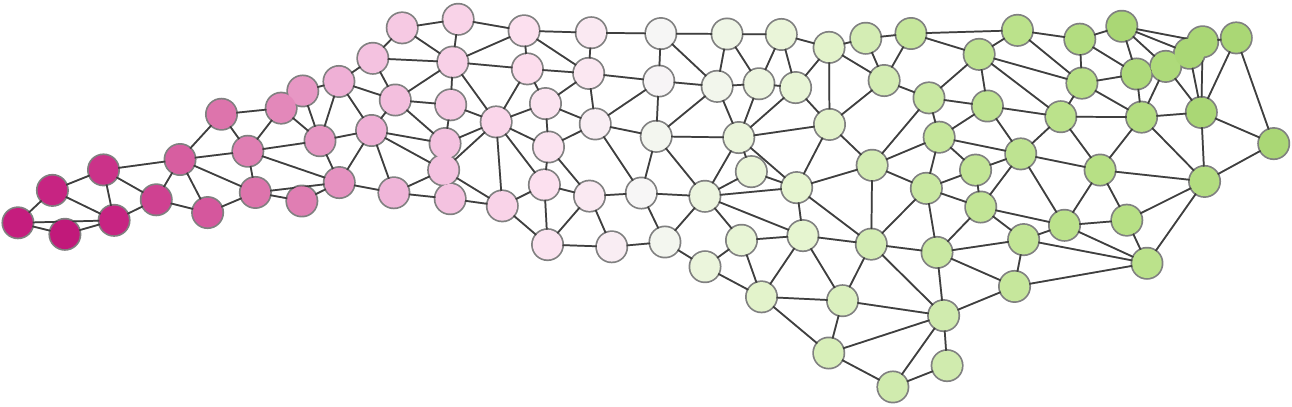}};

\node at (0,0) {\includegraphics[width=1.3in]{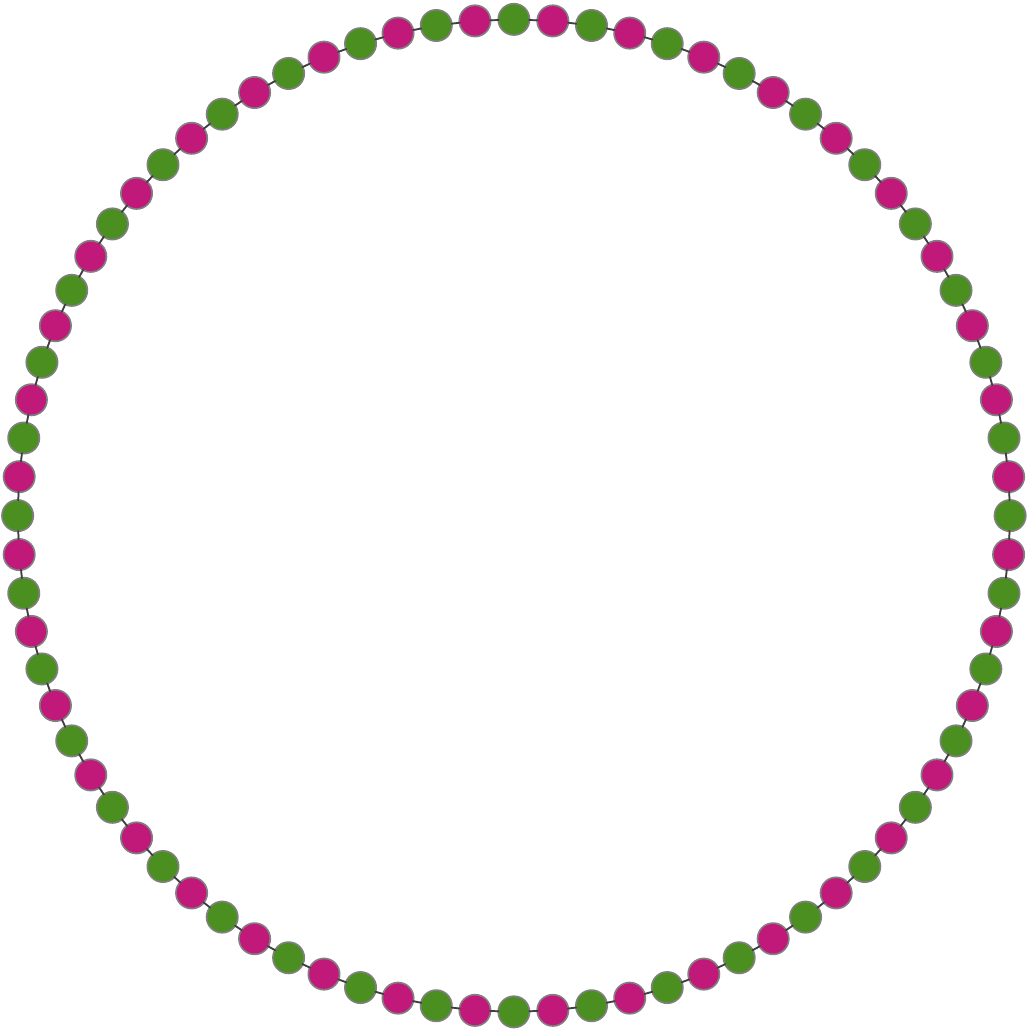}};

\node at (4.4,0) {\includegraphics[width=1.9in]{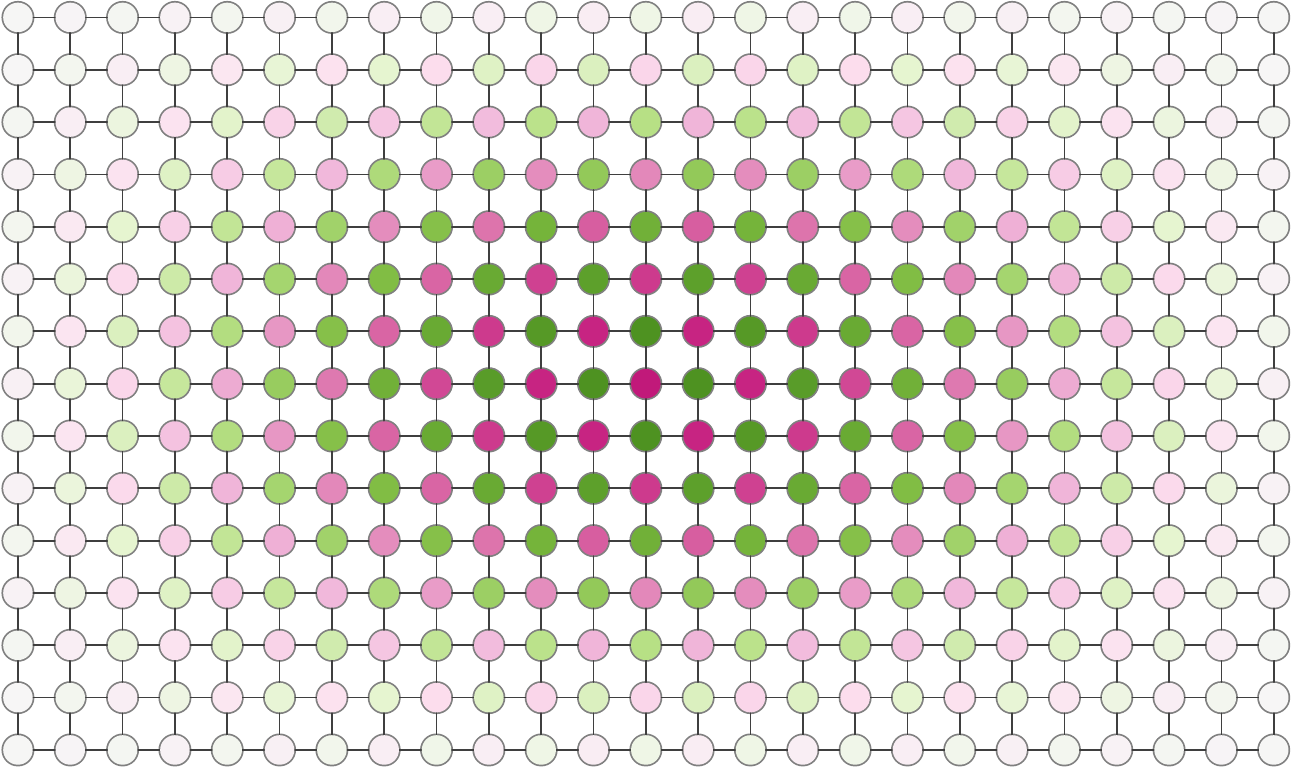}};

\node at (10.2,0) {\includegraphics[width=2.4in]{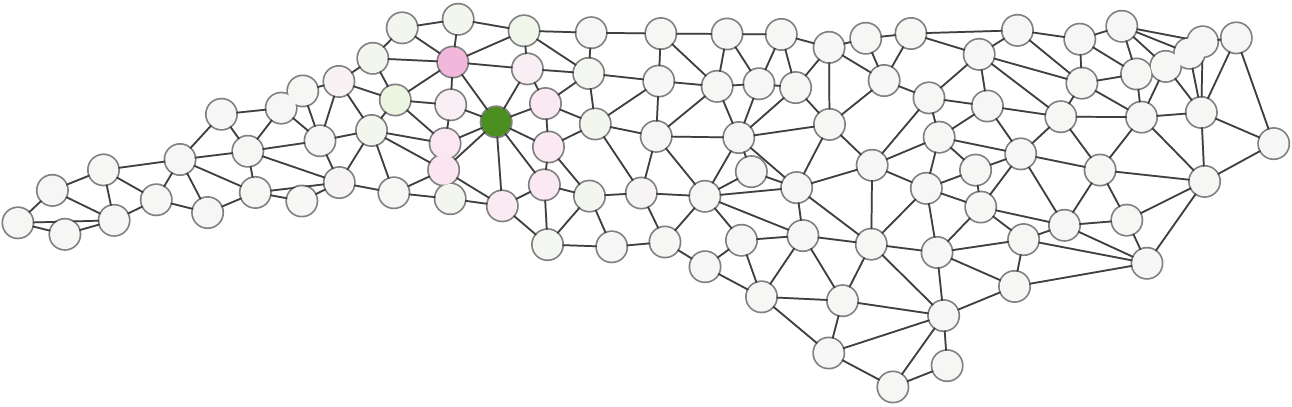}};

\end{tikzpicture}
\caption{For the (even-length) cycle graph, the lowest-frequency non-trivial eigenvector $\Psi_2$ of $L$---called the {\em Fiedler vector}---oscillates slowly and realizes a large $\I$ value, while the highest-frequency eigenvector $\Psi_{n}$ is a perfect alternating pattern and yields the extremal $\I=-1$. For the grid graph, $\Psi_{2}$ oscillates slowly.  On the other hand, $\Psi_{n}$ oscillates very rapidly.  Note that the low-frequency eigenvector is highly clustered, while the high-frequency eigenvector is not.  The right column shows the county dual graph of North Carolina with its  lowest- and highest-frequency eigenvectors.  One captures cluster structure, while the other is highly localized at a high-degree vertex.}
    \label{fig:LowAndHighFrequency}
\end{figure}

We can leverage well known facts about the Laplacian to rephrase some of the empirical observations in this paper.  
In particular, if the underlying graph has $k$ internally well-connected components that are weakly connected to each other, then  $L$ will have $k$ eigenvalues close to 0 and $\mu_{k+1}\gg 0$ \cite{Ng2001_Spectral}.  This gives many qualitatively different functions on the graph that all have a low $\I(\cdot \ ;L)$ indicative of clustering.

As we have seen, the case of vertex-regular graphs is one where the eigenvectors of $L$ provide exact solutions to the extremization problem for all four weight matrices.  These can be phrased in familiar spectral graph theory language as {\em nodal decompositions} of the graph into maximal regions where the eigenfunction does not change sign \cite{urschel}.  For the lowest-frequency non-constant eigenvector $\Psi_2$, this optimal partition solves a relaxation of the combinatorial normalized cuts functional \cite{Shi2000_Normalized}.  Though realistic graphs are not regular, Figure~\ref{fig:LowAndHighFrequency} gives an indication that configurations registering maximal segregation are not so far from what a nodal decomposition might predict.

In lattice graphs, which are regular except for along their boundaries, the high-frequency eigenvectors are damped checkerboard patterns, shown in Figure~\ref{fig:HexagonalLattice}.  In highly irregular graphs, such as in Figure \ref{fig:LowAndHighFrequency}, the checkerboarding may be strongly localized around vertices of high degree.  It is clear from both this example and from the double-star graphs in Figure~\ref{fig:TwinStarGraph} that $L$ is still highly sensitive to degree disparities. In general, the largest eigenvalue may be characterized as a measure of how close the graph is to being bipartite \cite{Desai1994_Characterization}, but a corresponding characterization of the highest-frequency eigenvectors is elusive.

  \FloatBarrier

\subsection{Dirichlet Energy Functionals}  With the graph Laplacian $L$, we can quantify a kind of smoothness via a notion of \emph{Dirichlet energy} on a graph.

Suppose that a vector is written as  $\V=\displaystyle\sum_{i=1}^{n}\alpha_{i}\Psi_{i}$ 
in the orthonormal basis of eigenvectors of $L$, so 
that $\X=\displaystyle\sum_{i=2}^{n}\alpha_{i}\Psi_{i}$.  Then $\V^{\T} L \V=\displaystyle\sum_{i=1}^{n}{\alpha_{i}}^{2}\mu_{i},$ where the right-hand side is large when a large portion of the coefficient energy localizes on the highest-frequency eigenvectors.  In analogy with classical Dirichlet energy functionals \cite{Evans1998_Partial}, we may define a \emph{graph Dirichlet energy} for general functions $\V$ on the graph, which are not necessarily zero-centered.  

\begin{definition}[Dirichlet energy]\label{defn:E}Let $\V=\displaystyle\sum_{i=1}^{n}\alpha_{i}\Psi_{i}$ be a function on a graph with Laplacian $L$ and associated orthonormal eigenvectors $\{\Psi_{i}\}_{i=1}^{n}$.  The Dirichlet energy associated to $L$ is given by $\mathcal{E}(\V)=\sqrt{\sum_{i=1}^{n}\alpha_{i}^{2}\mu_{i}}=\sqrt{\V^{\T}L\V}.$ 
\end{definition}

Note that for any scalar $\alpha>0$, $\mathcal{E}(\alpha\V)=\alpha\mathcal{E}(\V)$.  In particular, $\mathcal{E}$ is not scale-invariant, since it lacks the denominator of $\V^{T}\V$ used when computing $\I$.    

Compare this to the  classical Dirichlet energy functional on $f:\Tor\rightarrow\R$ given by
\[E(f)=\displaystyle\int_{\Tor}{\|\D f(x)\|_2}^{2}dx,\]
where $\D f$ is interpreted in a weak sense.
By Stokes' Theorem, $\displaystyle\frac{1}{(2\pi)^{m}}\int_{\Tor}{\|\D f(x)\|_2}^{2}dx=\langle -\Delta f, f\rangle$ %under Dirichlet boundary conditions that $f=0$ on $\partial\Tor$, 
so that $E(f)=\sqrt{\langle \Lap f, f\rangle}$, in direct analogy to Definition \ref{defn:E}.  To further develop the connection between  the graph definition and the classical definition, assume  $m=1$ for simplicity.  If $f$ has Fourier expansion $f(x)=\sum_{k=-\infty}^{\infty}c_{k}\exp(-i k x)$, then $\D f(x)=-\sum_{k=-\infty}^{\infty}c_{k}ik\exp(-i k x)$.  By Parseval's Theorem,
$${\|f\|_2}^{2}=\sum_{k=-\infty}^{\infty}{c_k}^{2}, \qquad {\|\D f\|_2}^{2}=\sum_{k=-\infty}^{\infty}k^{2}{c_k}^{2}.$$
Noting that $\{k^{2}\}_{k=0}^{\infty}$ are precisely the eigenvalues of $\Lap$ defined on $L^{2}(\Torone)$, we see that the classical and graph definitions agree.

If $\displaystyle\int_{\Tor}{\|\D f(x)\|_2}^{2}dx$ is small, then $f$ is locally smooth in the sense that $\nabla f$ has small magnitude in most areas.  Noting again that \[E(f)=\displaystyle\int_{\Tor}{\|\D f(x)\|_2}^{2}dx=-\displaystyle\int_{\Tor}\Delta f(x) f(x)dx=\langle \Lap f, f\rangle\]
suggests that $\X^{\T} L \X$, the graph discretization of $\langle \Lap f,f\rangle$, is a measure of local smoothness of the function $\X$ on the graph. This connection is elaborated in  \cite{Chung1997_Spectral}.  Under this interpretation, $\I(\X;L)$ is small (indicating segregation) when $\X$ is mostly smooth.  Importantly, $\I(\cdot \ ;L)$ is scale-invariant, unlike $\mathcal{E}(\cdot)$; a note about developing scale-sensitive measures of segregation will be discussed in Section \ref{sec:Conclusions}.

%%%%%%%%%%
%%%%%%%%%%

\section{Random Walks and \I}\label{sec:WalksandI}
We now consider a class of weight matrices for which $\I$ has a random walk interpretation, namely bistochastic matrices (those with rows and columns summing to one). 

We first recall some basic facts about Markov chains and random walks.  
By definition, a {\em Markov chain} on a finite state space (with states indexed $1,\dots,n$) is a random process encoded by a stochastic $n\times n$ matrix $K$.  The associated random walk steps from the $i$th to the $j$th state with probability $K_{ij}$.  
This can be visualized as random walk on a graph with $n$ nodes, and an edge $(i,j)$ present when $K_{ij}$ or $K_{ji}>0$.  
We can encode the walk by matrix multiplication if a probability vector $\V$ is interpreted as describing a probabilistic position on the state space.  Then $\V^\T K$ is the new position after one step of the walk.

With this, the reader can verify that our spatial weight matrix $P$ discussed above has an interpretation as the simple random walk on the geography units, making all neighboring units equally likely at each stage.  As long as the graph is connected and aperiodic (for instance, if it has any triangles), this random walk converges to a unique stationary distribution in which the probability of being at any node in the long term is proportional to its degree.  When the graph encodes the tracts of a city, as in many of our examples here, this stationary distribution for $P$ is not very meaningful.  

Given an arbitrary Markov chain $K$, the classical {\em Metropolis-Hastings} construction allows us to modify the random walk so that it targets a specified stationary distribution $\pi$ (an arbitrary probability distribution on the states $1,\dots,n$).  The Metropolis-Hastings matrix $M=M(K,\pi)$ gives a reversible chain, meaning that $\pi(i)M_{ij}=\pi(j)M_{ji}$ for all $i,j$.  Note that if the Metropolis-Hastings matrix is set to target the uniform distribution, then reversibility means the matrix is symmetric and therefore bistochastic.
Next, we establish that for any bistochastic matrix $Q$, such as for 
the uniformizing Metropolis matrix $\MH=M(P,\frac 1n \one)$, Moran's \I can be interpreted in terms of variance reduction.

\begin{theorem}[Random walk interpretation of $\I$]\label{thm:randomwalks}
For a bistochastic matrix $Q$ and  a column vector $\V$,  consider $\W=\V^\T Q$, the value of $\V$ after one step of the Markov chain given by $Q$. Let $\sigma_0$ and $\sigma_1$ be the standard deviation of the values in $\V$ and $\W$ respectively, so that the ratio $\sigma_1/\sigma_0$ gives the variance reduction in one step of the walk. Let $\rho(\V, \W)$ be the correlation between the values in $\V$ and $\W$.  Let $\X=\V-\vbar\one$ and $\Y=\W-\bar{w}\one$ be the zero-centered vectors before and after applying $Q$.  Then
\begin{enumerate}[(a)]
    \item $\I(\V; QQ^\T) = \left(\frac{\sigma_1}{\sigma_0}\right)^2$.
    \item $\I(\V; Q) = \frac{\Y^\T\X}{\X^\T\X}  = \rho(\V, \W) \cdot \frac{\sigma_1}{\sigma_0}$.
\end{enumerate}
\end{theorem}
\begin{proof}
To see (a), note that because $Q$ is bistochastic, the average values satisfy $\overline{\V^\T Q} = \frac{1}{n} \cdot \V^\T Q \one = \frac{1}{n} \cdot \V^\T \one = \overline{\V^\T}$. Thus
\begin{align*} 
    \frac{{\sigma_1}^2}{{\sigma_0}^2} & =
    \frac{( \V^\T Q - \overline{\V^\T}\one^\T) (\V^\T Q - \overline{\V^\T}\one^\T)^\T}{(\V - \overline{\V}\one)^\T (\V - \overline{\V}\one)} 
     = \frac{( \V^\T Q - \overline{\V^\T}\one^\T Q) (\V^\T Q - \overline{\V^\T}\one^\T Q)^\T}{(\V - \overline{\V}\one)^\T (\V - \overline{\V}\one)}  \\
    & = \frac{( \V^\T - \overline{\V^\T}\one^\T ) Q Q^\T (\V^\T - \overline{\V^\T}\one^\T)^\T}{(\V - \overline{\V}\one)^\T (\V - \overline{\V}\one)}
     = \I(\V; QQ^\T).
\end{align*}

A similar calculation yields (b): 
\begin{align*}
    \I (\V; Q) & = \frac{( \V^\T - \overline{\V^\T}\one^\T ) Q (\V - \overline{\V}\one)}{(\V - \overline{\V}\one)^\T (\V - \overline{\V}\one)} 
     = \frac{( \V^\T Q - \overline{\V^\T Q}\one^\T ) (\V - \overline{\V}\one)}{(\V - \overline{\V}\one)^\T (\V - \overline{\V}\one)} = \frac{\Y^\T \X}{\X^\T \X}  \\
    & = \frac{( \V^\T Q - \overline{\V^\T Q}\one^\T ) (\V - \overline{\V}\one)}{||\V^\T Q - \overline{\V^\T Q}\one^\T ||_2 \cdot ||\V - \overline{\V}\one||_2}  \cdot \frac{||\V^\T Q - \overline{\V^\T Q}\one^\T ||_2}{||\V - \overline{\V}\one||_2} 
     = \rho(\V^\T, \V^\T Q) \cdot \frac{\sigma_1}{\sigma_0}.
\end{align*}
\end{proof}

Note $\rho(\V,\W)$ is just the one-step autocorrelation (i.e., time lag $1$) for the Markov chain $Q$.  Part (a) of Theorem \ref{thm:randomwalks} states that $\I(\V; QQ^\T)$ can be interpreted as the factor by which the variance is reduced in two steps of evolution under the Markov chain associated to $Q$. Note that a general weight matrix $W$ admits a decomposition $W = QQ^\T$ for such a matrix $Q$ iff $W$ is bistochastic, positive semidefinite, and symmetric. We also observe that $\I(\V; QQ^\T)$ is always non-negative. Part (b) of Theorem \ref{thm:randomwalks} states that $\I(\V; Q)$ decomposes as the product of the one-step autocorrelation for $\V$ and the reduction in standard deviation after one step. To see how this plays out in extreme cases, observe that if $|\I(\V; Q)| \approx 1$ then the standard deviation of $\V$ must remain roughly the same after one step of $Q$, with the value of $\V$ after one step being either highly correlated ($\I(\V; Q) \approx 1$) or anti-correlated ($\I(\V; Q) \approx -1$) with the initial value of $\V$.
The alternative expression $\I(\V;Q)= \frac{\Y^\T\X}{\X^\T\X}$ makes it clear that if a step of $Q$ changes $\V$ to something near-uniform, then $\I$ will have small magnitude.  Large $\I$ can only occur when the diffusion process leaves $\Y^\T\X\approx \X^\T\X$.

When we pass from the (symmetric) adjacency matrix $A$ to the stochastic normalization $P$, this need not be bistochastic unless the underlying graph is regular.  To take advantage of this theorem, we will use the uniformizing Metropolis-Hastings matrix $\MH=M(P,\frac 1n \one)$, which is both symmetric and bistochastic.  

In fact, $\MH$ is a best symmetric bistochastic approximation to $P$ in the sense that it minimizes the difference $\sum_{i \neq j} | Q_{ij} - P_{ij}|$ among all symmetric bistochastic matrices $Q$. 
This follows from the work of Billera and Diaconis in  \cite{billera2001geometric}, which proves the more general statement that the Metropolis-Hastings matrix $M(K,\pi)$ whose $i,j$ term is 
$$M(K,\pi)_{ij} = \min\left( K_{ij}, \quad \frac{\pi(j)}{\pi(i)} K_{ji}\right)$$
minimizes the $\pi$-weighted $\ell^1$ distance from $K$ to $\{Q\}$.  
In particular, in our setting, if the graph $\G$ is regular, then $P$ is already symmetric and bistochastic, so $\MH = P$. 
In the irregular case, $\MH$
can be thought of as a modification of the simple random walk; it works by introducing a rejection step that makes the walk extremely lazy when a low-degree vertex is next to a high-degree vertex.  (See Figure~\ref{fig:PvsM} for an illustration.)

\begin{figure}[ht]\centering
\begin{tikzpicture}
\node at (0,0)
{\includegraphics[width=2in]{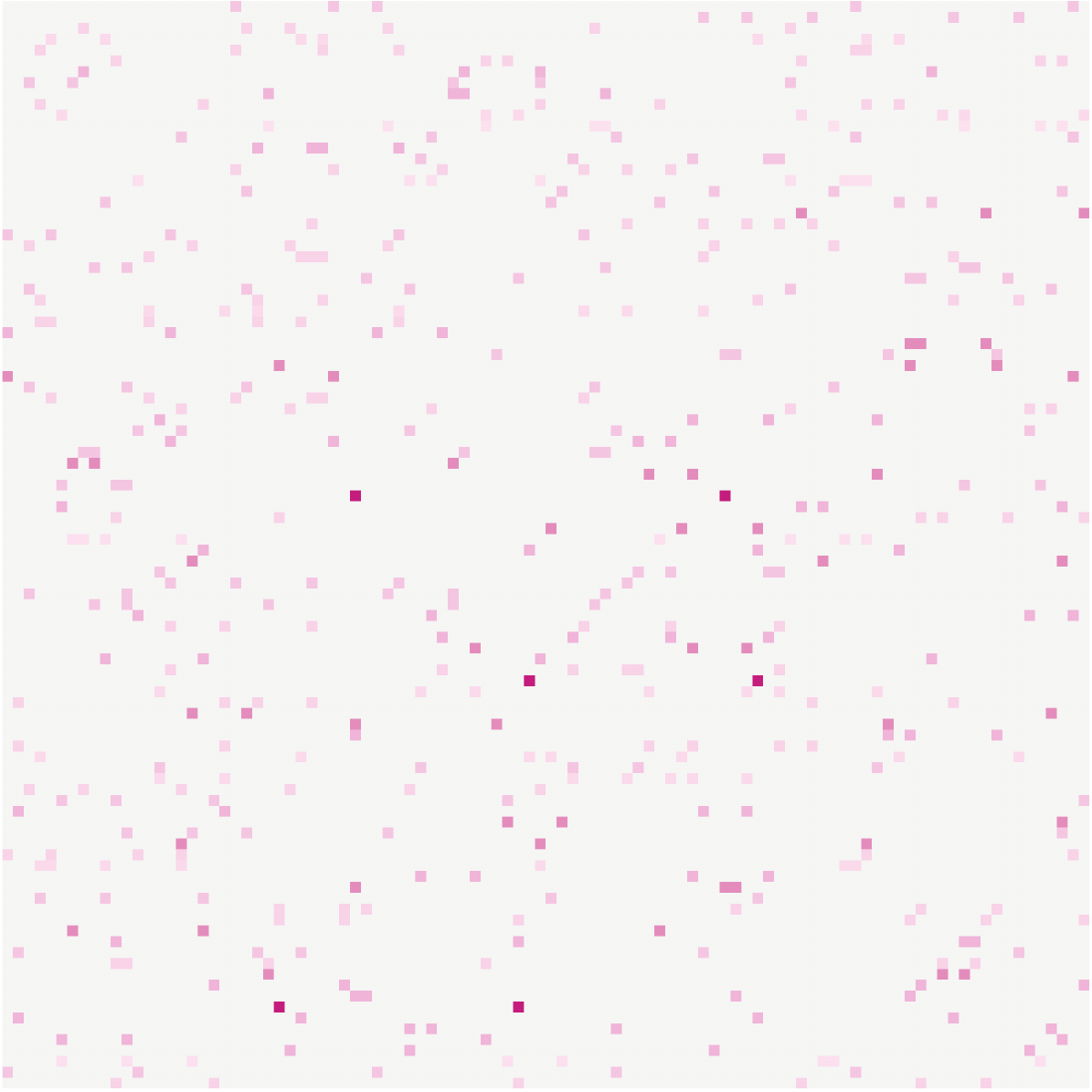}};
\node at (6,0) {\includegraphics[width=2in]{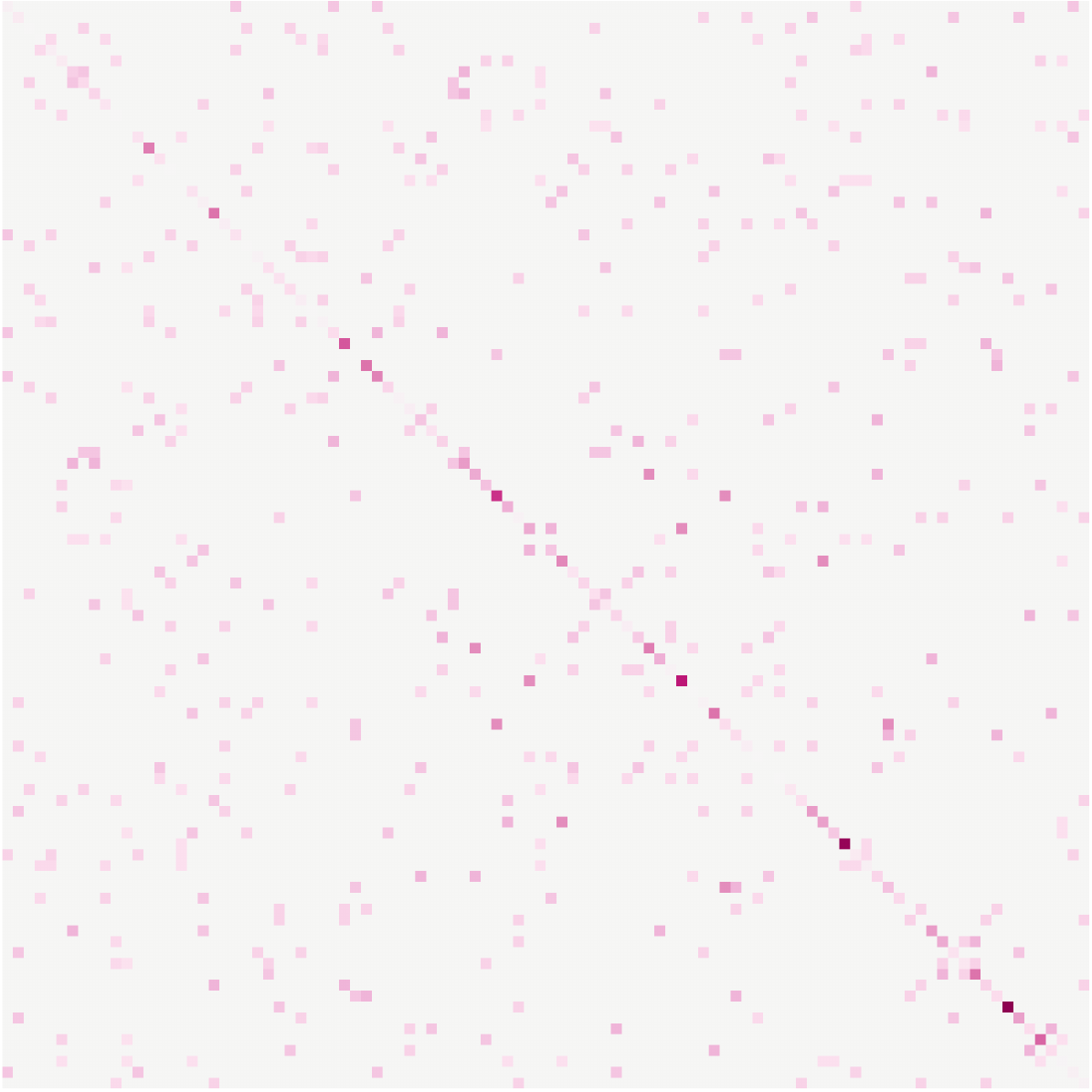}};
\end{tikzpicture}

\caption{The matrices $P$ (left) and $\MH$ (right) for the North Carolina county dual graph shown in Figure~\ref{fig:LowAndHighFrequency}.  Darker colors indicate higher matrix entries. We see that $\MH$ has some large diagonal entries, meaning that the associated random walk is very lazy at some nodes.}\label{fig:PvsM}
\end{figure}

Using $\MH$ for the spatial weighting in \I succeeds where $P$ does not in mitigating the degree effects discussed throughout this paper.  This is because $\MH$ eliminates $W$-degree discrepancies in both the rows and columns of $A$, while $P$ only standardizes the rows.  Indeed, because $\MH$ is a non-negative stochastic matrix, its largest eigenvalue is $1$, realized by its stationary vector $\one$.  (This is a well known Markov chain property following from the Perron-Frobenius theorem.)  We recall from Corollary~\ref{cor:extremeI} that the extreme values of $\I(\cdot \ ;M)$ are realized at the eigenvalues $\lambda_2$ and $\lambda_n$ of $M$.

\begin{corollary}[The range of \I with weights from $\MH$]
For any graph $\G$ and any function $\V$, we have $-1\le \I(\V;\MH)\le 1$.
\end{corollary}

This discussion suggests another way in which the random walk interpretation can be fruitful.  It is a standard fact in Markov chain theory that the convergence statistics (such as {\em mixing time}) of a chain have upper and lower bounds in terms of the spectral gap, here $1-\lambda_2$, of the associated matrices.  Random walks that converge more slowly correspond to smaller spectral gaps and smaller Cheeger constants (graphs that have relatively short cuts into large pieces).  In the world of geography dual graphs, this says that if the locality itself can be cut in half with a relatively short cut, as in all of the realistic examples here, then $\I$ scores for $\MH$ can get close to 1. Indeed, planar graphs where the number of vertices is much greater than the largest vertex degree are slow-mixing and have a small spectral gap \cite{SpielmanTeng}, which ensures that $\I(\cdot\ ;\MH)\approx 1$ is achievable.  

In sum, using $\MH$ for the spatial weight matrix allows us to characterize $\I$ in extremely intuitive language.  We imagine a diffusion process that begins with the observed demographic distribution in a locality and conducts a random walk of residents that targets the uniform distribution.  Over the long term, this random walk process must reduce the variance, which is initially $\|\X\|$, to zero.  
Moran's \I now measures {\em how well a uniformizing diffusion succeeds in a single step}.  It is quite reasonable to regard this as a measurement of segregation:  a certain group will be considered very far from uniformly dispersed in a population if  many steps through neighboring geography are required for the group to approach uniformity.

%diagonal entries in $\MH$ associated to the leaves in the double-star graph are $n/(n+1)$.  

\section{Conclusions, Recommendations, and Future Directions}
\label{sec:Conclusions}

Moran's \I is a valuable way to detect spatial patterns, or to test for spatial correlation in the residuals of some models, especially when combined with a statistical significance test. However, users must exercise caution when using Moran's \I as a gradated measurement (and not just a qualitative test) for a number of reasons.  The underlying graph topology and the choice of spatial weight matrix used in computing $\I$ both strongly impact the range of possible $\I$ values, so using \I to compare across localities remains challenging.  

We summarize the findings of this paper with the following practical recommendations.

\bigskip

\begin{enumerate}[(1)]

\item  For a given graph $\G$ and weight matrix $W$, use the methods here to compute the range of achievable \I values in order to decide whether a particular demographic vector $\V$  has an extreme score.  However, intermediate \I values (say, $\I=.6$)  remain hardest to interpret.  

\item Use circumspect language when comparing \I values for different graphs $\G$ and $\G'$, particularly with standard choices of weight matrix like $A$ and $P$.  The computation $\I(\V;P_\G)>\I(\V';P_{\G'})$ should not be presented as a finding that the first city is more segregated than the second.

\item Both for within-graph and between-graph comparisons, the best-suited spatial weight matrix is $\MH$, which makes \I interpretable in terms of how the subgroup's population diffuses in a random walk.  Furthermore, with $\MH$ weights, \I is actually bounded between $-1$ and $1$ and for large planar graphs can achieve $\I\approx 1$.

\item The discussion above suggests a novel role for Moran's \I when a 
demographic function $\V$ is fixed on a sufficiently fine graph.  
Recall that the dependence of a measurement on the choice of units is an important problem in geography called the {\em modifiable areal unit problem}.
Given $\V$, a choice of units for which $\I<0$ can be interpreted as an aggregation of the underlying fine data that captures regions that are demographically distinct---having similar demographics within the unit and different demographics on neighboring units.  
That is, when considering alternative choices of geographical units (like census tracts in Chicago versus official neighborhoods maintained in city statistics),  a {\em negative} Moran's score can be interpreted as a signal that the units track with demographic differences.  From this point of view, \I can facilitate a kind of demographic community detection.  
\end{enumerate}

\bigskip

This study suggests several interesting questions for future exploration.  
While the connection between high segregation and low-frequency eigenvectors follows from our analysis, the connection between low segregation and high-frequency eigenvectors is more subtle.  This is due to the large impact that the underlying graph geometry has on even the local properties of high-frequency eigenvectors (see Figure \ref{fig:LowAndHighFrequency}). 
Understanding the extent to which high-frequency eigenvectors localize (i.e., have concentrated support) on irregular graphs would be interesting in its own right, and has potential connections to Anderson localization in the continuum setting \cite{Filoche2012_Universal}.

Another useful direction of inquiry would be to modify the definition of $\I$ by building new metrics that make use of $\X^\T W \X$ without normalizing by the denominator $\X^\T \X$.  By creating scale-sensitive scores for the deviation in a population, we could remediate the degeneration in the interpretation of \I for low-variance distributions.

\section*{Acknowledgments}

We thank Larry Guth and Eugene Henninger-Voss for enlightening conversations.  We also thank the two anonymous referees for their helpful comments that significantly improved the paper.  MD and TW were partially supported by the NSF through grant OIA-1937095, and MD through DMS-2005512.  JMM was partially supported by the NSF through grants DMS-1924513 and DMS-1912737. 

\bibliographystyle{abbrv}
\bibliography{Moran.bib}

\clearpage

\section*{Supplementary Material}

\section{Georgia Dual Graph Degree Distributions}

The dual graphs for Georgia's census geography on the level of blocks, block groups, and counties are shown in the main text in Figure 2. % \ref{fig:GA-duals}
Some statistics for those graphs are recorded in 
Table 2; % \ref{tab:tab-GA}  
in particular, they have an average vertex degree of roughly 9.6, 5.5, and 5.3, respectively.  The degree distributions are shown here.

\begin{figure}[ht] \centering
\begin{tikzpicture}
\node at (0,2.75) {\bf Blocks};
\node at (5,2.75) {\bf Block Groups};
\node at (10,2.75) {\bf Counties};
\node at (0,0) {\includegraphics[width=1.8in]{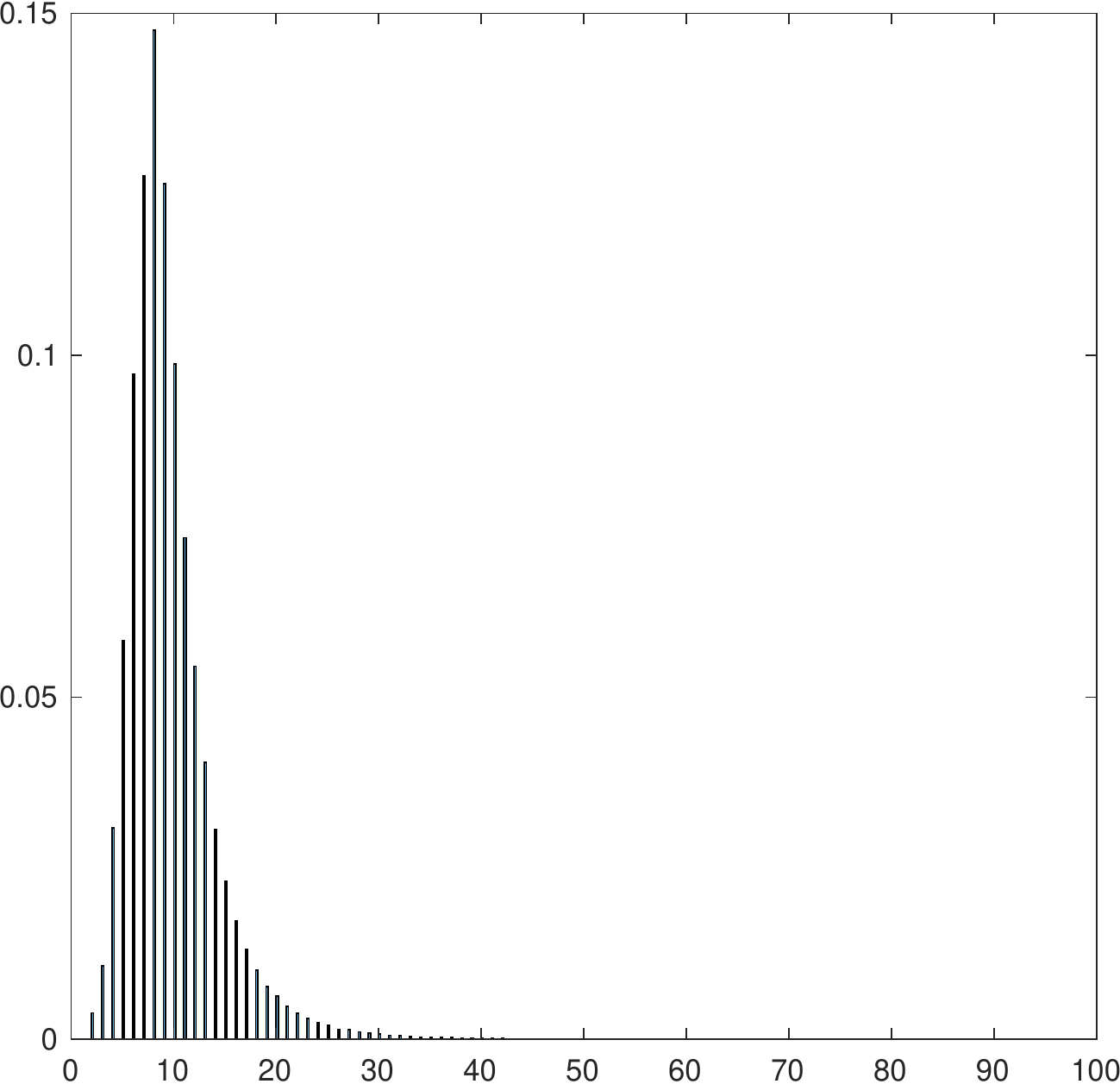}};
\node at (5,0) {\includegraphics[width=1.8in]{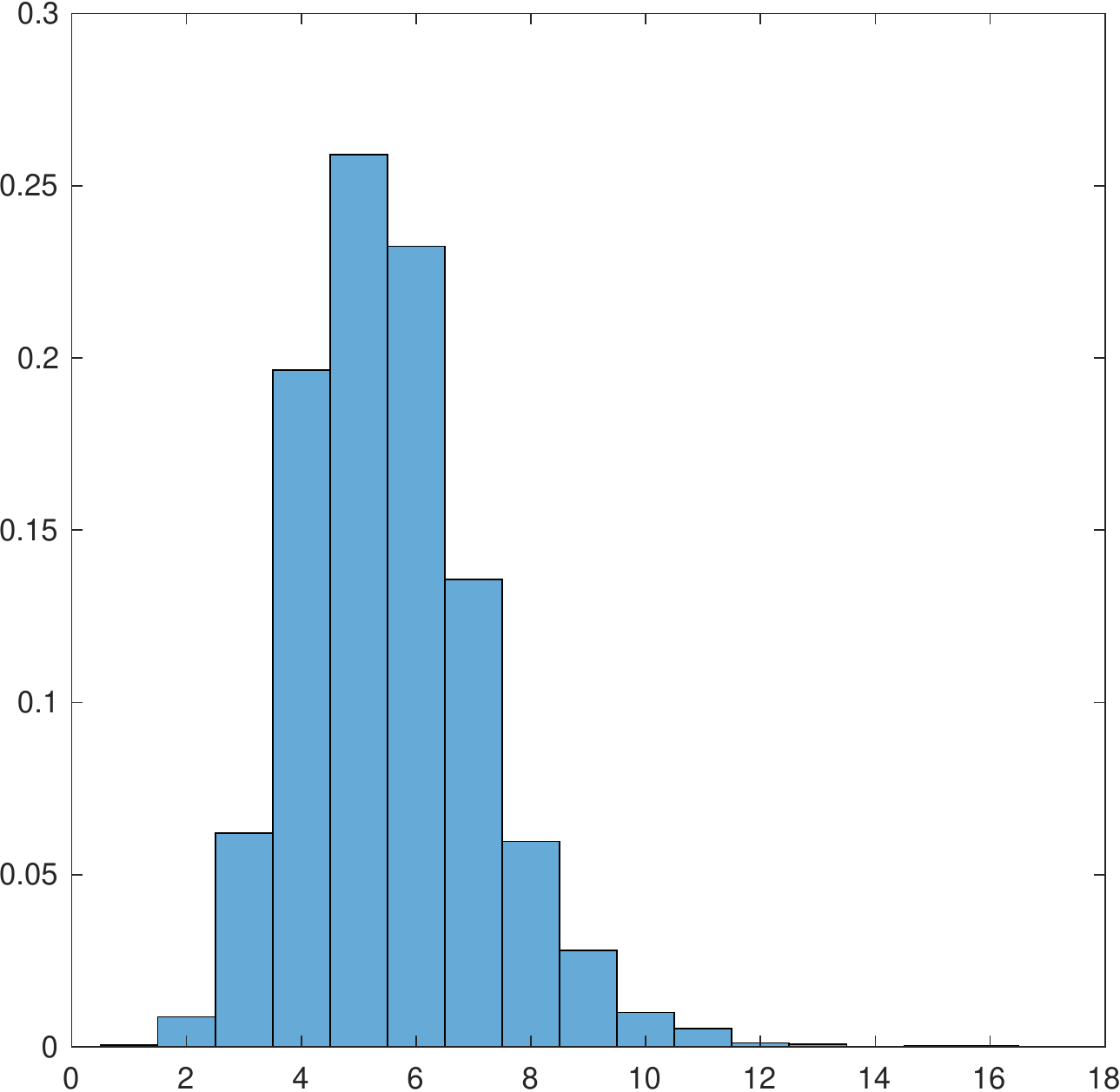}};
\node at (10,0) {\includegraphics[width=1.8in]{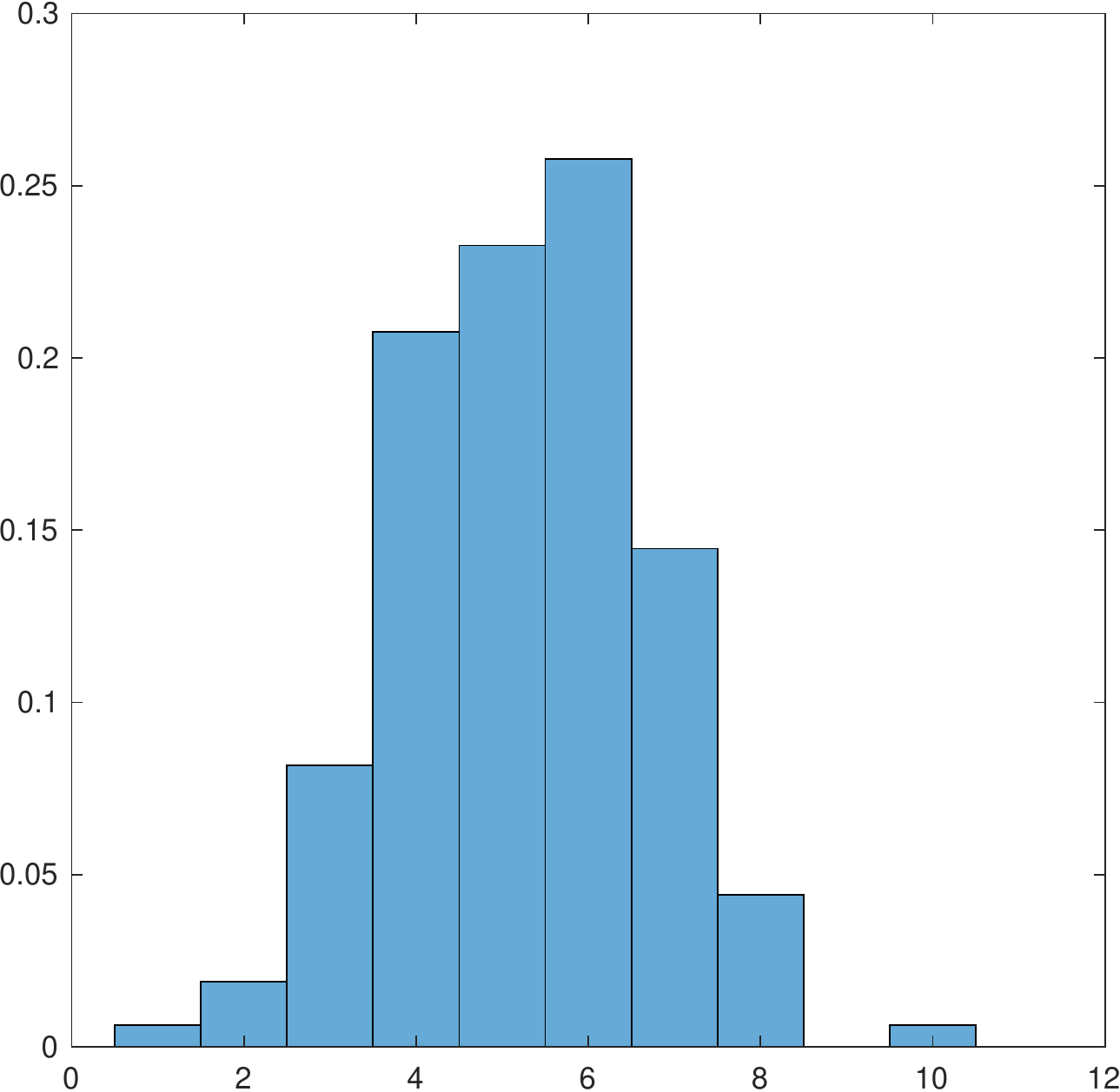}};
\end{tikzpicture}
\caption{Degree distribution for Georgia's  census geography dual graphs.}\label{fig:GA_DegDists}
\end{figure}
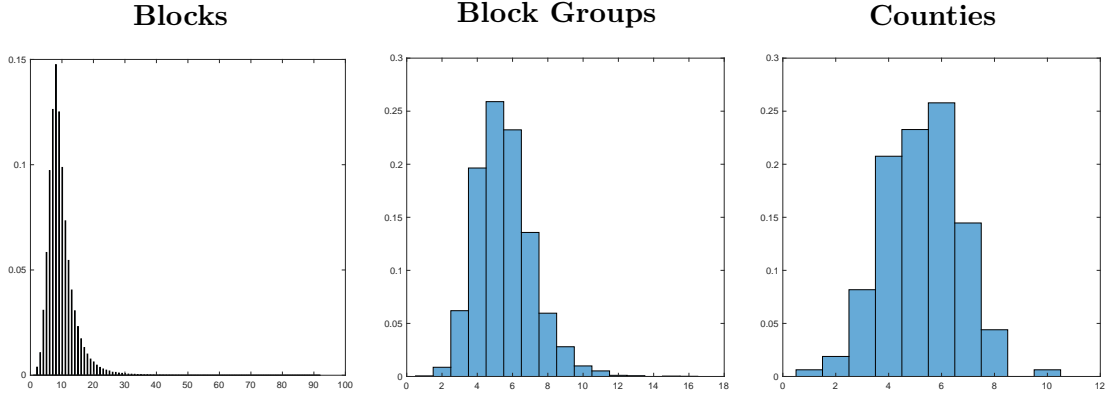

\section{Section 5.1 Experimental Details}
In comparing Moran's \I across different choices of weight matrix, one example relied on an experiment using edge deletions from a triangulated graph to produce some square faces and some variation in degree.  The following figures illustrate the loose sense in which this interpolates from hex to square lattice graph.  
(Compare to Figure~\ref{fig:GA_DegDists}.)

\begin{figure}[ht] \centering
\begin{tikzpicture}
\begin{scope}
\draw (0:1)--(60:1)--(120:1)--(180:1)--(240:1)--(300:1)--cycle;
\draw (120:1)--(300:1);
\draw (0:1)--(180:1);
\draw (60:1)--(240:1);
\end{scope}
\begin{scope}[xshift=1.5cm,yshift=.866cm]
\draw (0:1)--(60:1)--(120:1)--(180:1)--(240:1)--(300:1)--cycle;
\draw (120:1)--(300:1);
\draw (0:1)--(180:1);
\draw (60:1)--(240:1);
\end{scope}

\node at (3.1,.6) {$\cong$};

\begin{scope}[xshift=4.8cm,yshift=-1cm]
\draw (0,0) rectangle (1,3);
\draw (-1,1) rectangle (2,2);
\draw (1,3) -- (2,2);
\draw (0,0)--(-1,1);
\draw (1,0)--(-1,2);
\draw (1,1)--(0,2);
\draw [green,line width=4,opacity=.5] (1,1)--(0,2);
\draw (2,1)--(0,3);
\end{scope}

\node at (3.1+4.4,.6) {$\to$};

\begin{scope}[xshift=9.2cm,yshift=-1cm]
\draw (0,0) rectangle (1,3);
\draw (-1,1) rectangle (2,2);
\draw (1,3) -- (2,2);
\draw (0,0)--(-1,1);
\draw (1,0)--(-1,2);
\draw [gray,dashed] (1,1)--(0,2);
\draw (2,1)--(0,3);
\end{scope}
\end{tikzpicture}
\caption{The hexagonal lattice is equivalent to the square lattice with diagonals.  Deletions of diagonals would strictly interpolate from hex to square.  Instead, we delete 10 or 20\% of all edges to create a graph with more variation in vertex degree.}\label{fig:deletions}
\end{figure}
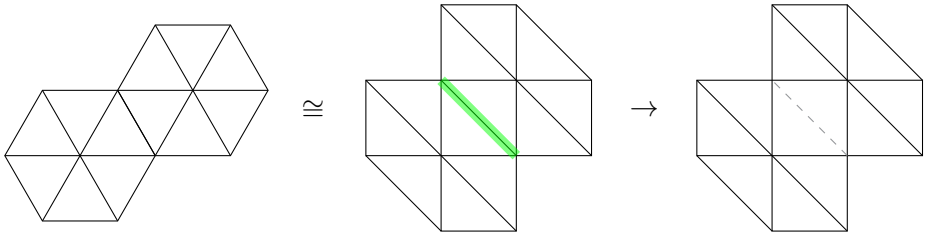

\begin{figure}[ht] \centering
\begin{tikzpicture}
\node at (0,3.7) {\bf 10\% deletions};
\node at (8,3.7) {\bf 20\% deletions};
\node at (0,0) {\includegraphics[width=2.5in]{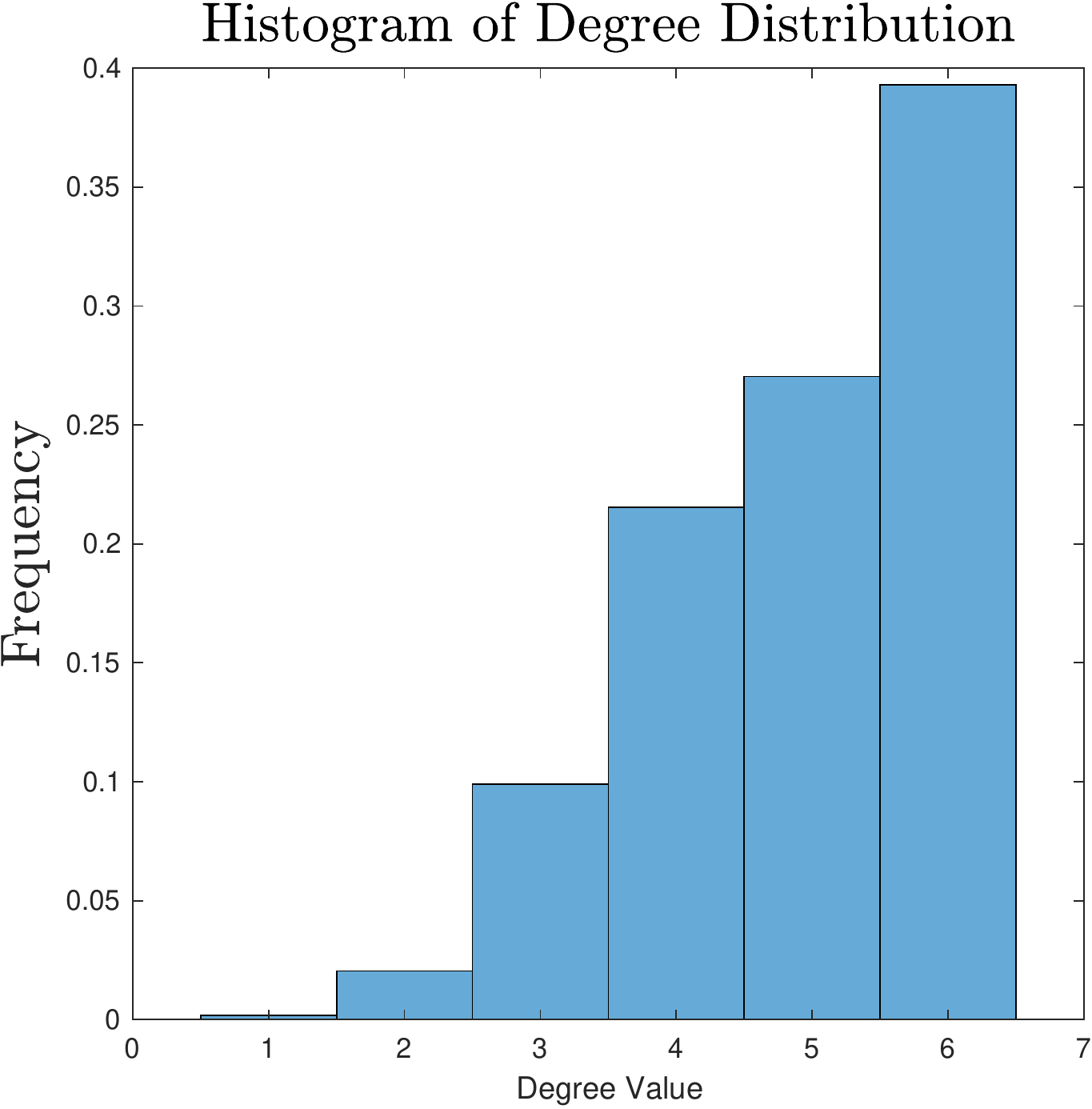}};
\node at (8,0) {\includegraphics[width=2.5in]{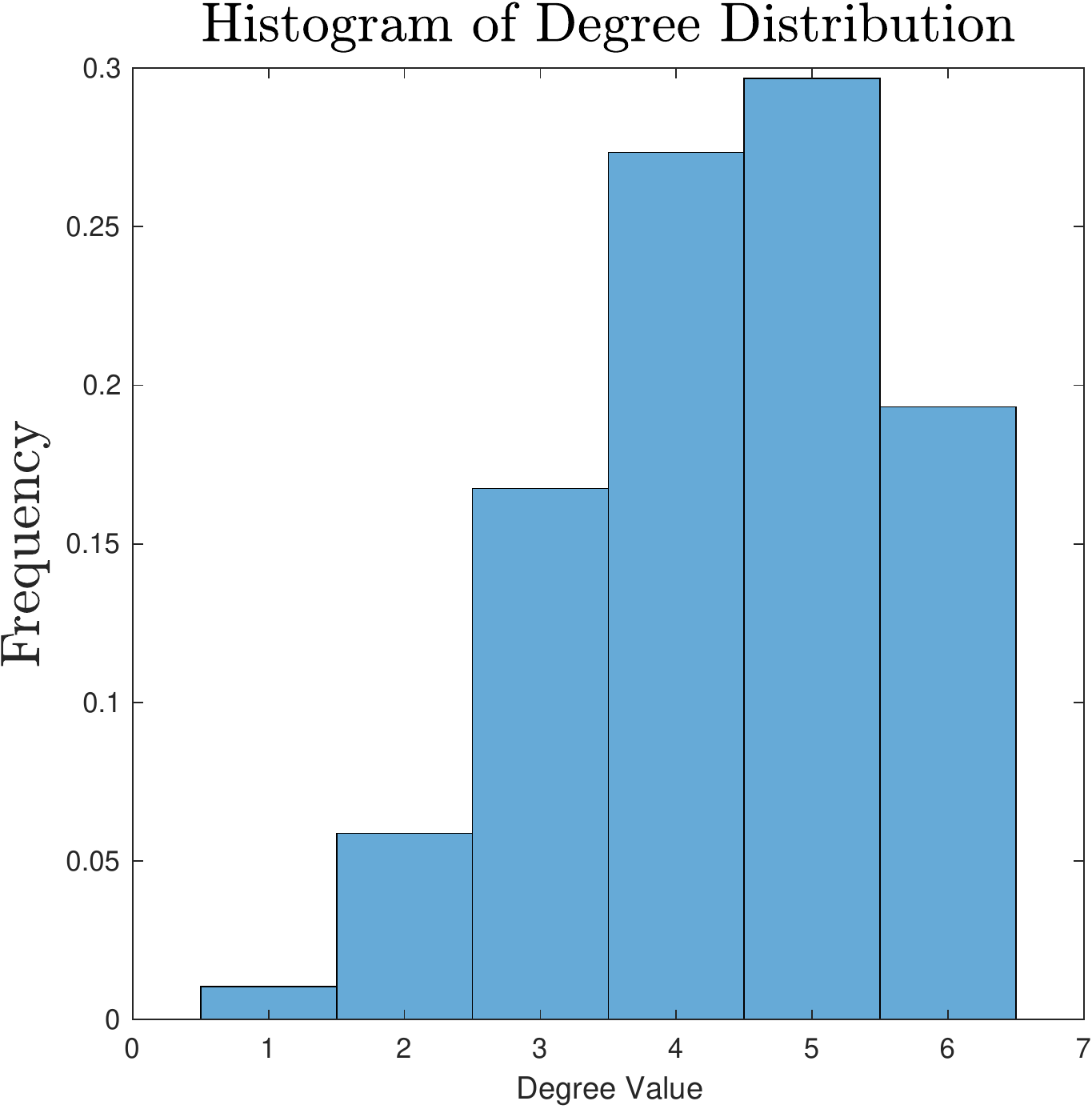}};
\end{tikzpicture}
\caption{Degree distribution for 169-node hexagonal lattice with random edge deletions.}\label{fig:delete-degree}
\end{figure}

\section{Supplemental Plots to Section 5.2}

The following tables present the pairwise comparisons for $A,P,L,M,$ and $M^2=M^\T M$ on the Black and Hispanic population by tract for all 50 states. (We include $M^2$ because Theorem 7.1 % \ref{thm:randomwalks}
implies that $\I(\V \ ; M^2)$ equals the two-step variance reduction in the uniformizing diffusion.)

\begin{figure}[ht]
    \centering
    \begin{tikzpicture}[scale=0.72]
\node at (4, -22.5) {$\I (Hisp, A)$};
\node at (8, -22.5) {$\I (Hisp, P)$};
\node at (12, -22.5) {$\I (Hisp, M)$};
\node at (16, -22.5) {$\I (Hisp, M^2)$};
\node at (20, -22.5) {$\I (Hisp, L)$};
\node[rotate=90] at (1.5, -4) {$\I (Hisp, A)$};
\node[rotate=90] at (1.5, -8) {$\I (Hisp, P)$};
\node[rotate=90] at (1.5, -12) {$\I (Hisp, M)$};
\node[rotate=90] at (1.5, -16) {$\I (Hisp, M^2)$};
\node[rotate=90] at (1.5, -20) {$\I (Hisp, L)$};
\foreach \Y [count=\i] in {A, P, M, M2, L}
{
\foreach \X [count=\j] in {A, P, M, M2, L}
{
\ifthenelse{\j > \i \OR \i > \j}{
\node at (4*\j,-4*\i) {\includegraphics[width=3cm]{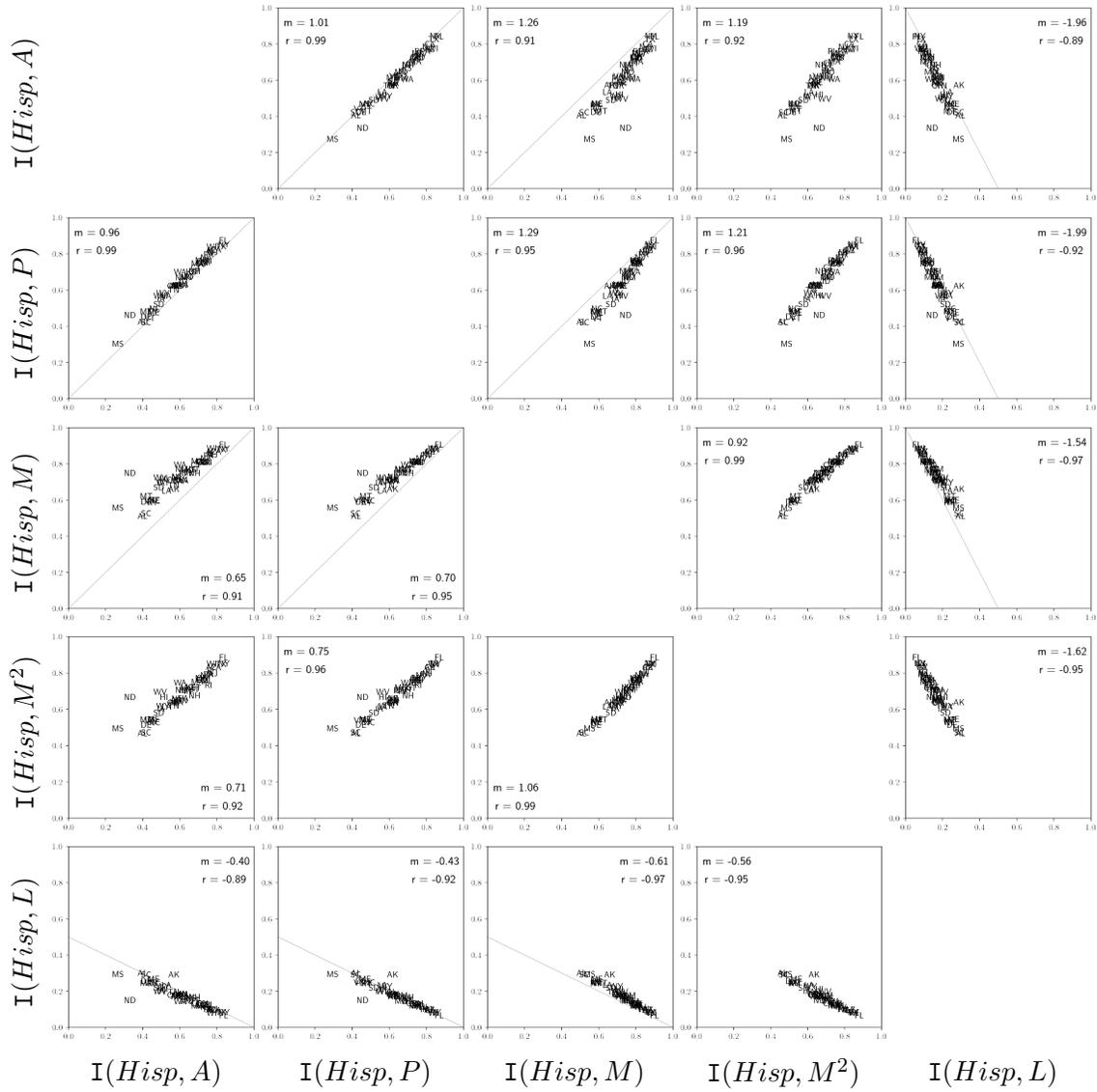}};
}{}
}
}
\end{tikzpicture}
\caption{Comparing \I with various spatial weight matrices for 50-state Census tract data for Hispanic population shares. The slope $m$ of the best fit line and the correlation $r$ are reported for each plot. The dashed lines indicate the relationship for regular graphs.}
\end{figure}

\begin{figure}[ht]
    \centering
\begin{tikzpicture}[scale=0.72]
\node at (4, -22.5) {$\I (Black, A)$};
\node at (8, -22.5) {$\I (Black, P)$};
\node at (12, -22.5) {$\I (Black, M)$};
\node at (16, -22.5) {$\I (Black, M^2)$};
\node at (20, -22.5) {$\I (Black, L)$};
\node[rotate=90] at (1.5, -4) {$\I (Black, A)$};
\node[rotate=90] at (1.5, -8) {$\I (Black, P)$};
\node[rotate=90] at (1.5, -12) {$\I (Black, M)$};
\node[rotate=90] at (1.5, -16) {$\I (Black, M^2)$};
\node[rotate=90] at (1.5, -20) {$\I (Black, L)$};
\foreach \Y [count=\i] in {A, P, M, M2, L}
{
\foreach \X [count=\j] in {A, P, M, M2, L}
{
\ifthenelse{\j > \i \OR \i > \j}{
\node at (4*\j,-4*\i) {\includegraphics[width=3cm]{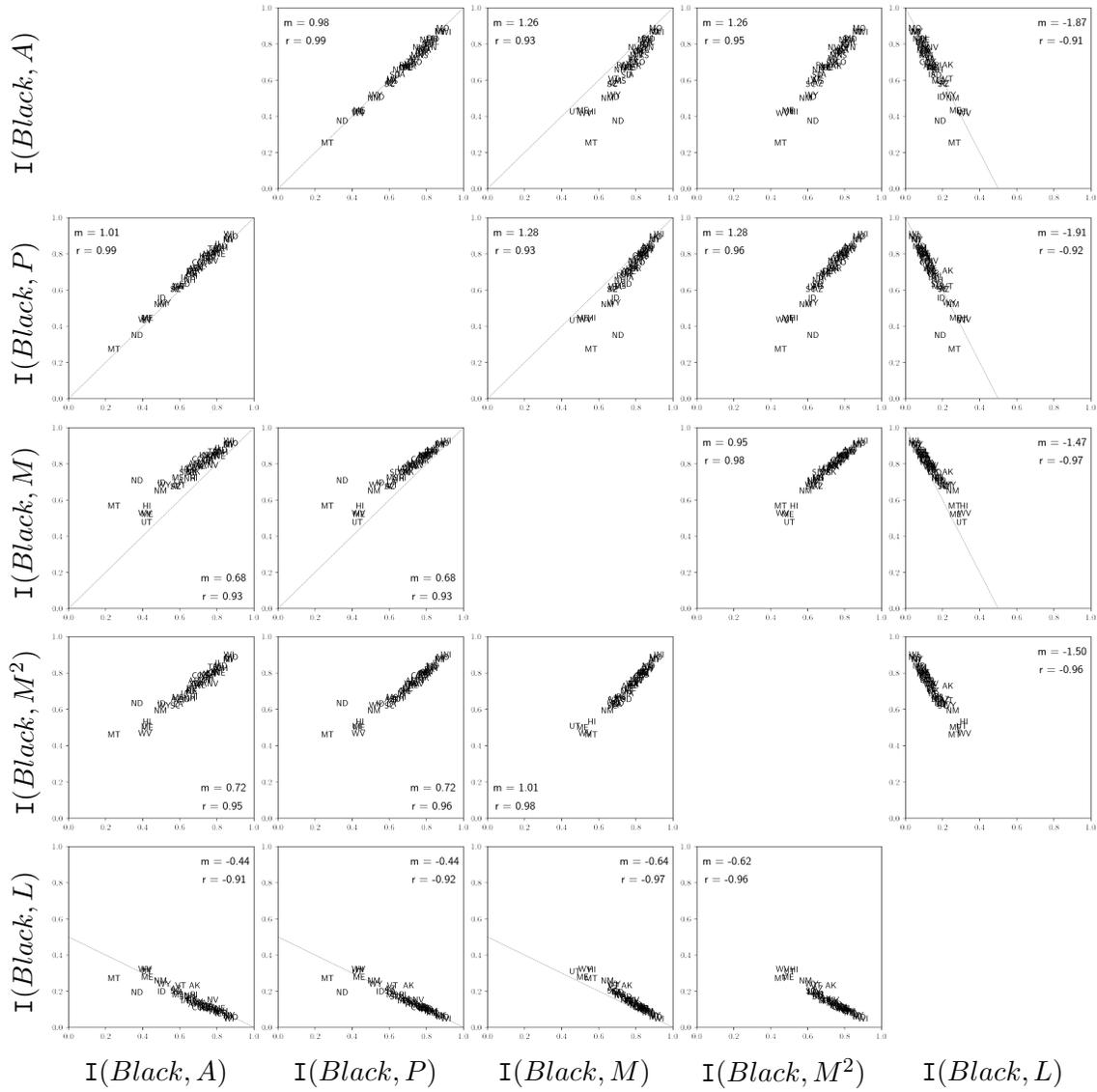}};
}{}
}
}
\end{tikzpicture}
    \caption{Comparing \I with various spatial weight matrices for 50-state Census tract data for Black population shares. The slope $m$ of the best fit line and the correlation $r$ are reported for each plot. The dashed lines indicate the relationship for regular graphs.}
    \label{fig:my_label}
\end{figure}

\end{document}